\newif\ifsubmission
\newif\ifanon
\newif\ifextendedabstract
\newif\ifstocproceedings

\documentclass[11pt]{article}
\usepackage[margin=1in]{geometry}
\usepackage[T1]{fontenc}
\usepackage{newtxtext,newtxmath}

\usepackage{amsmath}

\usepackage{amsthm}
\usepackage{mathtools}
\usepackage{braket}
\usepackage{ytableau}

\numberwithin{equation}{section}

\usepackage{enumitem}
\usepackage{float}

\usepackage{graphicx}
\usepackage{xcolor}
\usepackage{tikz}
\usetikzlibrary{calc,positioning,fit,backgrounds}

\usepackage[
  backend=biber,
  style=alphabetic,
  maxnames=99,
  minnames=5,
  maxalphanames=3,
  minalphanames=3
]{biblatex}
\DeclareDelimFormat{multicitedelim}{\addcomma\space}

\usepackage[page]{appendix}

\usepackage{hyperref}
\usepackage{zref-clever}
\zcsetup{
    abbrev=true,
    cap
}
\hypersetup{
    colorlinks=true,
    linkcolor=blue,
    filecolor=magenta,
    urlcolor=cyan,
    citecolor=blue
}

\usepackage{tcolorbox}
\usepackage{caption}
\usepackage{algorithm}
\usepackage{algpseudocode}

\usepackage[colorinlistoftodos]{todonotes}

\usepackage{aliascnt}

\usepackage{array}

\newcommand{\cref}[1]{\zcref{#1}}
\newcommand{\Cref}[1]{\zcref[S]{#1}}

\tikzset{
  dnode/.style={circle, draw=black, line width=0.4pt, minimum size=4.6mm, inner sep=0pt},
  dlab/.style={font=\scriptsize, fill=white, inner sep=0.6pt},
}

\let\oldtodo\todo
\ifsubmission
    \newcommand{\ben}[1]{}
    \newcommand{\barak}[1]{}

    \renewcommand{\todo}[1]{}
\else
    \newcommand{\ben}[1]{\oldtodo[color=blue!15]{\textbf{Ben}: #1}}
    \newcommand{\barak}[1]{\oldtodo[color=olive!15]{\textbf{Barak}: #1}}

    \renewcommand{\todo}[1]{\oldtodo{\underline{\textbf{ToDo}}:\\#1}}
\fi

\theoremstyle{plain}
\newtheorem{theorem}{Theorem}[section]

\newaliascnt{lemma}{theorem}
\newtheorem{lemma}[lemma]{Lemma}
\aliascntresetthe{lemma}

\newaliascnt{proposition}{theorem}

\aliascntresetthe{proposition}

\newaliascnt{corollary}{theorem}
\newtheorem{corollary}[corollary]{Corollary}
\aliascntresetthe{corollary}

\newaliascnt{conjecture}{theorem}

\aliascntresetthe{conjecture}

\newaliascnt{claim}{theorem}

\aliascntresetthe{claim}

\theoremstyle{definition}
\newaliascnt{definition}{theorem}
\newtheorem{definition}[definition]{Definition}
\aliascntresetthe{definition}

\newaliascnt{example}{theorem}
\newtheorem{example}[example]{Example}
\aliascntresetthe{example}

\newaliascnt{fact}{theorem}
\newtheorem{fact}[fact]{Fact}
\aliascntresetthe{fact}

\zcRefTypeSetup{fact}{
  name-sg = fact,
  name-pl = facts,
  Name-sg = Fact,
  Name-pl = Facts,
  name-sg-ab = fact,
  name-pl-ab = facts,
  Name-sg-ab = Fact,
  Name-pl-ab = Facts,
}

\theoremstyle{remark}
\newaliascnt{remark}{theorem}

\aliascntresetthe{remark}

\newaliascnt{note}{theorem}

\aliascntresetthe{note}

\newcommand{\wh}[1]{\widehat{#1}}

\newcommand{\wt}[1]{\widetilde{#1}}
\newcommand{\ft}[0]{\mathtt{FT}}
\newcommand{\pn}[0]{\mathrm{pn}}
\newcommand{\cc}[0]{\mathrm{cc}}
\newcommand{\op}[0]{\mathrm{op}}
\newcommand{\tr}[0]{\mathrm{Tr}}

\newcommand{\cont}[0]{\mathrm{cont}}
\newcommand{\poly}[0]{\mathrm{poly}}

\newcommand{\stirlingii}[2]{\genfrac{\{}{\}}{0pt}{}{#1}{#2}}
\DeclarePairedDelimiter{\abs}{\lvert}{\rvert}

\DeclareMathOperator{\polylog}{polylog}

\title{Efficient Quantum Fourier Transforms For Semisimple Algebras}

\ifanon
    \author{}
\else
    \author{
        Ben Foxman\thanks{Yale University. Email: \texttt{ben.foxman@yale.edu}}
        \and
        Barak Nehoran\thanks{Columbia University. Email: \texttt{b.nehoran@columbia.edu}}
        \and
        Yongshan Ding\thanks{Yale University. Email: \texttt{yongshan.ding@yale.edu}}
    }
\fi

\date{}

\begin{document}
\maketitle

\begin{abstract}
The quantum Fourier transform (QFT) is a fundamental primitive in quantum computation and quantum information. In this work, we extend the framework for the QFT from finite groups to finite-dimensional semisimple algebras, and give efficient QFTs for the partition algebra $P_n(d)$, Brauer algebra $B_n(d)$, and walled Brauer algebra $B_{r,s}(d)$. These algebras play important roles in generalized Schur-Weyl duality, statistical physics and many-body systems, and have recently found several applications in quantum algorithms. 

Unlike the group case, the Fourier transform over a semisimple algebra can be non-unitary. Nevertheless, we show that when the parameter $d$ is sufficiently large, the Fourier transform is well approximated by a unitary operator. Furthermore, we show that for each of the algebras $A$ from above, such an approximate Fourier transform can be implemented efficiently: we give a quantum algorithm with gate complexity $\poly(n,\log d,\log(1/\varepsilon))$ for approximating the Fourier transform to error $(d^{-1/2} + \varepsilon) \cdot \poly(|A|)$. Along the way, we establish several properties of the Fourier basis of semisimple algebras that may be of independent interest.
\end{abstract}
\tableofcontents
\newpage
\section{Introduction}
The quantum Fourier transform (QFT) is one of the most important building blocks in quantum computing, with broad applications in quantum algorithms, quantum complexity theory, and quantum information theory. Famous applications of the QFT include Simon's algorithm for hidden shifts~\cite{Simon1997}, the period-finding step of Shor's algorithm~\cite{Shor1994}, and Kitaev's phase estimation algorithm~\cite{Kitaev1995}. In the last three decades, these breakthroughs have inspired extensive work aiming to generalize the QFT~\cite{beals1997qft, moore2003genericquantumfouriertransforms, bacon2005quantumschurtransformi}, improve its efficiency~\cite{CleveWatrous2000, Nam_2020, kahanamokumeyer2025logdepthinplacequantumfourier}, and uncover new applications~\cite{Hallgren2002, Kuperberg2005, HarrowHassidimLloyd2009,Jordan_2025}.

In this paper, we build on this long line of work by giving efficient QFTs for certain semisimple algebras, which generalize finite groups. In particular, we identify a class of semisimple algebras, known as \textit{diagram algebras}~\cite{kauffman1988new,Koenig2008Panorama}, for which an $N$-dimensional Fourier transform can be implemented by a quantum algorithm using only $\polylog(N)$ quantum gates. As in the setting of finite groups, our algorithm is exponentially faster than the $O(N \polylog N)$ complexity of the best known classical algorithm~\cite{maslen2016efficientcomputationfouriertransforms}. 
\textit{To the best of our knowledge, ours is the first efficient QFT algorithm for non-group algebras.}

Our algorithm is based on the well-known \textit{separation of variables} approach~\cite{CooleyTukey1965,DiaconisRockmore1990,moore2003genericquantumfouriertransforms}, which has been highly successful in the construction of efficient classical and quantum Fourier transforms for finite groups. In the classical setting, Maslen, Rockmore, and Wolff~\cite{maslen2016efficientcomputationfouriertransforms} used this framework to obtain a classical Fourier transform for the Brauer algebra, one of the diagram algebras studied in this paper. Bringing this framework into the quantum setting, however, requires several new ideas and modifications to the underlying algorithm. The main obstacle is that, by contrast to the group algebra case, the Fourier transform of a semisimple algebra may be non-unitary. As a consequence, key steps of the algorithm can no longer be implemented directly on a quantum computer. We discuss these challenges, and how we address them, in more detail in~\cref{sec:overview}.

Just as QFTs for finite groups have been a fruitful source of advances in quantum algorithms and complexity theory~\cite{EttingerHoyerKnill1998, FriedlEtAl2002, fefferman2015powerquantumfouriersampling, bravyi2024quantumcomplexitykroneckercoefficients, Bostanci_2025}, we believe that QFTs for semisimple algebras will also give rise to a number of interesting new directions. Along these lines, we discuss several potential applications in~\cref{sec:applications}. 

Before elaborating on our main results, we give a brief overview of semisimple algebras (\cref{sec:intro_semisimple}) and the generalized Fourier transform (\cref{sec:intro_ft}), since we expect our results to be relevant to a variety of areas of quantum information and computer science.
In~\cref{sec:rep_theory_background}, we give a more detailed background tailored to the results in this paper. 
\subsection{Semisimple Algebras}
\label{sec:intro_semisimple}
In this paper, we study quantum algorithms for semisimple algebras. An \textit{algebra} is a vector space over a field $\mathbb{F}$, equipped with a bilinear multiplication operator. In this work, we consider $\mathbb{C}$-algebras, where $\mathbb{F}$ is the field of complex numbers. Typical examples of such algebras include: $\mathbb{C}^n$ with coordinate-wise addition and multiplication, the set of all complex polynomials $\mathbb{C}[x]$, and the \textit{group algebra} $\mathbb{C}[G]$, obtained by taking the $\mathbb{C}$-vector space with basis $\{\ket{g}\}_{g \in G}$ and using the group operation for multiplication. 

The central example of such an algebra in the context of representation theory (or module theory) is the \textit{matrix algebra} $M_{d}(\mathbb{C})$, the vector space of $d \times  d$ complex matrices with multiplication given by the usual matrix product. If an algebra is isomorphic to a matrix algebra, then it is a \textit{simple algebra}. An algebra is \textit{semisimple} if it is isomorphic to a direct sum of simple algebras (see~\cref{def:semisimple_algebra}).\footnote{A common equivalent definition in the literature is based on having a trivial Jacobson radical.} Semisimple algebras are a broad class of algebras for which a Fourier transform can be defined, and include the group algebras of all finite groups as a strict subset.

This paper focuses on a family of semisimple algebras known as \textit{diagram algebras}~\cite{kauffman1988new,Koenig2008Panorama}. Diagram algebras are generalizations of the symmetric group algebra \(\mathbb{C}[S_n]\), although they are typically not group algebras. To illustrate the generalization, we can represent permutations (the basis elements of $\mathbb{C}[S_n]$) using \textit{permutation diagrams}. Permutation diagrams are \(n\times 2\) grids whose bottom \(n\) vertices are paired with the top \(n\) according to the permutation, with multiplication given by stacking one diagram on top of another, and merging and removing the vertices in the internal rows: 

\begin{figure}[H]
    \centering
    \includegraphics[width=\linewidth]{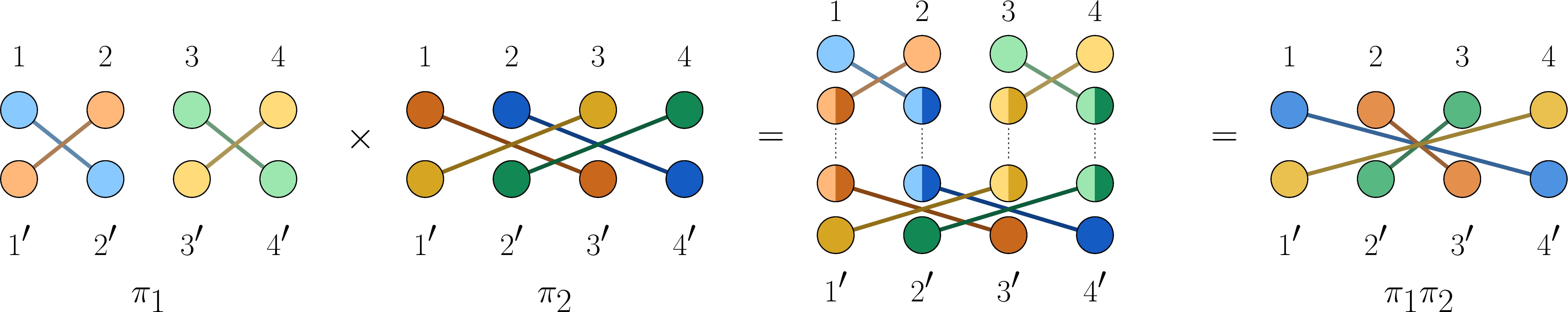}
    \caption{Diagram multiplication in the symmetric group algebra.}
    \label{fig:symmetric_group_multiplication}
\end{figure}
\noindent Diagram algebras are obtained by enlarging the class of allowed diagrams and extending this composition rule accordingly, so that connected components are merged when diagrams are concatenated (see~\cref{sec:diagram_algebras} for a rigorous definition). 
Among the most important of these diagram algebras, the \textit{partition algebra}~\cite{Martin1996StructurePartitionAlgebras}, is a diagram algebra whose basis includes all set partitions of the \(2n\) vertices (see~\Cref{fig:examples}, left, for an example of such a partition diagram).
The most common diagram algebras arise as either subalgebras or quotients of the partition algebra.
Examples include the \textit{Brauer algebra}~\cite{Brauer1937AlgebrasSemisimpleContinuousGroups}, which allows any perfect matching among the $2n$ vertices, or equivalently partitions where each component has size 2 (that is, a component may contain two bottom vertices, or two top vertices, or a top vertex and bottom vertex; see \Cref{fig:examples}, middle), and the \textit{walled Brauer algebra}~\cite{Koike1989TensorProducts, Turaev1989OperatorInvariantsTangles}, which imposes further restrictions on the Brauer algebra (see~\Cref{fig:examples}, right).
\begin{figure}[H]
    \centering
    \includegraphics[width=0.9\linewidth]{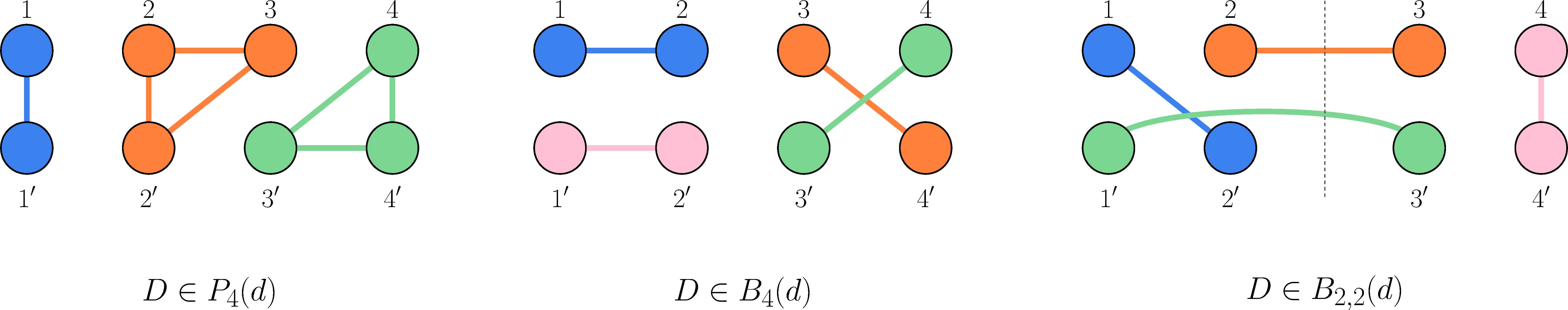}
    \caption{Example elements of the partition algebra $P_n(d)$, Brauer algebra $B_n(d)$, and the walled Brauer algebra $B_{r,s}(d)$. See~\cref{sec:diagram_algebras} for formal definitions of each of these algebras.}
    \label{fig:examples}
\end{figure}
\noindent
Each of these diagram algebras has a variety of applications across mathematical physics, many-body theory, and quantum information. The partition algebra was introduced in statistical mechanics, where it appears in the study of the Potts model, a generalization of the well-known Ising model~\cite{MartinSaleur1994,Martin2000PartitionPotts}, and has been used to analyze braiding constructions of entangling gates~\cite{PadmanabhanSuginoTrancanelli2020}. The Fourier transform on the partition algebra is also relevant in certain many-body settings involving subspaces invariant under local permutations~\cite{Barnes_2022}. The Brauer algebra appears in several quantum settings, including quantum field theory, gauge theory, and matrix models~\cite{KimuraRamgoolam2007,CanduSaleur2009,KimuraRamgoolamTurton2010}, as well as the theory of symmetric quantum states~\cite{Allerstorfer_2025}.

Besides the symmetric group algebra, the walled Brauer algebra is perhaps the most directly connected to recent advances in quantum algorithms. Representations of the walled Brauer algebra have been used to give efficient circuits for the mixed Schur transform~\cite{fei2023efficientquantumalgorithmportbased, nguyen2023mixedschurtransformefficient, grinko2023gelfandtsetlinbasispartiallytransposed}, port-based quantum teleportation~\cite{fei2023efficientquantumalgorithmportbased, grinko2024efficientquantumcircuitsportbased}, and to analyze unitarily equivariant channels~\cite{nguyen2023mixedschurtransformefficient, grinko2024linear}. The walled Brauer algebra also arises in several problems in statistical mechanics~\cite{Huber_2022, PhysRevX.13.041028}.

An important place where diagram algebras appear is in the context of tensor powers of group representations.%
\footnote{
A representation of a group assigns to each group element a (unitary) matrix in a way that preserves the group multiplication.
}
For instance, for any unitary $U \in U(d)$, its $n$-fold tensor power representation takes $n$ copies, $U^{\otimes n}$.
Its commutant, which consists of all operators that commute with $U^{\otimes n}$ for all unitaries $U$, is spanned by permutations among the $n$ registers, and is a representation of the symmetric group algebra.
Other diagram algebras arise from the commutants of other group tensor power representations: for instance the Brauer algebra arises from the tensor power of the orthogonal group, and the partition algebra appears in the commutant of the tensor power of the symmetric group (represented as permutation matrices).
These commutant relationships are summarized in \Cref{tab:schur-weyl-partition}, and form the basis of several important instances of Schur-Weyl duality, which we will return to later (see \Cref{sec:schur-transform-application}).

\begin{table}[H]
\centering
\setlength{\tabcolsep}{6pt}
\renewcommand{\arraystretch}{1.2}
\scalebox{0.96}{
\begin{tabular}{
  >{\raggedright\arraybackslash}m{1.7cm}
  >{\raggedright\arraybackslash}m{1.0cm}
  >{\centering\arraybackslash}m{2.8cm}
  >{\centering\arraybackslash}m{3.5cm}
  >{\raggedright\arraybackslash}m{4.3cm}
  >{\raggedright\arraybackslash}m{1.2cm}
}
\hline
\multicolumn{2}{c}{
    \textbf{Group}
}
& 
\textbf{Tensor space} 
& 
\textbf{Representation} 
& 
\multicolumn{2}{c}{
    \textbf{Commutant algebra} 
}
\\
\hline
Unitary
& 
$U(d)$ 
& 
$V_d^{\otimes n}$ 
& 
$U^{\otimes n}$
& 
Symmetric Group Algebra
& 
$\mathbb{C}[S_n]$ 
\\
\noalign{\vspace{6pt}}

Unitary
& 
$U(d)$ 
& 
$V_d^{\otimes r} \otimes (V_d^*)^{\otimes s}$ 
& 
\shortstack[c]{
$U^{\otimes r} \otimes \overline{U}\vphantom{U}^{\otimes s}$
\\
{\tiny $r$ copies of the unitary tensored with}
\\
{\tiny $s$ copies of its complex conjugate}
}
& 
Walled Brauer Algebra
& 
$B_{r,s}(d)$ 
\\
\noalign{\vspace{6pt}}

Orthogonal
& 
$O(d)$ 
& 
$V_d^{\otimes n}$ 
&
\shortstack[c]{
$U_{g}^{\otimes n}$
\\
{\tiny $U_g$ is the standard embedding}
\\
{\tiny of $g \in O(d)$ as a unitary over $\mathbb{C}^n$}
}
& 
Brauer Algebra
& 
$B_n(d)$ 
\\
\noalign{\vspace{6pt}}

Symmetric
& 
$S_d$ 
& 
$(V_{d-1}\oplus V_1)^{\otimes n}$ 
& 
\shortstack[c]{
$P_{\sigma}^{\otimes n}$
\\
{\tiny $P_{\sigma}$ is the permutation matrix} 
\\
{\tiny of $\sigma \in S_d$}
}
& 
Partition Algebra
& 
$P_n(d)$ 
\\
\noalign{\vspace{6pt}}

Symmetric
& 
$S_{d-1}$ 
& 
$(V_{d-2}\oplus V_1 \oplus V_1)^{\otimes n}$ 
& 
\shortstack[c]{
$P_{\sigma}^{\otimes n}$
\\
{\tiny same as above, but with one} 
\\
{\tiny of the dimensions fixed}
}
& 
Half Partition Algebra
& 
$P_{n+\frac12}(d)$ 
\\
\noalign{\vspace{6pt}}

\hline
\end{tabular}
}
\caption{Schur-Weyl dualities between group actions on tensor spaces and their commutant algebras.}
\label{tab:schur-weyl-partition}
\end{table}

\subsection{The Fourier Transform}
\label{sec:intro_ft}
The Fourier transform for a group $G$, denoted $\ft_G$, is a basis change on the group algebra $\mathbb{C}[G]$ which maximally block-diagonalizes the natural $G$-action. For example, when $G = \mathbb{Z}_N$ (as in Shor's algorithm), the Fourier basis states given by 
\begin{equation}
    \left\{\ket{\wt{k}} = \frac{1}{\sqrt{N}}\sum_{j=0}^{N-1} e^{2 \pi i jk/N} \ket{j} \right\}_{k = 0}^{N - 1}
\end{equation}
are eigenstates of the translation action $U_x\ket{j} =  \ket{x + j}$, for all $x \in \mathbb{Z}_N$. Because $\mathbb{Z}_N$ is abelian, these translation operators pairwise commute, and the Fourier basis is their simultaneous eigenbasis.

When $G$ is non-abelian, $U_g$ may not commute with $U_h$ for $g \ne h$, so no simultaneous eigenbasis exists. In this case, the Fourier basis states can be grouped into bases for minimal $G$-invariant subspaces of $\mathbb{C}[G]$, known as irreducible subspaces. Within each irreducible subspace, the $G$-action is given by an \textit{irreducible representation} (or, \textit{irrep}) $\rho$: 
\begin{figure}[H]
    \centering
    \includegraphics[width=0.6\linewidth]{figures/group_ft_block_diagonalization.jpg}
    \caption{The Fourier transform block-diagonalizes the natural $G$-action, known as the regular representation of $G$. The $\ket{\rho}$ label specifies an irreducible representation, the $\ket{i}$ register denotes a particular copy of $\rho$, and the $\ket{j}$ register is a basis vector within the irreducible subspace indexed by $(\rho, i)$.}
    \label{fig:diagonalization_ft}
\end{figure}
For certain groups, such as the symmetric group $S_n$, the structure of the irreducible representations is well understood and has many deep connections in algebra, combinatorics, and representation theory~\cite{JamesKerber1984, Sagan2001, Kleshchev2014}. The semisimple algebras considered in this paper may be viewed as natural generalizations of the symmetric group algebra (\cref{sec:diagram_algebras}), and many (though not all) properties of their irreducible representations extend in a corresponding way.

For a general semisimple algebra $A$, the Fourier transform is defined analogously, and we give a more formal treatment in~\cref{sec:rep_theory_background}. Unlike a finite group, where $U_g$ is unitary\footnote{Since $h \mapsto gh$ is a bijection for any group $G$, $U_g$ is simply a permutation on the group basis states.}, the action of $U_a$ on $A$ may not be unitary. For instance, some may be projectors, while others may be more general operators. Moreover, $\ft_A$ may itself fail to be unitary, which initially appears to pose a fundamental obstacle for implementing the corresponding QFT. To get around this obstacle we implement an \textit{approximate} QFT instead of an exact one, and give precise conditions for when such an approximate QFT exists. For more details, see~\cref{sec:qft_exposition}.   

As one of the most prominent algorithms of the 20th century, the Fourier transform has a long history of algorithmic instantiations, beginning with classical algorithms. For any $N$-dimensional semisimple algebra $A$, the Fourier transform can be computed classically in $O(N^2)$ time with standard matrix methods. The famous \textit{fast Fourier transform}~\cite{CooleyTukey1965} reduces the complexity to $O(N \log N)$ when $A$ is a cyclic group algebra. Algorithms with complexity $O(N \polylog N)$ were found for certain non-abelian groups by Beth~\cite{Beth1987GeneralizedFourier} and separately by Diaconis and Rockmore~\cite{DiaconisRockmore1990}, who developed the separation of variables approach for the non-abelian Fourier transform. Since then, several other works have built upon these fundamental results~\cite{ClausenBaum1993FFTsn, MaslenRockmore1997SOVI, umans2019fastgeneralizeddftsfinite}.

On the quantum side, an efficient QFT for cyclic group algebras has been known since the mid-1990s~\cite{Coppersmith1994ApproximateFourier, Shor1994}.\ Since an \(N\)-dimensional quantum state can be represented using \(\log N\) qubits, the quantum Fourier transform is considered efficient if it can be implemented using \(\poly(\log N)\) gates. The first efficient QFT for a non-abelian group was given by Beals for the symmetric group $S_n$~\cite{beals1997qft}, by ``quantizing'' the framework used in~\cite{DiaconisRockmore1990}. Generalizing this framework to other non-abelian groups was subsequently done in~\cite{pueschel1998fastquantumfouriertransforms, moore2003genericquantumfouriertransforms}. 

To our knowledge, the only $O(N\polylog N)$-time classical Fourier transform for semisimple algebras is the work of Maslen, Rockmore, and Wolff~\cite{maslen2016efficientcomputationfouriertransforms}. In that work, the authors similarly consider diagram algebras, namely the Brauer algebra, Birman-Murakami-Wenzl algebra, and Temperley-Lieb algebra. While we do not consider the latter two algebras in this paper, we believe our techniques could be adapted to derive efficient QFT's for these algebras as well. 

\subsection{Our Contributions}
\label{sec:overview}
The main goal of this work is to efficiently implement the QFT for various diagram algebras. We summarize the main result in~\cref{thm:main_intro}, but first explain some of the challenges that arise along the way. 

As an operator, the QFT maps the standard basis%
\footnote{For diagram algebras, the natural basis, $\mathcal{B}(A)$, of the algebra is simply the set of all valid diagrams, much like how the standard basis for a group algebra is the set of all group elements. For the partition algebra $P_n(d)$, technical reasons will require us to rescale the diagrams in order to implement the Fourier transform. For more details, see~\cref{sec:diagram_algebra_conventions}.} 
$\{\ket{a}\}_{a \in \mathcal{B}(A)}$ to the Fourier basis% 
\footnote{
    We write the Fourier basis states here as $\ket{E_{ij}^\rho}$ here for simplicity. In fact, they will be encoded as the tuple $\ket{\rho, i,j}$, but we use these interchangeably in the introduction.
}
$\{\ket{E_{ij}^\rho}\}_{\rho, i, j}$:
\begin{equation}
    \ft_A\ket{a} = \sum_{\rho \in \wh{A}} \; \sum_{i, j = 1}^{d_\rho} \braket{{E_{ij}^\rho}, a}_0 \ket{E_{ij}^\rho}
\end{equation}
Here, $\wh{A}$ denotes a complete set of irreducible representations for $A$, and $\braket{\cdot, \cdot}_0$ denotes the usual inner product on the Hilbert space $\mathbb{C}^{\mathcal{B}(A)}$ (to distinguish it from other bilinear forms we will use later). 

Several obstacles arise when trying to implement $\ft_A$ on a quantum computer. Some of these obstacles are common in the QFT literature, while others are specific to the case of semisimple algebras and require new techniques to overcome.
\begin{enumerate}
    \item We must take care to give proper normalization to the Fourier basis. When we encode the elements of the algebra as quantum states, we can encode the algebra's standard basis as the standard basis of the quantum computer's Hilbert space. However, in such an encoding, the natural definition of the algebra's Fourier basis may not correspond to normalized quantum states, i.e. $\abs{\braket{E_{ij}^\rho|E_{ij}^\rho}}^2 \ne 1$. So $\ft_A$ cannot be unitary. Of course, this arises even when $A$ is a group algebra, and can be handled by giving the Fourier basis of the abstract algebra a normalization that depends on the encoding. 
    \item Even once we properly normalize them, the quantum states corresponding to the algebra's Fourier basis states are not pairwise orthogonal, i.e. $\braket{E_{ij}^\rho|E_{kl}^\sigma} \ne 0$. This is a more fundamental issue, and rules out even inefficient implementations of the QFT for a generic semisimple algebra. Nevertheless, we show that the diagram algebras considered in~\cref{sec:intro_semisimple} have Fourier states which are \textit{approximately orthogonal}, in the limit of large $d$:
    \begin{theorem}[Corollary~\ref{cor:d_nice_algebras}, Informal]
        \label{thm:intro_approx_orthogonality}
        Let $A \in \{P_n(d), B_n(d), B_{r,s}(d)\}$. For distinct Fourier basis states $E_{ij}^\rho, E_{kl}^\sigma \in A$, 
        \begin{equation}
            \abs{\braket{E_{ij}^\rho|E_{kl}^\sigma}} \le \poly(|A|) \cdot d^{-1/2}
        \end{equation}
    \end{theorem}
    \cref{thm:intro_approx_orthogonality} becomes meaningful when $d \gg \poly(|A|)$. Since $|A| = 2^{O(n \log n)}$ for the algebras we consider, we typically assume $d$ to be a sufficiently large exponential in $n$. While this is somewhat restrictive, the $\poly(|A|)$ factor is tight (in particular, it is at least $\Omega(|A|)$), so one could not hope to substantially improve the bound without modifying the definition of the QFT.
    \item Even when an approximate QFT exists, the standard approach for implementing the QFT on a group algebra does not extend cleanly to semisimple algebras. This standard approach is known as the \textit{separation of variables} approach~\cite{CooleyTukey1965,DiaconisRockmore1990,moore2003genericquantumfouriertransforms}, and allows one to reduce the QFT for a group $G$ to a QFT for a subgroup $H$. We outline the full approach at the beginning of~\cref{sec:efficient_qft_for_algebras}, and now briefly summarize a few of the key modifications: 
    \begin{enumerate}
        \item When reducing the QFT for a group $G$ to the QFT for a subgroup $H$~\cite{beals1997qft, moore2003genericquantumfouriertransforms}, a group element $g \in G$ is factored into $g = w\,h$, with a left transversal $w$ and an element $h \in H$. For example, when $G = S_n$, and $H = S_{n-1}$, the transversal $w$ is a transposition of the form $(i\;n)$. In general, diagram algebras do not admit such a factorization. Instead, for a basis diagram $a$ in the superalgebra $A$, we need both a left transversal $w_1$ and a right transversal $w_2$, such that $a = w_1 \, b \, w_2$ for a basis diagram $b$ in the subalgebra $B$. We give such a factorization for each of the diagram algebras we consider in~\Cref{sec:factor}.
        \item
        Moreover, this factorization is inherently not unique: there will exist multiple choices of $(w_1, w_2)$ that lead to a valid factorization (for instance, two different transpositions may act the same way on a basis diagram). This poses challenges when adapting Beals' algorithm~\cite{beals1997qft}, which implicitly assumes the existence of only a single factorization.%
        \footnote{
            Beals' algorithm has a step that guesses a transversal, and determines if it is the correct one by checking if removing the transversal brings us into the subalgebra.
            When there are multiple possible valid factorizations, the algorithm may trigger on the wrong transversal.
        }
        We therefore introduce the notion of a \textit{last possible factorization} (\cref{def:ordered_factor}): We give an ordering on the valid factorizations, and when factoring a diagram, we always choose the last possible valid factorization with respect to this ordering. We choose the ordering carefully such that all steps of the algorithm (see \cref{sec:efficient_qft_for_algebras}) can be implemented efficiently.  
        \item When promoting a Fourier state of a subalgebra to a Fourier state of the full algebra, we need to compute the ratio of the norms of two Fourier states, i.e.\ $||E_{ij}^\rho||_2/||E_{kl}^\sigma||_2$. To efficiently compute this ratio, we prove that the Fourier basis states have norms which are approximately proportional to their \textit{Schur multiplicities}\footnote{By Schur-Weyl duality, these are equivalent to the irrep dimensions for various groups in~\cref{tab:schur-weyl-partition}.}(\cref{cor:norm_of_fourier_states}), which are derived from the Schur representation (see \cref{sec:schur_rep}) of the associated diagram algebra. Given the significance of Schur--Weyl duality for these algebras (\cref{sec:schur-transform-application}), we find this suggestive of a path towards promoting implementations of Fourier transforms for semisimple algebras into their corresponding Schur transforms, possibly similar to the method used in~\cite{Krovi_2019, burchardt2025highdimensionalquantumschurtransforms} for the symmetric group.
        \item When postprocessing the recursive call, we need to apply controlled applications of the representation $\rho$ to both the $\ket{i}$ and $\ket{j}$ registers of the Fourier state. For groups, these irrep matrices are always unitary.\footnote{In the orthogonal form (see~\cref{sec:main_body_orthogonal_form}).} However, for semisimple algebras, the irrep matrices can be highly non-unitary. For example, irrep matrices of the partition algebra include non-trivial orthogonal projectors, and even matrices that are all zero, which cannot be even approximately implemented by unitary operators. To avoid this issue, we develop several alternate embeddings (\cref{sec:apply_irreps}), which combine the embedding step with the application of a non-invertible algebra element. We show that each of these alternate embeddings can be implemented efficiently and unitarily. 
        \item In addition to bounding the error at each step of the algorithm, we must also bound the error of the ``ideal'' unitary algorithm from the non-unitary true Fourier transform. The coefficients of the Fourier basis states can be expressed in terms of the irreps applied to the \textit{dual basis} of the algebra (see~\cref{def:dual_basis} for a discussion of the dual basis). We show that approximate orthogonality of the Fourier basis allows us to approximately replace terms involving the dual basis elements with corresponding terms written in the standard basis. This substitution also simplifies the analysis of the algorithm (see~\cref{sec:efficient_qft_for_algebras} for more details).
    \end{enumerate}
\end{enumerate}
\vspace{1.7cm}
After addressing all of these obstacles, we arrive at the main result of the paper:
\begin{theorem}[Theorem \ref{thm:qft_is_efficient_to_implement}, Informal]
    \label{thm:main_intro}
    Let $A \in \{B_n(d), B_{r,s}(d), P_n(d)\}$. The Fourier transform $\ft_A$ can be implemented on a quantum computer using 
    \begin{equation}
        \poly(n, \log d, \log(1/\varepsilon))
    \end{equation}
    gates, up to operator norm error 
    \begin{equation}
        O\!\left(\poly(|A|)\cdot \bigl(d^{-1/2} + \varepsilon\bigr)\right)
    \end{equation}
\end{theorem}
\noindent In particular, the dependence of the gate count on the number of qubits is $\wt{O}(n^{8})$ for the Brauer and walled Brauer algebras, and $\wt{O}(n^{8.5})$ for the partition algebra. We believe that these bounds could be further improved, leaving the optimization for future work.

Finally, as part of our analysis, we prove that a certain invariant is (approximately) preserved by all of the Fourier transforms above. We choose to highlight this result since it may be of broader interest.\footnote{For example, the invariant likely holds for the Schur transform as well.}
\begin{theorem}[Theorem~\ref{thm:fourier_concentration_on_propagating_number}, Informal]
    \label{thm:intro_invariant}
    Given a diagram $D$, define the propagating number $\pn(D)$ to be the number of connected components of $D$ which intersect both the top and bottom rows. For any $A \in \{P_n(d), B_n(d), B_{r,s}(d)\}$, $\ft_A\ket{D}$ is almost entirely concentrated on Fourier states $\ket{\rho, i, j}$ where $\rho$ is a $\pn(D)$-box Young diagram. 
\end{theorem}
Of course, \cref{thm:intro_invariant} holds exactly for the symmetric group algebra, which is spanned by diagrams with propagating number $n$, and has irreps indexed by $n$-box Young diagrams (see~\cref{sec:young_diagrams}). To prove~\cref{thm:intro_invariant}, we combine the approximate orthogonality of Fourier states with established tools in the representation theory of diagram algebras~\cite{halverson2004partitionalgebras}.

\subsection{Applications and Further Directions}
\label{sec:applications}
In this section, we discuss several potential applications of our efficient QFT over semisimple algebras, as well as a number of follow-up directions. \Cref{sec:schur-transform-application,sec:hsp_algebras,sec:computing_multiplicities} describe possible applications, while \cref{sec:future_directions} outlines other questions to explore in future work.
\subsubsection{Generalized Efficient Quantum Schur Transforms}
\label{sec:schur-transform-application}
The standard quantum Fourier transform considered in this paper decomposes the regular representation of an algebra into a direct sum of irreducible representations. This Fourier transform for the regular representation can be generalized to one which similarly block-diagonalizes an arbitrary representation $\mathbf{R}$, as the basis change which carries out the following isomorphism:  
\begin{equation}
    \label{eq:general_fourier_isom}
    \mathbf{R}(a) \cong \bigoplus_{\rho \in \wh{A}} \rho(a) \otimes I_{m_\rho}
\end{equation}
In~\cref{eq:general_fourier_isom}, $m_\rho$ denotes the \textit{multiplicity} of the irrep $\rho$, or number of copies of $\rho$ appearing in the decomposition of $\mathbf{R}$. 

\paragraph{Schur-Weyl Duality.}
An example of a representation besides the regular representation is the \textit{permutation representation} of $\mathbb{C}[S_n]$, which permutes $n$ distinct $d$-dimensional registers according to the permutation $\pi$:
\begin{equation}
    \mathbf{R}_{\text{perm}}(\pi)\ket{x_1, x_2, \dots, x_n} = \ket{x_{\pi^{-1
    }(1)}, x_{\pi^{-1}(2)}, \dots x_{\pi^{-1}(n)}}
\end{equation}
Unlike the regular representation of $\mathbb{C}[S_n]$, which maps $a$ to a matrix acting on $\mathbb{C}^{|S_n|}$, the images of $\mathbf{R}_{\text{perm}}$ are matrices acting on $(\mathbb{C}^{d})^{\otimes n}$. Another representation which acts on $(\mathbb{C}^{d})^{\otimes n}$ is the $n$-fold tensor power representation of $U(d)$, which acts as $U^{\otimes n}$:
\begin{equation}
    \mathbf{R}_{\text{tp}}(U)\ket{x_1, x_2, \dots, x_n} = U\ket{x_1} \otimes U\ket{x_2} \otimes \dots\otimes U\ket{x_n}
\end{equation}
In fact, the permutation representation and tensor power representation are \textit{mutual commutants}.%
\footnote{
    The commutant of a set of operators $A$ on a Hilbert space is the set of operators that commute with every operator in $A$. We say that two sets of operators (in this case the two representations of the two algebras) are mutual commutants if each is the commutant of the other.
}

Besides the unitary and symmetric group algebras, there are other such commutant pairs acting on $(\mathbb{C}^{d})^{\otimes n}$, several of which are listed in~\Cref{tab:schur-weyl-partition}. Generally (though not always), one algebra is a group algebra $\mathbb{C}[G]$, with a tensor power representation acting in parallel on the $n$ different tensored copies of $\mathbb{C}^{d}$, while the other is a diagram algebra $A$ with a representation acting across the $n$ copies (in a way that generalizes the permutation representation above).

When two such representations are mutual commutants, Schur-Weyl duality~\cite{schur1927rationalen, weyl1939classical} says that there is a correspondence between the irreps of the two algebras, and that, furthermore, the space $(\mathbb{C}^{d})^{\otimes n}$ can be decomposed as a direct sum over the paired irreps of the two algebras:
\begin{align}
    \label{eq:schur-isomorphism}
    (\mathbb{C}^{d})^{\otimes n} 
    \cong 
    \bigoplus_{\rho \in \wh{A}} 
    W_{\rho}
    \otimes 
    V_{\rho} 
\end{align}
Where $W_{\rho}$ is an irrep space of the group algebra $\mathbb{C}[G]$, and $V_{\rho}$ is an irrep space of the diagram algebra $A$.
When $G = U(d)$ and $A = \mathbb{C}[S_n]$, the isometry which carries out the isomorphism in~\Cref{eq:schur-isomorphism} is known as the \textit{Schur transform}~\cite{bacon2005quantumschurtransformi}. 
For each of the other Schur-Weyl dual pairs (see~\Cref{tab:schur-weyl-partition}), one can also define a corresponding \textit{generalized Schur transform}.

\paragraph{Generalized Schur Transforms.}
There are numerous applications of the Schur transform across quantum information and computation, including state tomography and quantum hypothesis testing~\cite{Keyl_2001, Hayashi_2001,Christandl_2006,odonnell2015quantumspectrumtesting, pelecanos2025mixedstatetomographyreduces}, quantum communication~\cite{Bartlett_2003, ChiribellaDArianoPerinottiSacchi2004ReferenceFrame}, compression of quantum states~\cite{Jozsa_1998, Yang_2016}, port-based teleportation~\cite{grinko2024efficientquantumcircuitsportbased}, and was recently used to generalize Zhandry's compressed oracle technique~\cite{zhandry2019record, grinko2025quantumsimulationrandomunitaries,foxman2026perfectlyrecordingqueries}. For more applications, we refer readers to Harrow's thesis~\cite{harrow2005applicationscoherentclassicalcommunication}.

An efficient implementation of the Schur transform was first given in the work of Bacon, Chuang, and Harrow~\cite{bacon2005quantumschurtransformi}. The ``BCH'' algorithm constructs the Schur transform as a cascade of $n - 1$ local \textit{Clebsch-Gordan} transforms on the unitary group $U(d)$. While an efficient Clebsch-Gordan transform is not known for general irreps,~\cite{bacon2005quantumschurtransformi} shows that it is efficient when one of the irreps involved is the fundamental irrep, which suffices to implement the Schur transform. 

On the other hand, the algorithm of Krovi~\cite{Krovi_2019, burchardt2025highdimensionalquantumschurtransforms} uses the representation theory of $\mathbb{C}[S_n]$. On an input $\ket{x_1x_2\dots x_n}$, Krovi's algorithm first preprocesses the input into a permutation $\ket{\pi}$ and a ``type'' $\ket{T}$, followed by applying a QFT for $S_n$ on the $\ket{\pi}$ register, obtaining superpositions of states of the form $\ket{\lambda, i, j, T}$, where $\ket{i}$ and $\ket{j}$ are basis vectors of irrep spaces for $\mathbb{C}[S_n]$. Finally, $\ket{j, T}$ is postprocessed to obtain a superposition\footnote{Krovi's original algorithm in~\cite{Krovi_2019} contained an error in this step; this was subsequently fixed by~\cite{burchardt2025highdimensionalquantumschurtransforms}.} over basis vectors of irrep spaces for $U(d)$, completing the algorithm. 

Since our paper gives an efficient QFT for all the commutant algebras in~\cref{tab:schur-weyl-partition}, it is natural to ask whether Krovi's algorithm can be extended to obtain efficient Schur transforms for these algebras. There are several reasons this seems plausible. First, the basis vectors of the irrep spaces of $S_n$ and $U(d)$ have natural generalizations, namely the \textit{Bratteli paths} defined in~\cref{sec:subalgebra_chains}. In addition, the ``F-moves'' used in the postprocessing step of~\cite{burchardt2025highdimensionalquantumschurtransforms} have analogues for the partition algebra, whose Schur-Weyl dual is $S_d$~\cite{haase1986symmetric}, and the Brauer algebra, whose Schur-Weyl dual is $O(d)$~\cite{pan1996racah}.

We expect that efficient Schur transforms for the semisimple algebras would have a wide variety of applications in quantum computation. For example, one could generalize the compressed oracle for random unitaries~\cite{ma2025constructrandomunitaries, grinko2025quantumsimulationrandomunitaries,foxman2026perfectlyrecordingqueries} to random permutations or random orthogonal matrices~\cite{foxman2026perfectlyrecordingqueries}, or generalize tomography of unitarily-invariant properties~\cite{odonnell2015quantumspectrumtesting} to permutation-invariant properties. Additionally, attempts to generalize the Schur transform to semisimple algebras may reveal new structural properties of these algebras; see the first paragraph of~\cref{sec:future_directions}.

\subsubsection{The Hidden Subalgebra Problem}
\label{sec:hsp_algebras}
The hidden subgroup problem~\cite{Kitaev1995}, often abbreviated HSP, is a framework capturing many algebraic problems in quantum computation~\cite{Jozsa_2001, HallgrenRussellTaShma2003, Childs_2010}. The general formulation of the HSP is as follows: given a group $G$ and black-box access to a function $f: G \rightarrow \{0, 1\}^{*}$, we say that $f$ \textit{hides} a subgroup $H \le G$ if 
\begin{equation}
    \label{eq:hsp}
    f(g_1) = f(g_2) \iff g_1 = g_2h,\;\; h \in H
\end{equation}
equivalently, $g_1$ and $g_2$ are in the same coset of $H$ in $G$.\footnote{Note that we do not require $H$ to be a normal subgroup. Because of this, when $H$ is non-normal, one must distinguish between applying $h$ on the left or the right (corresponding to left and right cosets).} The goal of the HSP is to determine $H$, ideally with a $\poly(\log |G|)$-time (quantum) algorithm. 

In many applications of the QFT on groups, the underlying problem can be recast as a HSP for some group. For example, this is the case in both Simon's algorithm with $G = \mathbb{Z}_2^n$ and Shor's algorithm with $G = \mathbb{Z}_k$. The relationship between the QFT and HSP comes from a typical algorithmic strategy for solving the HSP, known as \textit{coset sampling}. In coset sampling, one uses the oracle for $f$ to prepare the state 
\begin{equation}
    \frac{1}{\sqrt{H}}\sum_{h \in H} \ket{gh}
\end{equation}
for some uniformly random $g \in G$, i.e.\ a random coset of $H$. Afterwards, the QFT for $G$ is applied to the state, and one hopes to use the resulting Fourier state to extract some information about $H$. While determining $H$ can always be done with $\polylog |G|$ queries to $f$~\cite{Ettinger_2004}, quantum algorithms which also have $\polylog |G|$ time complexity are only known for certain restricted families of groups.

Since this paper introduces an efficient QFT for certain semisimple algebras, it seems natural to generalize the HSP to this setting. In attempting to generalize~\cref{eq:hsp}, one has to contend with the following obstacles:  
\begin{itemize}
    \item Unlike group algebras, the natural bases of the semisimple algebras we consider are not closed under multiplication. For example, in the partition algebra $P_n(d)$ (\cref{sec:diagram_algebras}), the point diagram $p_i$ satisfies $p_i^2 = dp_i$, which is a scalar multiple of the original diagram. To rectify this issue, one could perhaps redefine the hiding property of $f$ up to rescaling the diagrams.
    \item Unlike groups, multiplication in an algebra is not invertible. Hence, the ``cosets'' $\{aB: a \in A\}$ may not partition $A$. In this case, the hiding property may have to defined with respect to some fixed factorization (see~\cref{sec:factor}). 
\end{itemize}
Despite these caveats, we believe the HSP for diagram algebras is an interesting direction for future work. Our main justification comes from viewing diagram algebras as generalizations of the symmetric group algebra, for which the HSP has attracted substantial attention. This is primarily due to the connection between the HSP for $S_n$ and the famous \textit{graph automorphism} and \textit{graph isomorphism} problems, which are widely believed to be $\mathsf{NP}$-intermediate~\cite{Schoning1988,Babai2018GI}. In particular, no polynomial-time classical algorithm is known for either problem. 

However, an efficient algorithm for the HSP on $S_n$ is known to imply an efficient algorithm for graph automorphism, which in turn implies an efficient algorithm for graph isomorphism~\cite{Childs_2010}. To see why, note that the hiding function given by $f(\pi) \coloneqq \pi(G)$, where $\pi(G)$ is a description of a graph $G$ with vertex labels permuted according to $\pi$, hides the automorphism group of the graph. An efficient \textit{quantum} algorithm for the symmetric group HSP would therefore imply an efficient \textit{quantum} algorithm for graph automorphism. 

In addition to acting on a graph $G$ by a permutation $\pi$, one could consider the action of a diagram $D$ from a larger diagram algebra. For example, in the partition algebra $P_n(d)$, a generating set for the diagram basis consists of the swap generators $s_i$, as well as the \textit{bridge} generators $b_i$, which join all vertices in columns $i$ and $i + 1$, and the \textit{point} generators, which disconnect the two vertices in column $i$ (see~\cref{lem:partition_generators}). Just as $s_i$ acts on $G$ by swapping vertex labels $i$ and $i + 1$, $b_i$ and $p_i$ have natural actions given by \textit{vertex contractions} and \textit{vertex resets}:
\begin{figure}[H]
    \centering
    \includegraphics[width=0.9\linewidth]{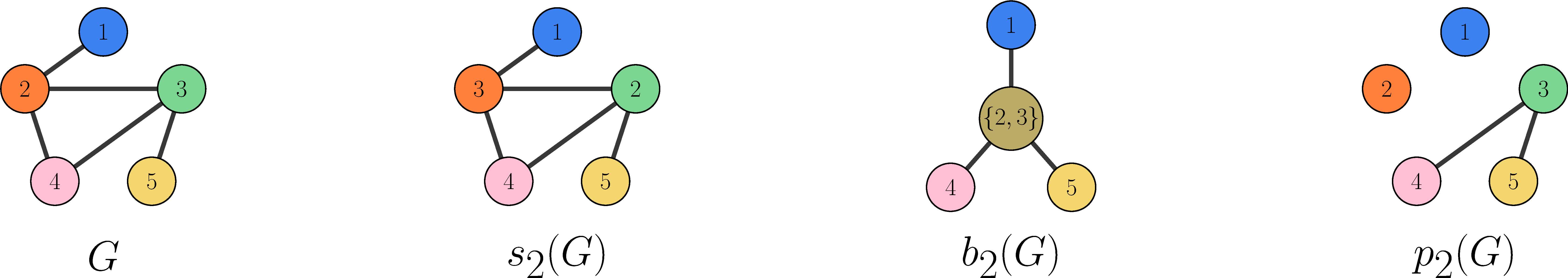}
    \caption{Actions of the partition algebra generators on a graph $G$. The generator $b_i$ contracts the vertices containing labels $i$ and $i+1$, merging their label sets. The generator $p_i$ removes the label $i$ from its current vertex and places it on a new isolated singleton vertex. If $i$ belongs to a vertex carrying multiple labels, only the label $i$ is reset; the remaining labels stay on the original vertex.}
    \label{fig:graph_actions}
\end{figure} 

One can verify that this is a valid representation of the partition algebra by checking that these actions respect all the relations between partition algebra generators, including the idempotent relations\footnote{Technically, $p_i^2 = dp_i$, but we ignore the scalar factors for the purposes of this discussion.} $p_i^2 = p_i$ and $b_i^2 = b_i$, the twisting relations $b_ip_ib_i = b_i$ and $b_ip_{i+1}b_i = b_i$, and the transport relation $s_ip_is_{i} = p_{i+1}$. For a full list of partition algebra relations, see Theorem 1.11 in~\cite{halverson2004partitionalgebras}. 
 
While graph isomorphism is not known to be $\mathsf{NP}$-complete, $\mathsf{NP}$-completeness is known for a related problem, the one of determining if two graphs are isomorphic after performing vertex contractions on one of the graphs~\cite{brouwer1987contractibility}. This problem corresponds to finding a hidden subalgebra of $P_n(d)$ generated by swap and bridge generators, extending the analogy for graph isomorphism, which asks for a hidden subalgebra of $\mathbb{C}[S_n]$. Hence, an efficient algorithm for the HSP on the partition algebra could imply a solution to an $\mathsf{NP}$-complete problem, which gives strong evidence for the hardness of the partition algebra HSP.\footnote{Since it is unlikely that $\mathsf{NP} \subseteq \mathsf{BQP}$~\cite{Aaronson2010BQPandPH}.} 
We leave the precise formulation and analysis of these connections to future work.

\subsubsection{Computing Multiplicities for Diagram Algebras}
\label{sec:computing_multiplicities}
A major line of work in representation theory aims to capture representation-theoretic quantities with ``purely combinatorial'' formulas~\cite{Stanley2000, Ikenmeyer_2017, pak2022combinatorialinterpretation, bravyi2024quantumcomplexitykroneckercoefficients, panova2025polynomialtimeclassicalversus, larocca2025quantumalgorithmsrepresentationtheoreticmultiplicities, christandl2026plethysmbqp}. One famous example of this connection is the dimension of an irreducible representation $\lambda$ of $\mathbb{C}[S_n]$, which is given by the \textit{hook length formula} (\cref{def:hook_length_formula}):
\begin{equation}
    \label{eq:intro_hook_length}
    \dim V^{\lambda} = \frac{|\lambda|!}{\prod_{(i,j) \in \lambda}h(i,j)} 
\end{equation}
\cref{eq:intro_hook_length} implies that $\dim V^\lambda$ can be computed classically in polynomial time, a fact we exploit in this paper. However, many other important representation-theoretic quantities are not known to have efficient classical algorithms, and in fact offer promising candidates for quantum speedups~\cite{bravyi2024quantumcomplexitykroneckercoefficients, larocca2025quantumalgorithmsrepresentationtheoreticmultiplicities}. 

As an example, consider the \textit{Kronecker coefficient} $g_{\rho\sigma}^\lambda$ defined in~\cite{Murnaghan1938}. $g_{\rho\sigma}^\lambda$ counts how many copies of an irrep $\lambda$ appear in the tensor product of two irreps $\rho$ and $\sigma$, i.e. 
\begin{equation}
    \label{eq:kronecker_coeff}
    \rho \otimes \sigma \cong \bigoplus_{\lambda} g_{\rho\sigma }^\lambda \lambda
\end{equation}
In addition to its mathematical significance, computing specific Kronecker coefficients is an important subroutine for many quantum algorithms~\cite{klyachko2004quantummarginalproblemrepresentations, Christandl_2006}. However, it is not a coincidence that these algorithms are restricted to special cases: computing an \textit{arbitrary} Kronecker coefficient (for $\mathbb{C}[S_n]$) is known to be $\#\mathsf{P}$-hard~\cite{Ikenmeyer_2017}. 

In a recent breakthrough,~\cite{bravyi2024quantumcomplexitykroneckercoefficients, larocca2025quantumalgorithmsrepresentationtheoreticmultiplicities, christandl2026plethysmbqp} proved that the Kronecker coefficients for $\mathbb{C}[S_n]$ can be computed in $\#\mathsf{BQP}$, the quantum analog of $\#\mathsf{P}$.\ Moreover, the techniques of~\cite{larocca2025quantumalgorithmsrepresentationtheoreticmultiplicities} give a $\poly(n)$ time quantum algorithm when all irreps have $\poly(n)$ dimension. The algorithms in~\cite{bravyi2024quantumcomplexitykroneckercoefficients, larocca2025quantumalgorithmsrepresentationtheoreticmultiplicities} are remarkably simple, requiring only the application of QFT as well as the application of an irrep matrix\footnote{The algorithm in~\cite{christandl2026plethysmbqp} requires a Schur transform. We discuss the generalization to diagram algebras in~\cref{sec:schur-transform-application}.}, and may be a source of quantum advantage over the best classical algorithms.\footnote{The result of~\cite{bravyi2024quantumcomplexitykroneckercoefficients} for $\mathbb{C}[S_n]$ was partially dequantized by Panova~\cite{panova2025polynomialtimeclassicalversus}, motivating the search for other types of coefficients where quantum advantage may persist.}

Since our work gives an efficient QFT for several semisimple algebras (and partially addresses the problem of applying their irrep matrices in~\cref{sec:apply_irreps}), this opens up the possibility that these algorithms could be generalized to compute the corresponding Kronecker coefficients. \ifanon
\else
\!\!\!\!\footnote{We thank Dmitry Grinko and Maris Ozols for a discussion which sparked this idea.}\!\!\!
\fi 
Although we used the Kronecker coefficients as an example, the methods in~\cite{larocca2025quantumalgorithmsrepresentationtheoreticmultiplicities, christandl2026plethysmbqp} apply to several types of multiplicities such as the Littlewood-Richardson coefficients~\cite{LittlewoodRichardson1934}, the Plethysm coefficients~\cite{Littlewood1958}, and the Kostka numbers~\cite{Kostka1882}, all of which have a generalization in the context of diagram algebras. 

\subsubsection{Additional Future Directions}
In addition to the applications above, our derivation of an efficient QFT for semisimple algebras raises several interesting questions for further exploration. We list some of these below: 
\label{sec:future_directions}

\paragraph{Perturbing the Diagram Basis.}
While the Fourier transforms considered in this work are only approximately unitary, the algorithm we give to approximate the Fourier transform is exactly unitary. As a consequence, one could define a ``perturbed diagram basis'' obtained by the taking the pre-images of $\{\ket{\rho, i, j}\}_{\rho, i, j}$ under our algorithm. The resulting basis is exactly orthonormal, but has elements of the form 
\begin{equation}
    \label{eq:perturbed_basis_form}
    \ket{a} + \sum_{\substack{b \in \mathcal{B}(A) \\ b \ne a}} \gamma_{b} \ket{b},\;\; \sum_{\substack{b \in \mathcal{B}(A) \\ b \ne a}} |\gamma_{b}|^2 \le \varepsilon
\end{equation}
for some $\varepsilon \le \poly(|A|) \cdot d^{-1/2}$. Setting the ``standard basis'' to be this basis would result in a perfectly unitary QFT. Does this basis have any representation-theoretic or operational interpretation? Furthermore, understanding this basis may be useful step towards a Krovi-style Schur transform for diagram algebras~\cite{Krovi_2019, burchardt2025highdimensionalquantumschurtransforms} (see~\cref{sec:schur-transform-application} above). 

\paragraph{The Fourier Transform as a Direct Sum.}
One consequence of~\cref{thm:intro_invariant} is that when $d \gg \poly(|A|)$, the Fourier transform over $A$ decomposes as a direct sum of Fourier transforms restricted to each propagating number. More precisely, define 
\begin{equation}
    K_r = \text{span}\{\ket{D}: \pn(D) = r\}
\end{equation}
Then, up to $\poly(|A|) \cdot d^{-1/2}$ error, 
\begin{equation}
    \label{eq:direct_sum_form}
    \ft_A = \bigoplus_{r} \ft_A|_{K_r}
\end{equation}
for any $A \in \{P_n(d), B_n(d), B_{r,s}(d)\}$. 
Several questions arise from this direct sum decomposition. For example, if one is promised that all inputs satisfy $\pn(D)=r$, could \cref{eq:direct_sum_form} yield a faster algorithm? More generally, can $\ft_A|_{K_r}$ be related to the Fourier transform on the algebra with $n$ replaced by $r$? Note that in general these operators are not equivalent, since $\ft_A|_{K_r}$ acts on a higher-dimensional space. 

\paragraph{Outline.} In Section 2, we cover the representation theory and linear algebra background needed for the paper. In Section 3, we give a rigorous definition of the Fourier transform over semisimple algebras, and introduce a notion of approximate orthogonality we call $\delta$-niceness. In Section 4, we give an overview of the diagram algebras studied in this paper, and prove several results about their Fourier states. In Section 5, we introduce the notion of a subalgebra adapted basis, and in Section 6 we give the algorithm for the QFT on semisimple algebras. 

Throughout, we assume familiarity with the basics of quantum computation, including braket notation and quantum circuits. We also use $\wt{O}(\cdot)$ to hide polylogarithmic factors in the big-$O$ notation.   

\section{Semisimple Algebras and the Fourier Basis}
\label{sec:rep_theory_background}
Let $A$ be an associative $\mathbb{C}$-algebra with unity. In this work, we assume that $A$ is finite-dimensional, and use $|A|$ to denote the dimension of $A$.
\subsection{Representation Theory}
\begin{definition}
    Let $V$ be a vector space. A \textit{representation} is an algebra homomorphism $\rho: A \rightarrow \text{End}(V)$, satisfying $\rho(a) + \rho(b) = \rho(a + b)$, $\rho(a)\rho(b) = \rho(ab)$, and $\rho(1) = I$.
\end{definition}
For a representation $\rho$, we will sometimes use $V^\rho$ to denote the image of $\rho$, which is a subspace of $\text{End}(V)$. 
\begin{definition}
    Two representations $\rho_1$, $\rho_2$ are said to be \textit{equivalent} if they are related by a change of basis, i.e. there is a matrix $S$ such that $S\rho_1(a)S^{-1} = \rho_2(a)$ for all $a \in A$. Otherwise, $\rho_1$ and $\rho_2$ are inequivalent. 
\end{definition}
\begin{definition}
    A representation $\rho: A\rightarrow \text{End}(V)$ is \textit{irreducible} if, for any subspace $W \subseteq V$, $$\rho(a)W \subseteq W \implies W \in \{\{0\}, V\}$$ Equivalently, there are no non-trivial $\rho$-invariant subspaces of $V$. 
\end{definition}
Irreducible representations are often called ``irreps'' for brevity. We use $\wh{A}$ to denote a complete set of inequivalent irreps of $A$. 

A useful way to study an algebra $A$ is by decomposing it into irreps. For example, if $A = \mathbb{C}[G]$ is a finite-dimensional group algebra, then Maschke's Theorem implies that 
\begin{equation}
    A \cong \bigoplus_{\rho \in \wh{A}} V^\rho
\end{equation} 
This decomposition may not hold for more general finite-dimensional algebras. When the above decomposition holds, we say that $A$ is semisimple:
\begin{definition}
    \label{def:semisimple_algebra}
    An algebra $A$ is \textit{semisimple}\footnote{An equivalent definition of a semisimple algebra is an algebra with a trivial Jacobson radical. In this paper, it will be useful to work with~\cref{def:semisimple_algebra}.} if 
    \begin{equation}
        \label{eq:semisimple_algebra}
        A \cong \bigoplus_{\rho \in \wh{A}} V^\rho \cong \bigoplus_{\rho \in \wh{A}} M_{d_\rho}(\mathbb{C})
    \end{equation} where $M_{d_{\rho}}(\mathbb{C})$ is the space of $d_\rho \times d_\rho$ matrices over $\mathbb{C}$. 
\end{definition}
In the above equation, $d_\rho$ is the \textit{dimension} of the irrep $\rho$. For any semisimple algebra $A$, a \textit{Fourier transform} is a function carrying out the isomorphism\footnote{Many such isomorphisms can exist. In~\cref{sec:subalgebra_adapteD_basis}, we will define a particular isomorphism which can be carried out efficiently.} in \cref{eq:semisimple_algebra}, which maps $a \in A$ to the tuple $(\rho(a))_{\rho \in \wh{A}}$. A Fourier transform gives rise to a natural basis on $A$, called the Fourier basis:
\begin{definition}
    \label{def:fourier_basis}
     The \textit{Fourier basis} of an algebra $A$ is defined as the preimage of the natural matrix unit basis $\{\ket{i}\bra{j}_\rho\}_{\rho, i, j}$ under the isomorphism in \cref{eq:semisimple_algebra}. We will use $E_{ij}^\rho$ to denote the preimage of $\{\ket{i}\bra{j}_\rho\}_{\rho, i, j}$.
\end{definition}
In \cref{sec:fourier_basis_of_algebra}, we will give a precise characterization of the Fourier basis elements $\{E_{ij}^\rho\}_{\rho, i, j}$ in terms of the matrix elements of irreps. 
\subsection{Trace Forms and Dual Bases}
\label{sec:trace_forms_and_dual_bases}
Let $A$ be an algebra with basis $\mathcal{B}(A) = \{a_i\}_i$.
\begin{definition}
    A \textit{trace form} $\tau: A \mapsto \mathbb{C}$ is a linear function satisfying $\tau(xy) = \tau(yx)$. 
\end{definition}
For any representation $\rho$, one can define the associated trace form $\tau_\rho(x) = \tr(\rho(x))$.\ That $\tau_\rho$ is indeed a trace form follows from $\rho(xy) = \rho(x)\rho(y)$ and the cyclic property of the trace. We will be particularly interested in the (left) \textit{regular} trace form $\tau_\mathbf{L}$, which is the trace form corresponding to the (left) regular representation $\mathbf{L}$ given by
\begin{equation}
     \mathbf{L}(a_1)\ket{a_2} = \ket{a_1a_2}
\end{equation}
The left regular representation acts on a vector space spanned by all elements from some basis of $A$.
\begin{lemma}[\cite{maslen2016efficientcomputationfouriertransforms}, Theorem 2.5]
    \label{lem:trace_expansion} 
    If $\tau$ is a trace form on a semisimple algebra $A$, then
    \begin{equation}
        \tau(a) = \sum_{\rho \in \wh{A}} c_\rho \tr(\rho(a))
    \end{equation}
    for some constants $c_\rho$.
\end{lemma}
In particular, for the left regular trace form, $c_\rho = d_\rho$. 
Any symmetric trace form $\tau$ gives rise to a symmetric bilinear form $\braket{a_i, a_j}_\tau = \tau(a_ia_j)$, such as $\braket{\cdot, \cdot}_\mathbf{L}$ from the left regular representation above.\ If this bilinear form is nondegenerate, we can also define \textit{dual basis} with respect to $\tau$:
\begin{definition}
    \label{def:dual_basis}
    For any nondegenerate trace form $\tau$, there exists a basis $\mathcal{B}^*(A) = \{a_i^*\}_i$ such that $\braket{a_i^*, a_j}_\tau = \delta_{ij}$. $\mathcal{B}^*(A)$ is the \textit{dual basis} of $\mathcal{B}(A)$ with respect to $\tau$. 
\end{definition}

\begin{lemma} 
\label{lemma:basis-breakdown}
For any $x \in A$, it holds that 
\begin{equation}
    \label{eq:dual_basis_expansion}
    x = \sum_{a \in \mathcal{B}(A)} \braket{a^*, x}_{\tau} a
\end{equation}
where $a^*$ is the dual basis element with respect to $\tau$.
\end{lemma}
\begin{proof}
    Express $x$ in the basis $\mathcal{B}(A)$ as $x = \sum_{b \in \mathcal{B}(A)} \alpha_b \, b$, where $\alpha_b \in \mathbb{C}$. Then
    \begin{align}
        \sum_{a \in \mathcal{B}(A)} \braket{a^*, x}_{\tau} a
        \;
        =
        \sum_{a \in \mathcal{B}(A)} 
        \sum_{b \in \mathcal{B}(A)} 
        \alpha_b
        \braket{a^*, b}_{\tau} a
        \;
        =
        \sum_{b \in \mathcal{B}(A)} 
        \alpha_b
        \,
        b
        \;
        = 
        \;
        x
    \end{align}
    where the second equality follows from \Cref{def:dual_basis}.  
\end{proof}
The dual basis can be obtained by inverting the \textit{Gram matrix} $(G)_{ij} = \braket{a_i, a_j}_{\tau}$, where $a_i, a_j \in \mathcal{B}(A)$, and then reading off the rows or columns. 

Finally, for any basis $\mathcal{B}(A)$, we can always define an associated computational inner product:
\begin{definition}
    For a basis $\mathcal{B}(A)$, define the \textit{computational inner product} to be the Hermitian form satisfying $\braket{a_i, a_j}_0 = \delta_{ij}$ for all $a_i, a_j\in \mathcal{B}(A)$. 
\end{definition}
The computational inner product is the standard complex inner product, with ``standard basis vectors'' given by some fixed basis of $A$. For example, say we have two elements $x, y \in A$.\ We can express them in the basis $\mathcal{B}(A)$ as $x = \sum_{a \in \mathcal{B}(A)} \alpha_a \, a$ and $y = \sum_{a \in \mathcal{B}(A)} \beta_a \, a$, where $\alpha_a, \beta_a \in \mathbb{C}$.\ The computational inner product is simply $\braket{x, y}_0 = \sum_{a \in \mathcal{B}(A)} \overline{\alpha_a} \beta_a$. That is, if the algebra elements were expressed as (un-normalized) quantum states over a basis labeled by $\mathcal{B}(A)$, this would be the inner product between the quantum states.

Note that the computational inner product is highly basis dependent. The choice of computational basis determines how elements of $A$ are encoded on a quantum computer, as well as the precise form of the Fourier transform.
Both the computational inner product $\braket{\cdot, \cdot}_0$ and the left-regular trace form $\braket{\cdot, \cdot}_{\mathbf{L}}$ are important for the QFT on semisimple algebras. In particular, determining when one approximates the other is a crucial part of our analysis (see \cref{sec:qft_exposition} for further discussion).

\subsection{The Fourier Basis of a Semisimple Algebra}
\label{sec:fourier_basis_of_algebra}
We are now ready to give a precise description of the Fourier basis from \cref{def:fourier_basis}:
\begin{theorem}
Let $A$ be a semisimple algebra. For any choice of basis $\mathcal{B}(A)$, 
\label{thm:description_of_fourier_basis}
\begin{equation}
    E_{ij}^\rho = d_\rho \sum_{a \in \mathcal{B}(A)} \rho(a^*)_{ji} a 
\end{equation} 
Where $a^*$ is the dual element with respect to the left regular trace form $\tau_\mathbf{L}$.
\end{theorem}
\begin{proof}
    Fix an arbitrary choice of basis $\mathcal{B}(A)$. By~\cref{eq:dual_basis_expansion},
    \begin{equation}
        E_{ij}^\rho = \sum_{a \in \mathcal{B}(A)} \braket{a^*, E_{ij}^\rho}_{\mathbf L} a.
    \end{equation}
    It therefore suffices to compute the coefficient $\braket{a^*, E_{ij}^\rho}_{\mathbf L}$. By definition of $\tau_\mathbf{L}$ and~\cref{lem:trace_expansion},
    \begin{equation}
        \label{eq:coefficient_computation}
        \braket{a^*, E_{ij}^\rho}_{\mathbf L}
        = \tau_{\mathbf L}(a^* E_{ij}^\rho) = \sum_{\sigma \in \wh{A}} d_\sigma \tr\bigl(\sigma(a^*) \sigma(E_{ij}^\rho)\bigr) = d_\rho \tr(\rho(a^*)\ket{i}\bra{j}) = d_\rho\rho(a^*)_{ji}
    \end{equation}
    where in the second to last equality, we used that $E_{ij}^\rho$ is the preimage of $\ket{i}\bra{j}_{\rho}$ under the Fourier isomorphism. Substituting this coefficient gives the result. 
\end{proof}
As a corollary to~\cref{thm:description_of_fourier_basis}, we obtain a generalization of Schur's orthogonality theorem to semisimple algebras: 
\begin{corollary}[Schur Orthogonality for the Irreps of Semisimple Algebras]
    \label{thm:schur_orthogonality}
    For any basis $\mathcal{B}(A)$, and any $\rho, \sigma \in \wh{A}$,
    \begin{equation}
        d_\rho \sum_{a \in \mathcal{B}(A)} \rho(a^*)_{ji} \sigma(a)_{kl}
        = \delta_{\rho\sigma}\delta_{ik}\delta_{jl}.
    \end{equation}
\end{corollary}

\begin{proof}
    Apply any irrep $\sigma \in \wh{A}$ to the formula of \cref{thm:description_of_fourier_basis}:
    \begin{equation}
        \sigma(E_{ij}^\rho)
        = d_\rho \sum_{a \in \mathcal{B}(A)} \rho(a^*)_{ji} \sigma(a).
    \end{equation}
    Taking the $(k,l)$ entry of both sides gives
    \begin{equation}
        \sigma(E_{ij}^\rho)_{kl}
        = d_\rho \sum_{a \in \mathcal{B}(A)} \rho(a^*)_{ji} \sigma(a)_{kl}.
    \end{equation}
    By definition of $E_{ij}^\rho$, $\sigma(E_{ij}^\rho)_{kl} = \delta_{\rho\sigma}\delta_{ik}\delta_{jl}$.
\end{proof}

Interestingly,~\cref{thm:description_of_fourier_basis} implies that the Fourier basis is independent of the algebra basis, even though the definition involves summing over a fixed basis $\mathcal{B}(A)$. Moreover, the Fourier basis is only defined up to some choice of basis for each of the matrix algebras $M_{d_\rho}(\mathbb{C})$ (see \cref{eq:semisimple_algebra}).\ In \cref{sec:subalgebra_adapteD_basis}, we will fix a specific choice of basis on the matrix algebras, known as a \textit{subalgebra adapted basis}, which will enable the corresponding Fourier transform to be carried out efficiently. 
\begin{lemma}
Independent of the basis for the matrix algebras, the Fourier basis states have several important properties:
\begin{enumerate}
 \item Matrix Multiplication Rule: $ E_{ij}^\rho E_{kl}^\sigma = \delta_{\rho\sigma}\delta_{jk} E_{il}^\rho$.
 \item Index Swap Duality: $(E_{ji}^\rho)^* = E_{ij}^\rho/d_\rho$, for the dual basis corresponding to $\tau_\mathbf{L}$. 
 \item Left Action Invariance: $\mathbf{L}(a) E_{ij}^\rho = \sum_i \rho(a)_{ki}E^\rho_{kj}$.
\end{enumerate}
\end{lemma}
\begin{proof}
    The first property follows by applying the inverse Fourier transform on two matrix units. 
    The second property follows from Schur orthogonality (\cref{thm:schur_orthogonality}): 
    \begin{align}
        \frac{1}{d_{\rho}}\braket{E_{ij}^\rho, E_{lk}^\sigma}_{\mathbf{L}} &= \frac{1}{d_{\rho}} \tr(\mathbf{L}(E_{ij}^\rho  E_{lk}^\sigma)) \\ = \frac{1}{d_{\rho}} \delta_{\rho\sigma} \delta_{jl} \tr(\mathbf{L}(E_{ik}^\rho)) &= \frac{1}{d_{\rho}} \delta_{\rho\sigma} \delta_{jl} \sum_{\sigma \in \wh{A}} d_\sigma \tr(\sigma(E_{ik}^\rho)) \\
       = \frac{1}{d_{\rho}} \delta_{\rho\sigma} \delta_{jl} d_\rho \tr(\ket{i}\bra{k}) &= \delta_{\rho\sigma}\delta_{ik}\delta_{jl} 
    \end{align}
    The first equality uses the definition of $\braket{\cdot, \cdot}_\mathbf{L}$, and the second uses the matrix multiplication property (the first property above). The third uses~\cref{lem:trace_expansion}, and the fourth uses the definition of the Fourier basis. Finally, the last equality uses $\tr(\ket{i}\bra{k}) = \delta_{ik}$. 
    
    For the third property, we evaluate
    \begin{equation}
        \mathbf{L}(a) E_{ij}^\rho 
        = 
        aE_{ij}^\rho   
        = 
        \sum_{\sigma, k, l} \left\langle\left(E_{kl}^\sigma\right)^*, a\right\rangle_\mathbf{L} E_{kl}^\sigma E_{ij}^\rho    
        = 
        \sum_{\sigma, k, l} \frac{1}{d_\sigma}\braket{E_{lk}^\sigma, a }_\mathbf{L} E_{kl}^\sigma E_{ij}^\rho    
        = 
        \sum_k \rho(a)_{ki}E^\rho_{kj}
    \end{equation}
    where the second equality follows from~\Cref{lemma:basis-breakdown}, the third follows from the second property above, and we have used~\cref{eq:coefficient_computation} for the final substitution.
\end{proof}

\section{The Quantum Fourier Transform For Semisimple Algebras}
\label{sec:qft_exposition}
The main goal of this work is to efficiently implement a basis change from the computational basis $\mathcal{B}(A)$ to the Fourier basis $\{E_{ij}^\rho\}_{\rho, i, j}$:
\begin{equation}
    \label{eq:incorrect_ft}
    \ket{a} \mapsto  \sum_{\rho \in \wh{A}} \; \sum_{i, j = 1}^{d_\rho} \braket{{E_{ij}^\rho}, a}_0 \ket{\rho, i, j}
\end{equation}
In particular, for several different algebras $A$, we will construct a quantum algorithm which approximately carries out this basis change using only $\polylog |A|$ gates.

However, even before we can consider an efficient algorithm, there are two difficulties that we need to handle for it to be well-defined as a unitary. 
Specifically, while the Fourier basis as defined in~\Cref{thm:description_of_fourier_basis} is orthonormal with respect to the left-regular trace form $\braket{\cdot, \cdot}_\mathbf{L}$, it is in general neither normalized nor orthogonal with respect to the computational inner product $\braket{\cdot, \cdot}_0$. 
In other words, the transformation in~\cref{eq:incorrect_ft} may be far from unitary, and thus cannot be implemented even approximately on a quantum computer. 

Even for a group algebra, it is necessary to normalize the Fourier basis in terms of $\braket{\cdot, \cdot}_0$.  For example, even when $A =\mathbb{C}[G]$ is a group algebra, we have that
\begin{align}
    \label{eq:group_ex}
     ||E_{ij}^\rho||_2^2 
     :=  
     \braket{E_{ij}^\rho, E_{ij}^\rho}_0 
     &
     = 
     d_\rho^2 \sum_{g \in G} \overline{\rho(g^{*})_{ij}} \rho(g^{*})_{ij} 
     = 
     \frac{d_\rho^2}{|G|} \sum_{g \in G} \overline{\rho(g^{-1})_{ij}} \rho(g^{*})_{ij} 
     \\
     &
     = 
     \frac{d_\rho^2}{|G|} \sum_{g \in G} \rho(g)_{ji}\rho(g^{*})_{ij}  
     = 
     \frac{d_\rho^2}{|G|} \cdot \frac{1}{d_\rho} 
     = 
     \frac{d_\rho}{|G|}
\end{align}
where the second equality follows from the definitions of the Fourier basis and the computational inner product; the third equality uses the fact that for group algebras, the dual basis has the form $g^* = g^{-1}/|G|$; and the fifth equality follows from Schur orthogonality (\Cref{thm:schur_orthogonality}). 

In the standard definition of the quantum Fourier transform on a group, it is sufficient to account for this by mapping the computational basis to the \textit{normalized} Fourier basis instead:  
\begin{definition}
    \label{def:algebra_qft}
    Let $A$ be a semisimple algebra with computational basis $\mathcal{B}(A)$. The \textit{(normalized) Fourier Transform over $A$}, denoted $\ft_{A}$, is the linear transformation which maps $\mathcal{B}(A)$ to the normalized Fourier basis, i.e.\
    \begin{equation}
        \ft_A\ket{a} = \sum_{\rho \in \wh{A}}  \sum_{i, j = 1}^{d_\rho} \frac{1}{||E_{ij}^\rho||_2}\braket{{E_{ij}^\rho}, a}_0 \ket{\rho, i, j} =  \sum_{\rho \in \wh{A}} d_\rho \sum_{i, j = 1}^{d_\rho} \frac{\overline{\rho(a^*)_{ji}}}{||E_{ij}^\rho||_2}\ket{\rho, i, j} 
    \end{equation}
\end{definition}
\noindent From~\cref{eq:group_ex}, $||E_{ij}^\rho||_2 = \sqrt{d_\rho/|G|}$ if $A = \mathbb{C}[G]$ is a group algebra, and so
 \begin{equation*}
        \ft_{\mathbb{C}[G]}\ket{g} =  \sum_{\rho \in \wh{A}}  \sqrt{\frac{d_\rho}{|G|}}\sum_{i, j = 1}^{d_\rho} \rho(g)_{ij} \ket{\rho, i, j} 
    \end{equation*}
 which is the standard definition of the quantum Fourier transform on a group. 
 
 Like the choice of Fourier basis states, $\ft_A$ is only defined up to a choice of basis on the matrix algebras in~\cref{eq:semisimple_algebra}. Ultimately, we will implement $\ft_A$ with respect to a \textit{subalgebra adapted basis} (\cref{sec:subalgebra_adapteD_basis}).

 However, even after normalization, the transformation may be non-unitary, since for non-group algebras, the Fourier basis may not be orthogonal under $\braket{\cdot, \cdot}_0$. Consider the following example, for a non-group algebra spanned by just two basis elements: 
 \begin{example}
    As a simple example, we will use the partition algebra $P_1(d)$. It is an algebra spanned by two elements, $I$ and $P$ (see \cref{sec:diagram_algebras} for a formal definition). These elements satisfy the relations $I^2 = I$, $IP = PI = P$, and $P^2 = dP$.\footnote{Later, when we analyze $P_n(d)$ for general $n$, we will first rescale the basis elements $I$ and $P$. See~\cref{eq:basis_for_a_diagram_algebra}).} $P_1(d)$ has two one-dimensional irreps: the \textit{one-box irrep} $\rho_{\Box}$, and the \textit{zero-box irrep} $\rho_{\emptyset}$. The action of these irreps on $I$ and $P$ is given as follows:%
    \footnote{
        It is easy to verify that (1) these are valid representations, as they satisfy the relations; (2) they are irreducible since they are one-dimensional; and (3) they are the only valid irreps of $P_1(d)$, since the algebra is 2-dimensional.
    }
    \begin{equation}
        \rho_{\Box}(I) = 1,\; \rho_{\Box}(P) = 0\qquad \rho_{\emptyset}(I) = 1,\; \rho_{\emptyset}(P) = d
    \end{equation}
     
    To compute the Fourier basis states of $P_1(d)$, we must first find the dual basis with respect to the regular trace form $\tau_\mathbf{L}$ 
    (see~\cref{sec:trace_forms_and_dual_bases}). To do so, we compute the Gram matrix of $\braket{\cdot, \cdot}_{\tau_\mathbf{L}}$:  
     
    \begin{equation}
        G = \begin{pmatrix}
            \tr(\mathbf{L}(I)) & \tr(\mathbf{L}(P)) \\ 
            \tr(\mathbf{L}(P)) & d\cdot \tr(\mathbf{L}(P)) \\ 
        \end{pmatrix}
        = \begin{pmatrix}
            2 & d \\ 
            d & d^2 \\ 
        \end{pmatrix},
    \end{equation}
    which has inverse
    \begin{equation}
        G^{-1} = \begin{pmatrix}
            1 & -1/d 
            \\ 
            -1/d & 2/d^2
        \end{pmatrix}.
    \end{equation}
    \noindent Reading off the rows of $G^{-1}$, we can see that the dual basis elements are $I^* = I - \frac1d P$ and $P^* = -\frac1d I + \frac{2}{d^2} P$.
    By linearity, the action of the two irreps above on the dual basis elements can be evaluated as
    \begin{align}
        \rho_{\Box}(I^*) = \rho_{\Box}\left(I - \frac1d P\right) = 1,
        &\quad 
        \rho_{\Box}(P^*) = \rho_{\Box}\left(-\frac1d I + \frac{2}{d^2} P\right) = -\frac1d 
        \\
        \rho_{\emptyset}(I^*) = \rho_{\emptyset}\left(I - \frac1d P\right) = 0,
        &\quad 
        \rho_{\emptyset}(P^*) = \rho_{\emptyset}\left(-\frac1d I + \frac{2}{d^2} P\right) = \frac1d
    \end{align}
    The Fourier basis elements of $P_1(d)$ are therefore 
    \begin{align}
        E_{11}^{\rho_{\Box}} 
        &= 
        \rho_{\Box}(I^*) \; I
        +
        \rho_{\Box}(P^*) \; P
        =
        I - \frac1d P
        \\
        E_{11}^{\rho_{\emptyset}} 
        &= 
        \rho_{\emptyset}(I^*) \; I
        +
        \rho_{\emptyset}(P^*) \; P
        =
        \frac1d P. 
    \end{align}
    These have norms 
    $
        \lVert E_{11}^{\rho_{\Box}} \rVert_2 
        = 
        \sqrt{1 + \frac{1}{d^2}}
        = 
        \frac{\sqrt{d^2 + 1}}{d}
    $
    and
    $
        \lVert E_{11}^{\rho_{\emptyset}} \rVert_2 
        = 
        \sqrt{\frac{1}{d^2}}
        = 
        \frac{1}{d}
    $.
    
    Finally, we can plug these values, as well as the norms of the Fourier basis states, into the definition of the Fourier transform (\cref{def:algebra_qft}): 
    \begin{align}
        \ft_{P_1(d)}\ket{I} 
        &= 
        \frac{
            \rho_{\Box}(I^*)
        }{
            \lVert E_{11}^{\rho_{\Box}} \rVert_2
        }
        \ket{\rho_{\Box}, 1, 1} 
        + 
        \frac{
            \rho_{\emptyset}(I^*)
        }{
            \lVert E_{11}^{\rho_{\emptyset}} \rVert_2
        }
        \ket{\rho_{\emptyset}, 1, 1} 
        \\ 
        &= 
        \frac{d}{\sqrt{d^2 + 1}}
        \;
        \rho_{\Box}(I^*)
        \ket{\rho_{\Box}, 1, 1} 
        + 
        d
        \;
        \rho_{\emptyset}(I^*)
        \ket{\rho_{\emptyset}, 1, 1} 
        \\ 
        &= 
        \frac{d}{\sqrt{d^2 + 1}}
        \ket{\rho_{\Box}, 1, 1} 
        \\ 
        \ft_{P_1(d)}\ket{P} 
        &= 
        \frac{
            \rho_{\Box}(P^*)
        }{
            \lVert E_{11}^{\rho_{\Box}} \rVert_2
        }
        \ket{\rho_{\Box}, 1, 1} 
        + 
        \frac{
            \rho_{\emptyset}(P^*)
        }{
            \lVert E_{11}^{\rho_{\emptyset}} \rVert_2
        }
        \ket{\rho_{\emptyset}, 1, 1} 
        \\ 
        &= 
        \frac{d}{\sqrt{d^2 + 1}}
        \;
        \rho_{\Box}(P^*)
        \ket{\rho_{\Box}, 1, 1} 
        + 
        d
        \;
        \rho_{\emptyset}(P^*)
        \ket{\rho_{\emptyset}, 1, 1} 
        \\ 
        &= 
        -\frac{1}{\sqrt{d^2 + 1}}
        \ket{\rho_{\Box}, 1, 1} 
        + 
        \ket{\rho_{\emptyset}, 1, 1} 
    \end{align}
     Therefore, $\ft_{P_1(d)}$ is not unitary, and cannot be implemented exactly on a quantum computer. Nonetheless, when the parameter $d$ is large, $\ft_{P_1(d)}$ is extremely close to unitary (in this example, 
     $
         \ft_{P_1(d)}\ket{I}
         \to
         \ket{\rho_{\Box}, 1, 1}
     $
     and
     $
         \ft_{P_1(d)}\ket{P}
         \to
         \ket{\rho_{\emptyset}, 1, 1}
     $
     \!), and so an \textit{approximate} unitary implementation does exist.
\end{example}
In~\cref{sec:diagram_algebras_nice}, we will generalize this observation, and show that the Fourier basis elements of the relevant diagram algebras become orthonormal in the limit of large $d$. 
   
\subsection{\texorpdfstring{\(\delta\)-nice}{delta-nice} Algebras}
As one may infer from the previous section, we will primarily be interested in algebras for which the (normalized) Fourier basis elements are approximately orthonormal under $\braket{\cdot, \cdot}_0$. We call such algebras $\delta$-nice: 
\begin{definition}
    Let $A$ be an algebra with computational basis $\mathcal{B}(A)$, with computational inner product $\braket{\cdot, \cdot}_0$.
    Let $G$ be the Gram matrix of the normalized Fourier basis with respect to $\braket{\cdot, \cdot}_0$, i.e. 
    \begin{equation}
        \label{eq:d_nice_equation}
        G_{(\rho, i, j), (\sigma, k, l)} = \frac{1}{||E_{ij}^\rho||_2||E_{kl}^\sigma||_2}\braket{E_{ij}^\rho, E_{kl}^\sigma}_0
    \end{equation} We say that $A$ is \textit{$\delta$-nice} (with respect to $\mathcal{B}(A)$) if $||G - I||_{\infty} \le \delta$. 
\end{definition}
When the choice of basis is fixed or clear from context, we say that $A$ is $\delta$-nice. In the several lemmas, it will useful to assume that $\delta \le 1/2$, so we make this assumption throughout.  However, we will later require $\delta = 1/\poly(|A|)$ as well.  

Many properties of a group algebra (which is $0$-nice) generalize to $\delta$-nice algebras. For example, when $A$ is $\delta$-nice, $\ft_A$ is closely approximated by a simpler transformation that resembles the standard definition of the group QFT: 
\begin{theorem}
    \label{thm:approx_qft}
    Let 
    \begin{equation}
        \wt{\ft_A}\ket{a} = \sum_{\rho \in \wh{A}}  \sum_{i, j = 1}^{d_\rho} ||E_{ij}^\rho||_2\cdot  \rho(a)_{ij} \ket{\rho, i, j} 
    \end{equation}
   Then, $||\ft_A - \wt{\ft_A}||_{\infty} \le O(\delta \cdot |A|^{3/2})$.
\end{theorem}
\begin{proof}
The proof of \cref{thm:approx_qft} is a corollary of the following lemma: 
\begin{lemma}
\label{lem:approx_nice}
If $A$ is $\delta$-nice, then
\begin{equation}
    \rho(a)_{ij} = \frac{d_\rho\overline{\rho(a^*)_{ji}}}{||E_{ij}^\rho||_2^2} \pm \frac{O(\delta \cdot  \sqrt{|A|})}{||E_{ij}^\rho||_2}  
\end{equation}
or equivalently, 
\begin{equation}
    ({\rho(a)_{ij} \cdot ||E_{ij}^\rho||_2) \pm {O(\delta \cdot  \sqrt{|A|})}}  = d_\rho \cdot \frac{\overline{\rho(a^*)_{ji}}}{||E_{ij}^\rho||_2} 
\end{equation}
\end{lemma}
\noindent Note that this equation is true with no error when $A$ is a group algebra.
\begin{proof}
    
     For any $a \in A$, we can expand it in the (normalized) Fourier basis using the $\braket{\cdot, \cdot}_0$ inner product:
     \begin{equation}
     \label{eq:coeff_compare_one}
         a = \sum_{\rho, i, j}  \left\langle{\left(\frac{E_{ij}^\rho}{||E_{ij}^\rho||_2}\right)^*, a} \right\rangle_{0} \frac{1}{||E_{ij}^\rho||_2}E_{ij}^\rho
     \end{equation}
     Where the superscript $*$ denotes the dual element of $E_{ij}^\rho/||E_{ij}^\rho||_2$ \textit{with respect to} $\braket{\cdot, \cdot}_0$. Using the definition of the dual basis and the Gram matrix $G$ (see \Cref{def:dual_basis}),  
    \begin{equation}
        \label{eq:coeff_compare_two}
       \left(\frac{E_{ij}^\rho}{||E_{ij}^\rho||_2}\right)^* = \sum_{\sigma, k, l} G^{-1}_{(\rho, i, j), (\sigma, k, l)} \frac{E_{kl}^\sigma}{||E_{kl}^\sigma||_2}
    \end{equation}
     substituting~\cref{eq:coeff_compare_two} into~\cref{eq:coeff_compare_one}, 
     \begin{align}
          a &= \sum_{\rho, i, j} \sum_{\sigma, k, l} G^{-1}_{(\rho, i, j), (\sigma, k, l)} \braket{E_{kl}^\sigma, a}_0 \frac{1}{||E_{kl}^\sigma||_2||E_{ij}^\rho||_2}E_{ij}^\rho \\ 
          &= \sum_{\rho, i, j} \sum_{\sigma, k, l} G^{-1}_{(\rho, i, j), (\sigma, k, l)}  \frac{d_\sigma \overline{\sigma(a^*)_{lk}}}{||E_{kl}^\sigma||_2||E_{ij}^\rho||_2}E_{ij}^\rho
     \end{align}
    We can also expand $a$ using the bilinear form $\braket{\cdot, \cdot}_{\mathbf{L}}$, defined in~\cref{sec:trace_forms_and_dual_bases}:
    \begin{equation}
        a 
        = \sum_{\rho, i, j} \left\langle\left(E_{ij}^\rho\right)^*, a\right\rangle_{\mathbf{L}} E_{ij}^\rho 
        = \sum_{\rho, i, j} \frac{1}{d_\rho} \braket{E_{ji}^\rho, a}_{\mathbf{L}} E_{ij}^\rho 
        = \sum_{\rho, i, j} \rho(a)_{ij} E_{ij}^\rho
    \end{equation}
    Since this is true for all $a \in A$, comparing coefficients implies that 
    \begin{align}
        \rho(a)_{ij} &=  \sum_{\sigma, k, l} G^{-1}_{(\rho, i, j), (\sigma, k, l)}  \frac{d_\sigma \overline{\sigma(a^*)_{lk}}}{||E_{kl}^\sigma||_2||E_{ij}^\rho||_2} \\
        &= \frac{d_\rho \overline{\rho(a^*)_{ji}}}{||E_{ij}^\rho||_2^2} +  \frac{1}{||E_{ij}^\rho||_2}  \sum_{\sigma, k, l} (G^{-1} - I)_{(\rho, i, j), (\sigma, k, l)} \frac{\braket{E_{kl}^\sigma, a}_0}{||E_{kl}^\sigma||_2}
    \end{align}
    
    So, it suffices to bound the right most term in the second line. By Cauchy-Schwarz, $\lvert\braket{E_{kl}^\sigma, a}_0\rvert \le ||E_{kl}^\sigma||_2$. \cref{lem:operator_norm_of_inverse}, \cref{lem:off_diagonal_sums}, and the $\delta$-nice assumption imply the absolute value of the error term is bounded as
    \begin{equation}
        \frac{1}{||E_{ij}^\rho||_2}\sum_{\sigma, k, l} (G^{-1} - I)_{(\rho, i, j), (\sigma, k, l)} \le \frac{\delta}{1 - \delta} \cdot \frac{ \sqrt{|A|}}{||E_{ij}^\rho||_2}
    \end{equation}
    Finally, $\delta \le 1/2 \implies \delta/(1- \delta) \le 2\delta$.
\end{proof}
Returning to the proof of \cref{thm:approx_qft},~\cref{lem:approx_nice} and \cref{lem:opnorm_bound} with $M = \ft_A - \wt{\ft_A}$ implies that 
\begin{equation}
    ||(\ft_A - \wt{\ft_A} )||_\infty \le O(\delta \cdot |A|^{3/2})
\end{equation}
\end{proof}

\section{Diagram Algebras}
\label{sec:overall_diagram_algebras}
The semisimple algebras we will consider in this paper belong to a family of algebras known as \textit{diagram algebras}. Diagram algebras are natural generalizations of the symmetric group algebra $\mathbb{C}[S_n]$, and have many applications in quantum information—see~\cref{sec:intro_semisimple} for an overview.  
\subsection{Partition Diagrams}
In order to define a diagram algebra, we first need the notion of a partition diagram:
\begin{definition}
A \emph{partition diagram} is a set partition of the $2n$ vertices $ \{1,\dots,n\}\ \cup\ \{1',\dots,n'\}$, usually visualized with a $n \times 2$ grid. For example, the following is a partition diagram with $n=4$:
\begin{figure}[H]
    \centering
    \includegraphics[width=0.3\linewidth]{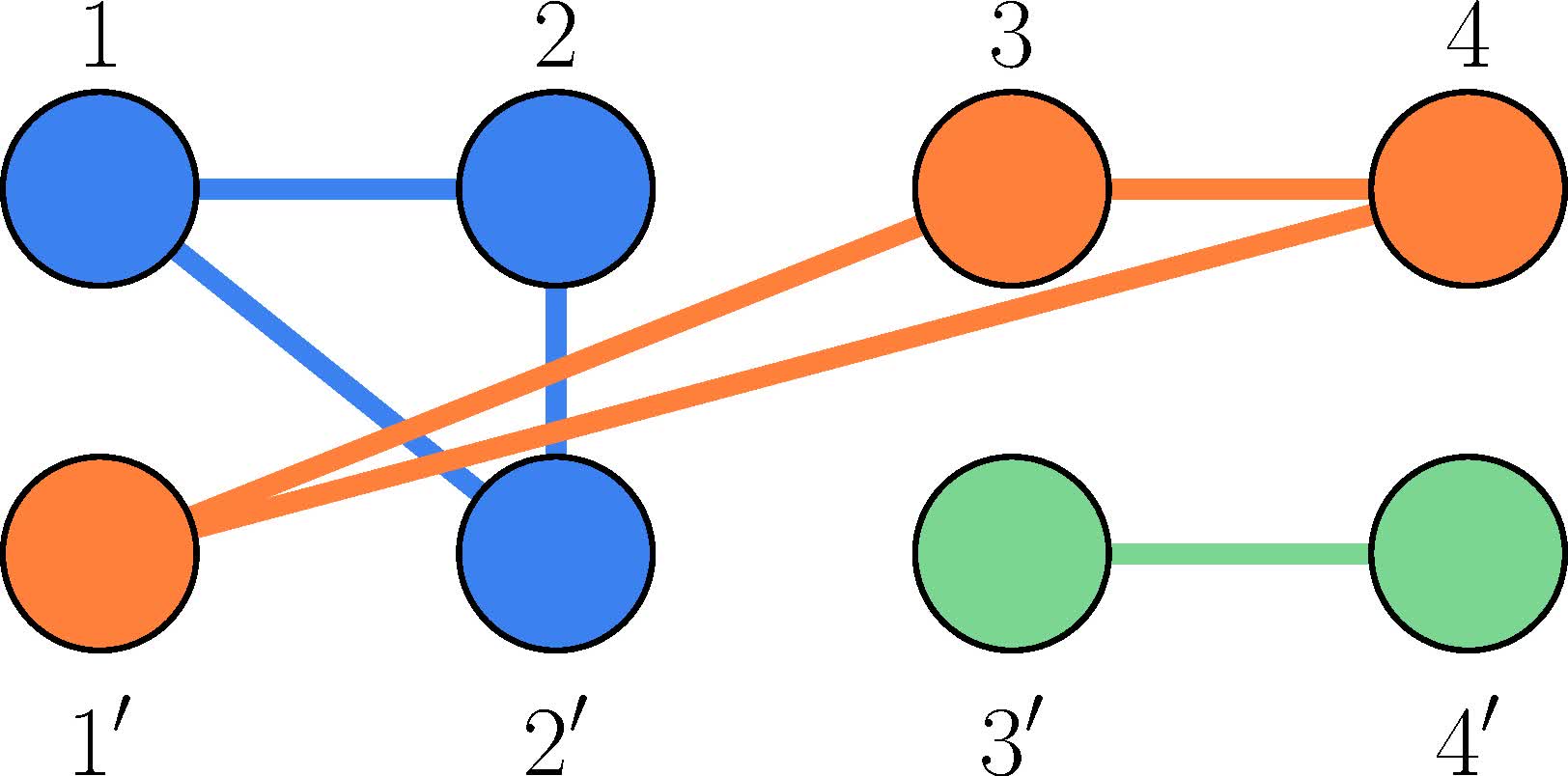}
\end{figure}
\end{definition}
Given a partition diagram $D$, we use $\cc(D)$ to denote the number of connected components in $D$.\ For example, $\cc(D) = 3$ for the above diagram. These connected components are often referred to as \textit{blocks}.

Another important statistic of a partition diagram is its propagating number:
\begin{definition}
Let $D$ be a partition diagram. A block of $D$ is \emph{propagating} if it intersects both $\{1,\dots,n\}$ and $\{1',\dots,n'\}$. The \emph{propagating number} of $D$ is
\begin{equation}
    \label{eq:propagating_number}
    \pn(D)\ :=\ \#\{\text{propagating blocks of }D\}.
\end{equation}
For example, in the diagram above, $\pn(D) =  2$.
\end{definition} 
\subsection{The Partition Algebra and its Subalgebras}
\label{sec:diagram_algebras} 
The \textit{partition algebra} endows the set of partition diagrams with additive and multiplicative operations. Diagram algebras are the subalgebras of the partition algebra that restrict the set of valid diagrams.%
\footnote{
    Diagram algebras are sometimes defined more broadly to include additional algebras which attach extra features to the diagrams, which we do not consider here.
}
In this section, we define the partition algebra, as well as the relevant subalgebras considered in this work. 
\paragraph{The Partition Algebra}
\begin{definition}
The \emph{partition algebra} $P_n(d)$ is the $\mathbb{C}$-algebra spanned by all partition diagrams with $n$ columns. Multiplication of two diagrams is defined as follows: given diagrams $D_1$ and $D_2$, $D_1D_2$ is the stacked diagram obtained by identifying the vertices in the bottom row of $D_1$ with the vertices in the top row of $D_2$. The vertices in the identified middle rows are then removed—if $c$ connected components are removed this way, the product is multiplied by $d^c$:
\end{definition}
\begin{figure}[H]
    \centering
    \includegraphics[width=\linewidth]{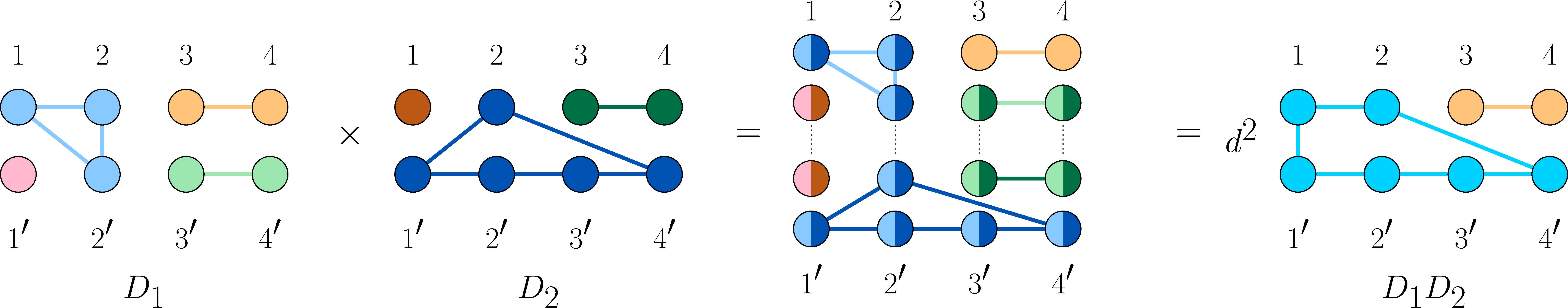}
    \label{fig:multiplicaiton_partition_algebra}
    \caption{An example of multiplying two partition diagrams.}
\end{figure}
\begin{lemma}[\cite{halverson2004partitionalgebras}]
    \label{lem:propgation_number_sub_mult}
    For all partition diagrams $D_1$ and $D_2$, $\pn(D_1D_2) \le \min \{\pn(D_1), \pn(D_2)\}$. 
\end{lemma}
\noindent $P_n(d)$ has the following canonical generating set:
\begin{lemma}[\cite{halverson2004partitionalgebras}]
\label{lem:partition_generators}
 $P_n(d)$ is generated by the \textit{swap} diagrams $ \{s_i\}_{i=1}^{n-1}$, the \textit{point} diagrams $ \{p_i\}_{i=1}^{n}$, and the \textit{bridge} diagrams $\{b_i\}_{i=1}^{n-1}$, shown below:
 \begin{figure}[H]
     \centering
     \includegraphics[width=\linewidth]{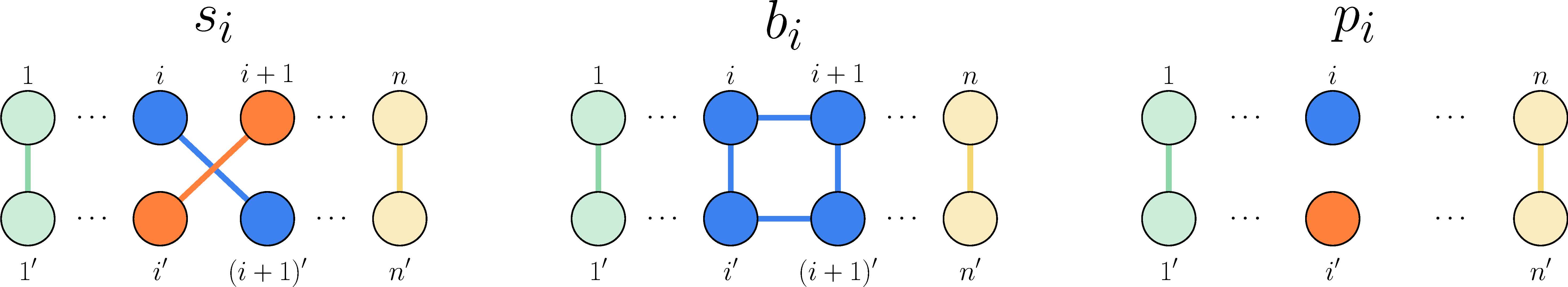}
     \label{fig:partition_algebra_generators}
 \end{figure}
\end{lemma}
\begin{lemma}[\cite{halverson2004partitionalgebras}]
    $|P_n(d)|$ is the $2n$-th Bell number, which is upper-bounded by $2^{O(n \log n)}$.
\end{lemma}

The first subalgebra of $P_n(d)$ we consider is the related \textit{half partition algebra}, denoted $P_{n - \frac12}(d)$.\ $P_{n - \frac12}(d)$ is the subalgebra of $P_n(d)$ spanned by all partition diagrams where $n$ and $n^\prime$ are in the same block. Equivalently, $P_{n - \frac12}(d)$ is the subalgebra of $P_n(d)$ generated by $ \{s_i\}_{i=1}^{n-2}$, $ \{p_i\}_{i=1}^{n -1 }$, and $\{b_i\}_{i=1}^{n-1}$. 
In $P_{n-\frac12}(d)$, the block containing $n$ and $n^\prime$ is omitted when computing the propagating number of a diagram in $P_{n - \frac12}(d)$.

\paragraph{The Brauer Algebra}
\begin{definition}
The \emph{Brauer algebra} $B_n(d)$ is the subalgebra of $P_n(d)$ spanned by partition diagrams in which every block has size $2$. That is, every diagram is a perfect matching among the $2n$ vertices. It is straightforward to check that $B_n(d)$ is indeed closed under diagram multiplication. $B_n(d)$ has the following generating set~\cite{Nazarov1996Brauer}: 
\begin{equation}
    \label{eq:brauer_generators}
    \{s_i\}_{i=1}^{n-1}\ \cup\ \{e_i\}_{i=1}^{n-1}
\end{equation}
where $e_i = b_ip_ip_{i+1}b_i$ is a \textit{contraction diagram}:
\begin{figure}[H]
    \centering
    \includegraphics[width=0.28\linewidth]{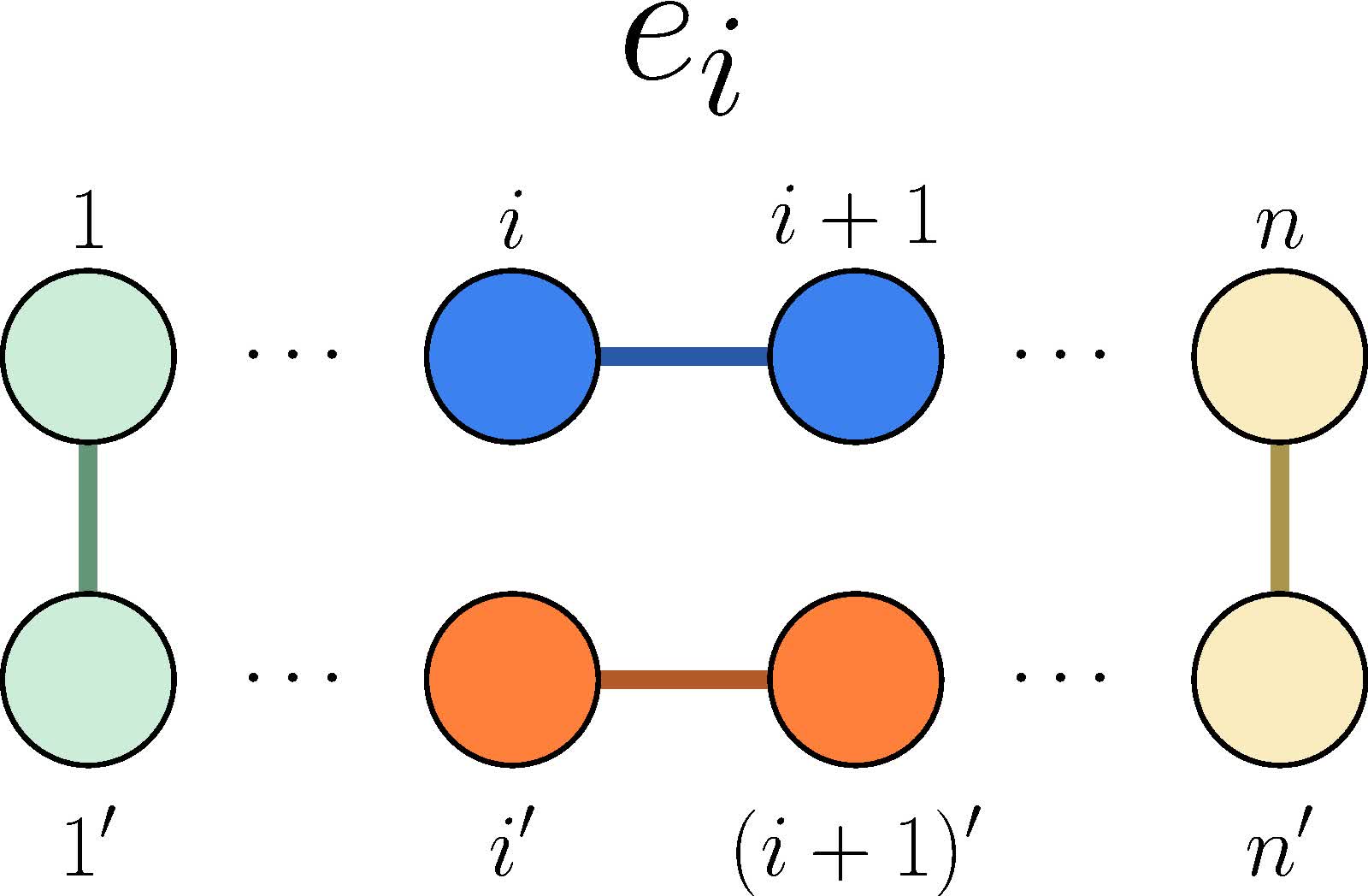}
    \label{fig:contraction}
\end{figure}
\begin{lemma}[\cite{Brauer1937AlgebrasSemisimpleContinuousGroups}]
    $|B_n(d)| = (2k - 1)!! = (2k - 1)(2k - 3)(2k - 5) \dots$
\end{lemma}

\end{definition}

\paragraph{The Walled Brauer Algebra}

\begin{definition}
For $r + s = n$, the \emph{walled Brauer algebra} $B_{r,s}(d) \subseteq B_n(d)\subseteq P_n(d)$ is spanned by all Brauer diagrams with the following additional constraint: blocks with both vertices in the first $r$ columns or both vertices in the last $s$ columns must be in opposite rows. Otherwise, the vertices must be in the same row (see the diagram on the next page). $B_{r,s}(d)$ is generated by $ \{s_i\}_{i\ne r} \cup\{e_r\}$, and $|B_{r,s}(d)| = (r + s)!$ (\cite{cox2007blockswalledbraueralgebra}).
\begin{figure}[H]
    \centering\includegraphics[width=0.35\linewidth]{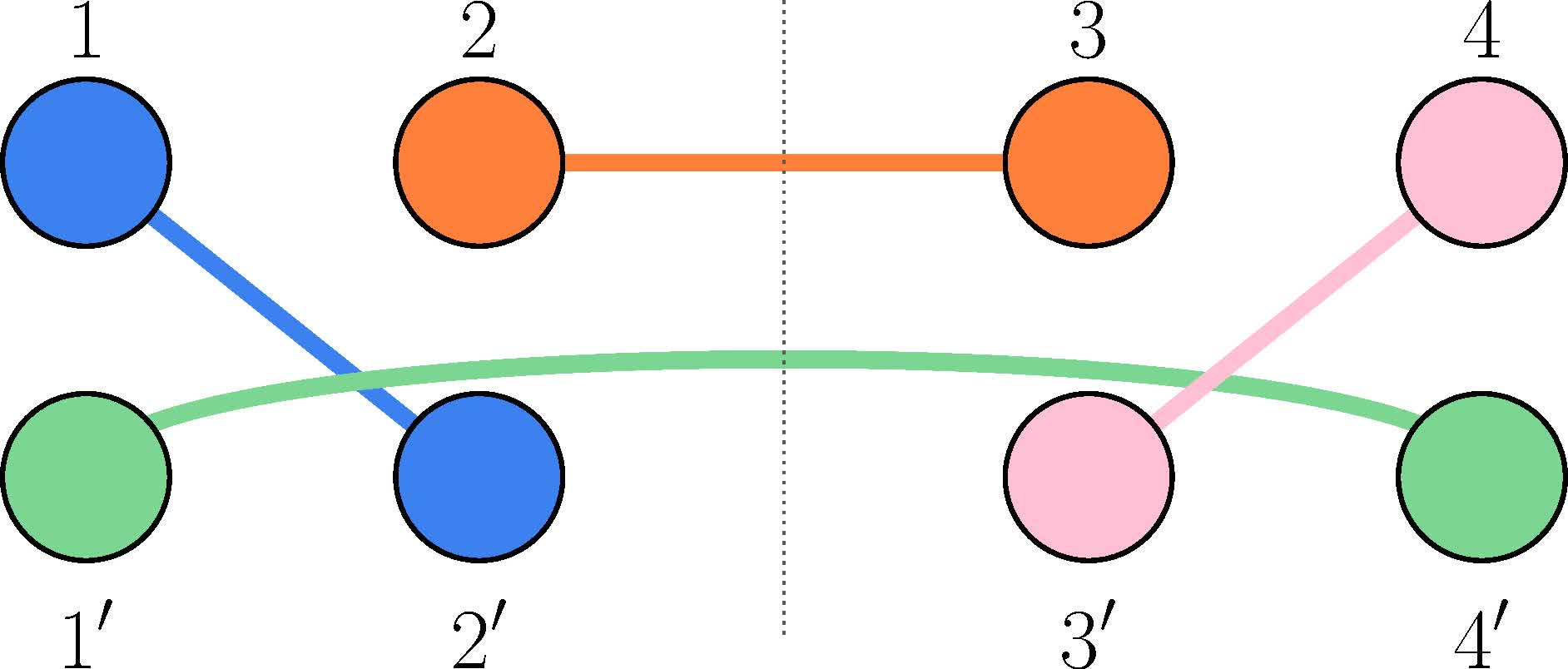}
    \caption{An example of a walled Brauer diagram in $B_{2,2}(d)$.}
    \label{fig:walled_brauer}
\end{figure}
\noindent 
\end{definition}
\noindent In addition to $s_i$ and $e_r$, we will also sometimes use $f_i$ to denote the contraction between column $i$ and column $s + i$.

\paragraph{The Symmetric Group Algebra}
\begin{definition}
The \emph{symmetric group algebra} $\mathbb{C}[S_n] \subseteq B_n(d)\subseteq P_n(d)$ is spanned by all partition diagrams that are perfect matchings between the $n$ vertices in the bottom row and the $n$ vertices in the top row.
In other words, each component is of size two, and contains a single bottom vertex and a single top vertex. These diagrams, of course, correspond to permutations on $n$ items. 
$\mathbb{C}[S_n]$ is generated by the transpositions
$
    \{s_i\}_{i= 1}^{n-1}
$.
\end{definition}

In this work, we focus on the partition algebra, Brauer algebra, and walled Brauer algebra. However, we believe our techniques could be generalized to obtain efficient quantum Fourier transforms for other diagram algebras, such as the Temperley-Lieb algebra $TL_n(d)$ or the planar partition algebra $PP_n(d)$. 

\subsection{Irreps of Diagram Algebras}
\label{sec:young_diagrams}
Since the goal of this work is to efficiently perform the Fourier transform on these diagram algebras, we will need to study their irreducible representations.
 All the diagram algebras we study in this paper have irreps indexed by \textit{Young diagrams}, which are integer partitions traditionally visualized using non-increasing, left-justified rows. Below is a Young diagram corresponding to the partition $(4,2,2,1)$ of $n = 9$: 
\begin{figure}[H]
    \centering
    \includegraphics[width=0.15\linewidth]{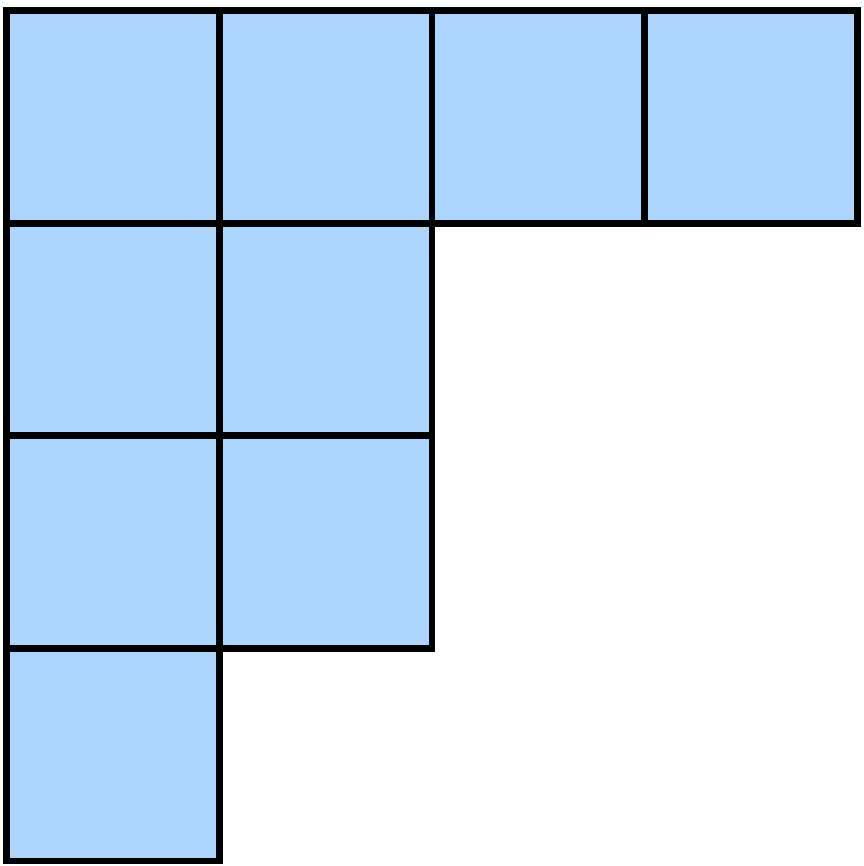}
    \label{fig:young_diagram}
\end{figure}
In this work, we will slightly abuse notation and use $\lambda$ to refer to both a Young diagram and the irrep it indexes, depending on the context. We also use $|\lambda|$ to denote the total number of boxes in $\lambda$, and $\lambda_i$ to denote the number of boxes in the $i$th row. Boxes in $\lambda$ can be given a label $(i, j)$ specifying the row and column of a box, and the \textit{content} of a box $b$, denoted $\cont(b)$, is $j - i$. 
The \textit{hook length}, $h(i,j)$, of a box $(i,j)$ is given by the number of boxes $(k, l)$ with $k = i$ and $l \ge j$ or $k \ge i$ and $l = j$. For example, the hook lengths of each box in the Young diagram above are as follows:  
\begin{figure}[H]
    \centering
    \includegraphics[width=0.15\linewidth]{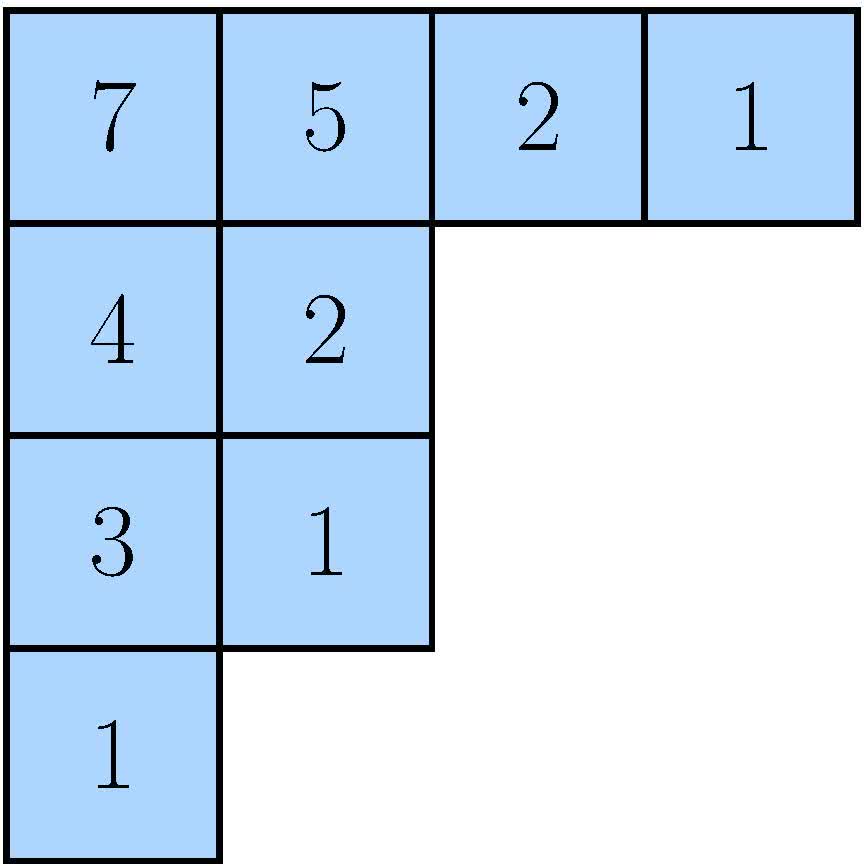}
    \label{fig:hook_length}
\end{figure}
\begin{definition}
    \label{def:hook_length_formula}
    Given a Young diagram $\lambda$, the \textit{hook length formula} is 
    \begin{equation}
        f^\lambda = \frac{|\lambda|!}{\prod_{(i,j) \in \lambda} h(i,j)}
    \end{equation}
\end{definition}

$f^\lambda$ plays an important role in the representation theory of diagram algebras, as it relates to the dimension of the irrep indexed by $\lambda$. In the symmetric group algebra $\mathbb{C}[S_n]$, which is the diagram algebra spanned by all permutation diagrams, $d_\lambda = f^\lambda$. Using the fact that the irreps of $\mathbb{C}[S_n]$ are indexed by Young diagrams with exactly $n$ boxes, one can derive the following well known sum-of-squares identity~\cite{Sagan2001}: 
\begin{equation}
        \label{eq:dim_identity_for_symmetric_group}
        \sum_{\lambda: |\lambda| = n} {(f^\lambda)^2} = |S_n| = n!
\end{equation}
Now, we give the irreps for each of the diagram algebras $P_n(d), P_{n - \frac12}(d), B_n(d)$, and $B_{r,s}(d)$.%
\footnote{
    As noted, for certain small values of $d$, these algebras are not semisimple, and thus do not decompose nicely into irreps. However, since we will take $d \sim 2^{\poly(n)}$, this will not be an issue in this regime. 
}
In the following theorems, assume $d \in \mathbb{Z}_{\ge 0}$. 
\begin{theorem}[\cite{halverson2004partitionalgebras}]
   $P_n(d)$ is semisimple whenever $d \not\in \{0, 1, \dots ,2n - 2\}$. In this regime,
   \begin{equation}
      \label{eq:irreps_of_parition_algebra}
      \widehat{P_n(d)} = \{\lambda:\ 0 \le |\lambda|\le n\},\qquad
      d_\lambda =f^\lambda \sum_{l = |\lambda|}^n \stirlingii{n}{l}\binom{l}{|\lambda|}.
   \end{equation}
\end{theorem}
Here, $\binom{l}{|\lambda|}$ is a binomial coefficient, and $\stirlingii{n}{l}$ is a Stirling number of the second kind.

\begin{theorem}[\cite{halverson2004partitionalgebras}]
   $P_{n - \frac12}(d)$ is semisimple whenever $d \not\in \{0, 1, \dots ,2n - 2\}$. In this regime,
   \begin{equation}
      \label{eq:irreps_of_half_parition_algebra}
      \widehat{P_{n - \frac12}(d)}
      = \{\lambda:\ 0\le |\lambda|\le n-1\},\qquad
      d_\lambda
      = f^\lambda \sum_{l=|\lambda|}^{\,n-1}
      \stirlingii{n}{\,l+1\,}\binom{l}{|\lambda|}.
   \end{equation}
\end{theorem}

\begin{theorem}[\cite{Nazarov1996Brauer}]
   $B_n(d)$ is semisimple whenever $d \not\in \{0, 1, 2, \dots, n\}$. In this regime,
   \begin{equation}
      \label{eq:irreps_of_brauer_algebra}
      \widehat{B_n(d)}
      = \{\lambda: |\lambda| \le n,\; |\lambda| \equiv n \!\!\!\!\!\mod 2\},\qquad
      d_\lambda = f^\lambda \binom{n}{|\lambda|}(n - |\lambda| - 1)!!.
   \end{equation}
\end{theorem}

\begin{theorem}[\cite{cox2007blockswalledbraueralgebra}]
   $B_{r,s}(d)$ is semisimple whenever $d \not\in \{0, 1, 2, \dots, n - 2\}$. In this regime, the irreps are parametrized by a pair of Young diagrams $(\lambda, \mu)$:
   \begin{equation}
      \widehat{B_{r,s}(d)}
      = \{(\lambda, \mu): 0 \le k \le \min(r, s),\; |\lambda| = r-k,\; |\mu| = s-k\}.
   \end{equation}
   with
   \begin{equation}
      \label{eq:irrep_dim_of_walleD_brauer_algebra}
      d_{(\lambda, \mu)} = k!\binom{r}{k}\binom{s}{k}f^\lambda f^\mu 
   \end{equation}
\end{theorem}
For ease of notation, we will sometimes write $|(\lambda, \mu)| = |\lambda| + |\mu|$.
\subsubsection{Diagram Algebra Conventions}
\label{sec:diagram_algebra_conventions}
For a diagram algebra $A$, a natural choice of basis is the diagram basis, i.e., each partition diagram taken with unit coefficient. If $A$ is the Brauer or walled Brauer algebra, we will show that $A$ is $\delta$-nice with respect to the diagram basis, where $\delta$ vanishes as $d \rightarrow \infty$.\ However, for the partition and half partition algebra, we must rescale each diagram $D$ by a factor depending on its number of connected components $\cc(D)$ in order to recover $\delta$-niceness.\footnote{As we will see in~\cref{sec:diagram_algebras_nice}, this reflects the fact that $\cc(D)=n$ for every Brauer diagram, whereas $\cc(D)$ varies across general partition diagrams.}
\begin{equation}
    \label{eq:basis_for_a_diagram_algebra}
    \mathcal{B}(A) = \{d^{\frac{n - \cc(D)}{2}}D\}_D
\end{equation}
Going forward, we take this ``scaled diagram basis'' to be the computational basis $\mathcal{B}(A)$. We use $\mathcal{D}(A)$ to denote the set of unscaled diagrams, i.e. with coefficient $1$. 

\subsection{Properties of Diagram Algebra Irreps}
\label{sec:properties_of_irreps}
In this section, we give two useful theorems of the irreducible representations of diagram algebras. First, we give a correspondence between the number of diagrams with $\pn(D) = k$, and the total dimension of $k$-box irreps: 
\begin{theorem}
\label{thm:prop_number_k_equals_irrep_dimension_at_level_k}
    For all $A \in \{P_n(d), P_{n - \frac12}(d),  B_n(d), B_{r,s}(d)\}$, and $0 \le k \le n$, 
    \begin{equation}
        |\{D \in \mathcal{D}(A): \pn(D) = k\}| = \sum_{\rho \in \wh{A}: |\rho| = k} d_\rho^2
    \end{equation}
\end{theorem}
We give the proof of~\cref{thm:prop_number_k_equals_irrep_dimension_at_level_k} in~\cref{app:properties_of_irreps}. Given~\cref{thm:prop_number_k_equals_irrep_dimension_at_level_k}, one might conjecture that the Fourier transform maps a diagram $D$ with $\pn(D) = k$ to a state supported on the subspace of $\ket{\lambda, i, j}$ states with $|\lambda| = k$. In turns out that this is indeed the case in the large-$d$ limit (we give a proof in \cref{sec:propagating_number_to_box_duality}). A key component of the proof is the following theorem, whose proof we also defer to~\cref{app:properties_of_irreps}. 
\begin{theorem}
    \label{thm:zero_irrep_matrix}
    Assume $A \in \{P_n(d), P_{n - \frac12}(d),  B_n(d), B_{r,s}(d)\}$, and $\lambda \in \wh{A}$. For any partition diagram $D \in \mathcal{D}(A)$ with $\pn(D) < |\lambda|$, $\lambda(D) = 0$.
\end{theorem}
\subsection{Diagram Algebras are \texorpdfstring{$\delta$}{delta}-nice}
\label{sec:diagram_algebras_nice}
In this section, we will show that as $d$ gets large, the normalized Fourier basis converges towards an orthonormal basis with respect to $\braket{\cdot, \cdot}_0$. In other words, the algebras become $\delta$-nice. To prove this result, it will be useful to first equip a diagram algebra $A$ with an involution operator $\op$:
\subsubsection{An Involution on Diagram Algebras}
\begin{definition}
    An \textit{involutive $\mathbb{C}$-algebra} $A$ is an algebra equipped with an operator $\op: A \rightarrow A$ satisfying the following properties: 
    \begin{itemize}
        \item $(a^{\op})^{\op} = a$
        \item $(ab)^{\op} = (b^{\op} a^{\op})$
        \item $(\alpha a+ \beta b)^{\op} = \overline{\alpha}a^{\op} + \overline{\beta}b^{\op}$
    \end{itemize}
\end{definition}
Typically, the involution operator $\op$ is denoted by $*$. In this work, we use $*$ for a dual basis element, so we use $\op$ to avoid confusion. Some examples of involutive algebras include the complex numbers with $\op$ given by complex conjugation, and a matrix algebra with $\op$ given by the adjoint operation.\ A diagram algebra is also an involutive algebra, with $\op$ given by exchanging the two rows of the diagram. 

Given a involutive $\mathbb{C}$-algebra and a trace form $\tau$, one can define an inner product $\braket{x, y}_{\tau}^{\op} = \tau(x^{\op} y)$. That $\braket{\cdot, \cdot}_{\tau}^{\op}$ is in fact an inner product follows from the definition of an involutive operator. Note that this is distinct from the bilinear form $\braket{x, y}_{\tau} = \tau(xy)$ introduced in \cref{sec:trace_forms_and_dual_bases}.

\subsubsection{The Schur Inner Product}
\label{sec:schur_rep}
In addition to the regular representation, another particularly important representation of a diagram algebra $A$ is the \textit{Schur} representation $\mathbf{S}$. The Schur representation generalizes the permutation representation of $\mathbb{C}[S_n]$, and is a $d^n$-dimensional representation defined as follows:  
\begin{equation}
    \mathbf{S}(D)
    =
    \sum_{\mathbf{x},\mathbf{y}\in [d]^n}
    \left(
        \prod_{\substack{u,v\in \{1,\dots,n,1',\dots,n'\}\\ \text{$u$, $v$ in the same block}}}
        \delta_{z_u,z_v}
    \right)
    \ket{\mathbf{x}}\bra{\mathbf{y}},
\end{equation}
with 
\begin{equation}
    z_r
    =
    \begin{cases}
        x_r 
        & 
        r \in  \{1,\dots,n\}
        \\
        y_r
        & 
        r \in  \{1',\dots,n'\}
    \end{cases}
\end{equation}
In other words, $\mathbf{S}(D)_{\mathbf{x},\mathbf{y}} = 1$ if the labeling of vertices given by $\mathbf{x},\mathbf{y}$ is constant on every connected component of $D$, and $0$ otherwise.  
\begin{figure}[H]
    \centering
    \includegraphics[width=0.9\linewidth]{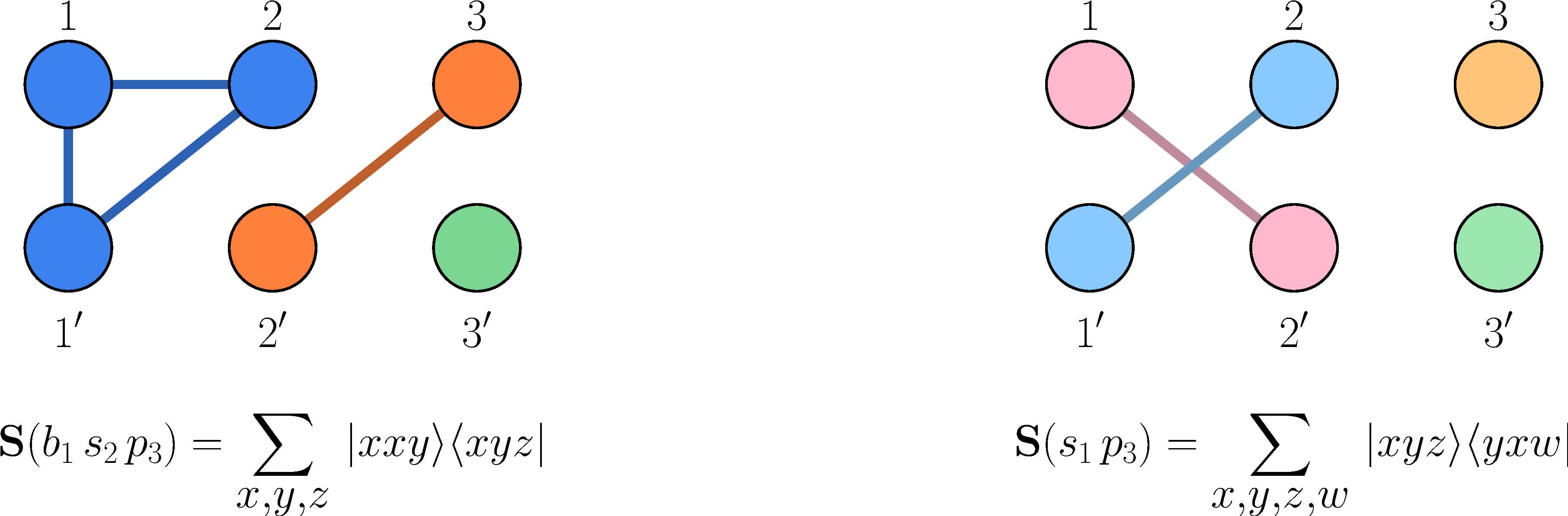}
    \caption{Examples of the Schur representation for different partition diagrams.}
    \label{fig:schur_reps}
\end{figure}
\noindent The Schur representation gives rise to the Schur inner product $\braket{x, y}_{\mathbf{S}}^{\op}$, which satisfies
\begin{equation}
    \label{eq:schur_is_antihomomorphism}
    \braket{a, b}_{\mathbf{S}}^{\op} = \tau_S(a^{\op} b) = \tr(\mathbf{S}(a^{\op})\mathbf{S}(b)) = \tr(\mathbf{S}(a)^\dagger\mathbf{S}(b))
\end{equation}
The last equality holds because $\op$ is an anti-homomorphism, i.e. $\mathbf{S}(a^{\op}) = \mathbf{S}(a)^\dagger$.\ This can be observed from definition of the Schur representation, where exchanging the rows of the diagram corresponds to exchanging the bras and kets in the corresponding matrix representation. The next important fact about Schur representation is that the induced inner product orthogonalizes the Fourier basis: 
\begin{lemma}
    \label{lem:fourier_states_orthogonal_under_schur}
    $\braket{E_{ij}^\rho, E_{kl}^\sigma}_{\mathbf{S}}^{\op} = m_\rho \delta_{\rho\sigma} \delta_{ik}\delta_{jl}$, where $m_\rho$ denotes the multiplicity of $\rho$ in $\mathbf{S}(\cdot)$.
\end{lemma}
\begin{proof}
Decompose the Schur representation as
\begin{equation}
\mathbf S(a) = U_{\mathbf{sch}}\bigoplus_{\rho\in\widehat A}\left(I_{m_\rho}\otimes \rho(a)\right)U_{\mathbf{sch}}^\dagger
\end{equation}
where $U_{\mathbf{sch}}$ is the Schur transform.  
Since Fourier states are by definition the preimages of irrep matrix units, this implies that $\mathbf{S}(E_{ij}^{\rho}) = U_{\mathbf{sch}}\left(I_{m_\rho} \otimes \ket{i}\bra{j}_\rho \right) U_{\mathbf{sch}}^\dagger$. Therefore,
\begin{equation}
\braket{E_{ij}^{\rho},E_{kl}^{\sigma}}_{\mathbf{S}}^{\op}
=\tr\Big(U_{\mathbf{sch}}\big(I_{m_\rho}\otimes |i\rangle\!\langle j|\big)^\dagger
           \big(I_{m_\sigma}\otimes |k\rangle\!\langle l|\big)U_{\mathbf{sch}}^\dagger\Big) =\delta_{\rho\sigma}\,
\tr\Big(I_{m_\rho}\otimes \big(|j\rangle\!\langle i|\big)\big(|k\rangle\!\langle l|\big)\Big) = m_\rho \delta_{\rho\sigma} \delta_{ik}\delta_{jl}
\end{equation}
Where we used the fact $U_{\mathbf{sch}}$ is a unitary operator~\cite{bacon2005quantumschurtransformi}.
\end{proof}
 Conveniently, the Schur inner product can be easily computed when both arguments are diagrams: 
\begin{lemma}
    For any two partition diagrams $D, E$, $\braket{D, E}_{\mathbf{S}}^{\op} = d^{\cc(cl(D^{\op}E))}$. Here, $cl(\cdot)$ denotes the \textit{closure} of $D^{\op}E$, obtained by identifying node $k$ with $k^\prime$ for all $k \in [n]$.
\end{lemma}
\begin{proof}
Our goal is to compute
\begin{align}
    \tr(\mathbf{S}(D)^\dagger\mathbf{S}(E))   &= \sum_{x,y} (\mathbf{S}(D))_{x,y}\,(\mathbf{S}(E))_{x,y} = \#\Big\{(x,y)\;:\;(\mathbf{S}(D))_{x,y}=1 \ \text{and}\ (\mathbf{S}(E))_{x,y}=1\Big\} 
\end{align}
Now, note that $\mathbf{S}(E)_{x,y} = 1$ iff the labeling of the $2n$ boundary vertices induced by the output index-string $x = (x_1, x_2, \dots, x_n)$ and input index-string $y = (y_1, y_2, \dots, y_n)$ is constant on each connected component of $E$. Hence, $\mathbf{S}(D)_{x,y}$ and $\mathbf{S}(E)_{x,y}$ are both equal to 1 iff, when the vertices of $D^{\op}$ and $E$ are labeled according to $x$ and $y$, all connected components in the closure of $D^{\op}E$ have a constant label. Since there are $d$ possibilities for each $x_i$ and $y_i$, there are exactly $d^{\cc(cl(D^{\op}E))}$ different $(x,y)$ pairs for which this occurs. 
\end{proof}
Another way to understand the connected components of the closure is via the \textit{join} of $D$ and $E$, denoted $D \vee E$. The join is defined as the finest partition which coarsens both $D$ and $E$, or equivalently the partition obtained by taking the transitive closure of the two equivalence relations implied by the partitions. Intuitively, this corresponds to identifying each vertex in $D$ with its corresponding vertex in $E$, which is equivalent to taking the closure of $D^{\op}E$ as a graph. Hence, we obtain the following corollary: 
\begin{corollary}
    For any diagrams $D, E$, $\braket{D, E}_{\mathbf{S}}^{\op} = d^{\cc(D \vee E)}$.
\end{corollary}
Since $D \vee D = D$, $\braket{D, D}_{\mathbf{S}}^{\op} = d^{\cc(D)}$. Hence, for any $A \in \{P_n(d), P_{n - \frac12}(d), B_n(d), B_{r,s}(d)\}$, 
\begin{equation}
    \label{eq:dn_inner_product}
    \forall a\in \mathcal{B}(A) \implies  \braket{a, a}_{\mathbf{S}}^{\op} = d^n
\end{equation}
The simplicity of~\cref{eq:dn_inner_product} suggests why the rescaling in~\cref{eq:basis_for_a_diagram_algebra} is useful. It is also necessary; one can show that $P_n(d)$ is not $\delta$-nice without rescaling the basis. 

Now, let $G_\mathbf{S}^{\op}(\mathcal{B}(A))$ denote the Gram matrix of $\braket{\cdot, \cdot}_{\mathbf{S}}^{\op}$ with respect to $\mathcal{B}(A)$. In addition to determining the diagonal entries of $G_\mathbf{S}^{\op}(\mathcal{B}(A))$, we would also like to bound the magnitudes of off-diagonal entries in any column:
\begin{lemma}
    \label{lem:off_diagonal_schur_elements_are_small}
    For any two partitions $D$ and $E$, $$2 \cdot \cc(D \vee E) \le {\cc(D) + \cc(E)}$$ with equality iff $D = E$. Hence, for $a \ne b \in \mathcal{B}(A)$, $\abs{\braket{a, b}_{\mathbf{S}}^{\op}} \le d^{n-1/2}$. 
\end{lemma}
\begin{proof}
    The bound holds because $\cc(D \vee E)$ is also at most $\cc(D)$, as the join operation can never increase the number of blocks relative to either of the starting partitions (similarly for $\cc(E)$). Thus $\cc(D \vee E) \le \min\{\cc(D), \cc(E)\}$. 
    If $D = E$, then $D \vee D = D$ and equality holds. 
    Otherwise, there exist sets $P \in D$, $Q \in E$ with $P \ne Q$, but $P \cap Q \ne \emptyset$. This also implies that there exists some element $p$ that is in one set and not the other (without loss of generality, $p \in P$ and $p \notin Q$). In the join, all vertices in $P \cup Q$ lie in the same partition, and hence $\cc(D \vee E)$ is strictly less than $\cc(E)$.
    Finally, the result follows since the number of connected components is always an integer. 
\end{proof}
We can now show that the diagram algebras are $\delta$-nice. We can do so by aggregating some of the above results: 
\begin{corollary}
    Let $\mathcal{F}(A) = \{E_{ij}^\rho/\sqrt{m_\rho}\}_{\rho, i, j}$ be the Fourier Basis rescaled by the (square root of) the Schur inner product, and $\mathcal{B}(A)$ be the computational basis as defined in~\cref{sec:diagram_algebra_conventions}. 
    Let $G_0$ and $G_{\mathbf{S}}^{\op}$ denote the Gram matrices with respect to the computational inner product and Schur inner product (as above), respectively. Then,
    \begin{equation}
        \label{eq:d_nice_corollary_one}
        G_0(\mathcal{B}(A)) = G_\mathbf{S}^{\op}(\mathcal{F}(A)) = I
    \end{equation}
    and 
    \begin{equation}
    \label{eq:d_nice_corollary_two}
        ||G_\mathbf{S}^{\op}(\mathcal{B}(A))/d^{n} - I||_{\infty} \le  d^{-1/2} \cdot |A|
    \end{equation}
    Therefore, 
    \begin{equation}
    \label{eq:d_nice_corollary_three}
        ||d^{n} \cdot G_0(\mathcal{F}(A)) - I||_{\infty} \le O(d^{-1/2} \cdot |A|)
    \end{equation}
\end{corollary}
\begin{proof}
    $G_\mathbf{S}^{\op}(\mathcal{F}(A)) = I$ follows from~\cref{lem:fourier_states_orthogonal_under_schur}, and $G_0(\mathcal{B}(A)) = I$ by definition of the computational inner product.
    \cref{eq:d_nice_corollary_two} follows because the off-diagonal entries of $G_\mathbf{S}^{\op}(\mathcal{B}(A))$ are at most $d^{n - 1/2}$ (\cref{lem:off_diagonal_schur_elements_are_small}) and the diagonal entries are $d^n$. Therefore, when dividing by $d^n$ and subtracting $I$, all of the matrix entries have absolute value at most $d^{-1/2}$, after which we can apply~\cref{lem:opnorm_bound}. Finally,~\cref{eq:d_nice_corollary_three} follows from \cref{lem:close_gram_implies_close_norms}, taking $B = \mathcal{B}(A)$, $C = \mathcal{F}(A)$, $\braket{\cdot, \cdot}_g = \braket{\cdot, \cdot}_0$, and $\braket{\cdot, \cdot}_h = \braket{\cdot, \cdot}_{\mathbf{S}}^{\op}$.
\end{proof}
\begin{corollary}
    \label{cor:norm_of_fourier_states}
    We have the following bounds:
    \begin{enumerate}
        \item
        $||E_{ij}^\rho||^2_2 \cdot \frac{d^n}{m_\rho}
        = 1 \pm O(d^{-1/2} \cdot |A|)$.
        \item
        $||E_{ij}^\rho||_2 \cdot \sqrt{\frac{d^n}{m_\rho}}
        = 1 \pm O(d^{-1/2} \cdot |A|)$.
    \end{enumerate}
\end{corollary}
\begin{proof}
By~\cref{eq:d_nice_corollary_three}, each diagonal entry of
$d^n \cdot G_0(\mathcal{F}(A))$ differs from $1$ by at most
$O(d^{-1/2} \cdot |A|)$. Hence,
\begin{equation}
    \left|
        ||E_{ij}^\rho||_2^2 \cdot \frac{d^n}{m_\rho} - 1
    \right|
    =
    \left|
        d^n \cdot  \left\langle
            \frac{E_{ij}^\rho}{\sqrt{m_\rho}},
            \frac{E_{ij}^\rho}{\sqrt{m_\rho}}
        \right\rangle_0 - 1
    \right|
    \le O(d^{-1/2} \cdot |A|).
\end{equation}
The second statement follows by taking square roots, using
$|\sqrt{x}-1| \le |x-1|$ for $x\ge 0$.
\end{proof}

\begin{corollary}
    \label{cor:d_nice_algebras}
    For any $A \in \{P_n(d), P_{n - \frac12}(d), B_n(d), B_{r,s}(d)\}$, $A$ is $O(d^{-1/2} \cdot |A|)$-nice. 
\end{corollary}
Going forward, we let $\delta \coloneqq O(d^{-1/2} \cdot |A|)$ when the algebra $A$ is clear from context. \cref{cor:d_nice_algebras} implies that the Fourier transform $\ft_A$ is unitary up to negligible error only when $d \gg \poly(|A|)$. For this reason, we will assume this regime for $d$ going forward.   
\subsection{\texorpdfstring{$\wt{\ft_A}\ket{D}$}{FT} is Concentrated on the Propagating Number}
\label{sec:propagating_number_to_box_duality}
While not strictly necessary to derive an efficient implementation of $\wt{\ft_A}$, the results of~\cref{sec:properties_of_irreps} imply an interesting invariant preserved by the Fourier transform. If $\ket{D}$ is a diagram with propagating number $r$, then after applying a Fourier transform the resulting state is almost entirely supported on $\ket{\rho, i, j}$ with $|\rho| = r$.\ A somewhat trivial example of this phenomenon is the symmetric group algebra $\mathbb{C}[S_n]$. Diagrams in $\mathcal{D}(\mathbb{C}[S_n])$ are precisely those with propagating number $n$, and all irreps of ${\mathbb{C}[S_n]}$ are $n$ box Young diagrams. \cref{thm:fourier_concentration_on_propagating_number} implies that this relationship holds for more general diagram algebras as well.
\begin{theorem}
\label{thm:fourier_concentration_on_propagating_number}
    Let $A \in \{P_n(d), P_{n - \frac12}(d), B_n(d), B_{r,s}(d)\}$ be $\delta$-nice. Let $a = d^{\frac{n - \cc(D)}{2}}D \in \mathcal{B}(A)$, and define $\ket{\wt{a}} = \wt{\ft_A}\ket{a}$.
    
    Define $\Pi_r$ to be the projection onto the subspace $S_r$, spanned by states of the form 
\begin{equation}
    \{\ket{\lambda, i, j}: |\lambda| = r\}
\end{equation}
If $D$ has propagating number $r$, then 
\begin{equation}
    \abs{\braket{\wt{a}|\Pi_r |\wt{a}}_0 - \braket{\wt{a}, \wt{a}}_0} = O(\poly(|A|) \cdot \delta)
\end{equation}
\end{theorem}
\noindent In other words, almost all of $\wt{\ft_A}\ket{a}$ is concentrated on irreps with a number of boxes equal to $\pn(D)$.\ We give the proof of~\cref{thm:fourier_concentration_on_propagating_number} in~\cref{sec:proof_of_invariant}.
\begin{figure}[H]
    \centering
    \includegraphics[width=\linewidth]{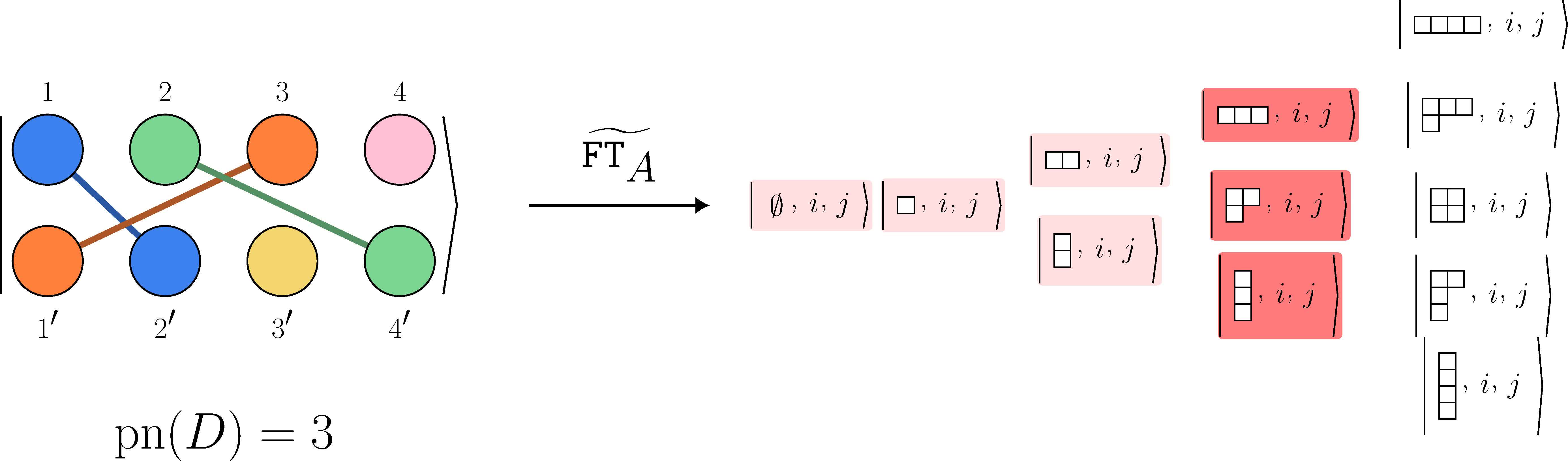}
    \caption{The Fourier transform maps diagrams with propagating number $r$ to irreps with at most $r$ boxes (\cref{thm:zero_irrep_matrix}). By~\cref{thm:fourier_concentration_on_propagating_number}, all but $O(\poly(|A|) \cdot \delta)$ of the total magnitude is concentrated on irreps with exactly $r$ boxes.}
    \label{fig:prop_invariant}
\end{figure}
\section{The Subalgebra Adapted Fourier Basis}
\label{sec:subalgebra_adapteD_basis}
When performing the Fourier transform, it will be useful to choose a basis for the irrep spaces $V^\rho$ with some structure. Concretely, we will choose a basis which block-diagonalizes the action of $\rho(b)$, for all $b$ in a specially chosen subalgebra $B$. The resulting basis is known as a \textit{subalgebra adapted basis}. Such subalgebra adapted bases are standard in efficient implementations of Fourier transforms and Schur transforms~\cite{DiaconisRockmore1990, moore2003genericquantumfouriertransforms,bacon2005quantumschurtransformi, maslen2016efficientcomputationfouriertransforms}.   

\subsection{Restricted Representations and Branching Rules}
In order to define subalgebra adapted bases, we first define the notion of a restricted representation:
\begin{definition}
    Suppose that $A$ is a semisimple algebra which has a semisimple subalgebra $B$.\ Given a representation $\rho: A \rightarrow \text{End}(V)$, the \textit{restricted representation} $\rho|_B: B \rightarrow \text{End}(V)$ is obtained by restricting the domain of $\rho$ to $B$. 
\end{definition}
Note that $\rho|_B$ will in general not be an irrep of $B$, even if $\rho$ is an irrep of $A$. Instead, it will be isomorphic to a direct sum of irreps of $B$, $\rho|_B \cong \bigoplus_{\sigma \in \wh{B}} m_\sigma\sigma$, or equivalently
\begin{equation}
    \label{eq:branching_rules}
     V^{\rho}\downarrow_B\; \cong \bigoplus_{\sigma \in \text{Res}(\rho)_B^A} (I_{m_\sigma} \otimes V^\sigma)
\end{equation}
Where $\text{Res}(\rho)_B^A$ denotes the set of all irreps $\sigma \in \wh{B}$ with $m_\sigma > 0$. These ``branching rules'' have particularly nice forms for diagram algebras: 
\begin{theorem}
 \label{thm:branching_rules}
Write $\mu=\lambda+\Box$ if $\mu$ is obtained from $\lambda$ by adding one box, and
$\mu=\lambda-\Box$ if $\mu$ is obtained from $\lambda$ by removing one box.
The branching rules for $P_n(d)$, $P_{n-\frac12}(d)$, $B_n(d)$, and $B_{r,s}(d)$ are as follows:

\paragraph{The Partition Algebra, $P_n(d)\supset P_{n-\frac12}(d)$ (\normalfont \cite[Theorem 2.24]{halverson2004partitionalgebras}).} 
\begin{equation}
V^\lambda\downarrow_{P_{n-\frac12}(d)}
\;\cong\;
\bigoplus_{\substack{\mu \in \wh{P_{n - \frac12}(d)} \\\mu=\lambda-\Box}} V^\mu
\;\oplus\;
\bigoplus_{\substack{\mu \in \wh{P_{n - \frac12}(d)} \\\mu=\lambda}} V^\mu 
\end{equation}

\paragraph{The Half Partition Algebra, $P_{n-\frac12}(d)\supset P_{n-1}(d)$ (\normalfont \cite[Theorem 2.24]{halverson2004partitionalgebras}).}
\begin{equation}
V^\lambda\downarrow_{P_{n-1}(d)}
\;\cong\;
\bigoplus_{\substack{\mu \in \wh{P_{n - 1}(d)} \\\mu=\lambda}} V^\mu
\;\oplus\;
\bigoplus_{\substack{\mu \in \wh{P_{n - 1}(d)} \\\mu=\lambda + \Box}} V^\mu 
\end{equation}

\paragraph{The Brauer Algebra, $B_n(d)\supset B_{n-1}(d)$ (\normalfont \cite[Theorem~4.8]{cox2008alcovegeometrytranslationprinciple}).}
\begin{equation}
V^\lambda\downarrow_{B_{n-1}(d)}
\;\cong\;
\bigoplus_{\substack{\mu \in \wh{B_{n-1}(d)} \\\mu=\lambda-\Box}} V^\mu
\;\oplus\;
\bigoplus_{\substack{\mu \in \wh{B_{n-1}(d)} \\\mu=\lambda+\Box}} V^\mu 
\end{equation}

\paragraph{The Walled Brauer Algebra, Left Restriction: $B_{r,s}(d)\supset B_{r-1,s}(d)$
(\normalfont \cite[Theorem~2.5]{brundan2011gradingswalledbraueralgebras}).}
\begin{equation}
V^{(\lambda,\mu)}\downarrow_{B_{r-1,s}(d)}
\;\cong\;
\bigoplus_{\substack{\lambda^\prime \in \wh{B_{r-1, s}(d)} \\\lambda'=\lambda-\Box}} V^{(\lambda',\mu)}
\;\oplus\;
\bigoplus_{\substack{\mu^\prime \in \wh{B_{r-1, s}(d)} \\\mu'=\mu+\Box}} V^{(\lambda,\mu')} 
\end{equation}
\paragraph{The Walled Brauer Algebra, Right Restriction: $B_{r,s}(d)\supset B_{r,s-1}(d)$
(\normalfont \cite[Theorem~2.5]{brundan2011gradingswalledbraueralgebras}).}
\begin{equation}
V^{(\lambda,\mu)}\downarrow_{B_{r,s-1}(d)}
\;\cong\;
\bigoplus_{\substack{\lambda^\prime \in \wh{B_{r, s-1}(d)} \\\lambda'=\lambda+\Box}} V^{(\lambda',\mu)}
\;\oplus\;
\bigoplus_{\substack{\mu^\prime \in \wh{B_{r, s-1}(d)} \\\mu'=\mu-\Box}} V^{(\lambda,\mu')} 
\end{equation}
\end{theorem}
\subsection{The Subalgebra Adapted Basis}
\cref{eq:branching_rules} implies that there exists a basis for $V^\rho$ such that 
\begin{equation}
    \label{eq:subalgebra_adapted_matrices}
    \rho(b) = \bigoplus_{\sigma \in \text{Res}(\rho)_B^A} \bigoplus_{i=1}^{m_\sigma}  \sigma(b) 
\end{equation}
for all $b \in \mathcal{B}(B) = B \cap \mathcal{B}(A)$.\footnote{Implicitly, we assume that $B$ is spanned by a subset of $\mathcal{B}(A)$. For each diagram algebra we consider, the subalgebras in~\cref{eq:branching_rules} have this property.} This is known as the \textit{$B$-adapted basis} for $V^\rho$. Taking this basis for $V^\rho$ for each $\rho \in \wh{A}$ gives a basis for $\bigoplus_{\rho \in \wh{A}} V^\rho$, the \textit{$B$-adapted basis}. 

If $m_\sigma$ is equal to $1$ for all $\sigma \in \text{Res}(\rho)_B^A$, then the branching rule is said to be \textit{multiplicity-free}. All of the branching rules we consider in \cref{thm:branching_rules} are multiplicity-free. In a $B$-adapted basis with multiplicity-free branching, basis vectors of $V^\rho$ are typically indexed as $(b_\sigma \circ \rho)$, where $\rho$ is the irrep label and  $b_\sigma$ indexes basis vectors of $V^\sigma$, for some $\sigma \in \text{Res}(\rho)_B^A$. Using this notation and \cref{eq:subalgebra_adapted_matrices}, 
\begin{equation}
    \label{eq:subalgebra_adapted_element}
    \sigma(b)_{kl} = \rho(b)_{k \circ\rho, l \circ\rho}
\end{equation}
A $B$-adapted basis is a useful assumption when computing the Fourier transform. In particular, our algorithm will rely on the following relationship between the Fourier states of $B$ and the (subalgebra-adapted) Fourier states of $A$:
\begin{lemma}
    \label{thm:subalgebra_restrictions} Assume the Fourier basis elements are given with respect to a $B$-adapted basis with a multiplicity-free branching rule. For any $\sigma \in \wh{B}$, 
    \begin{equation}
        \label{eq:fourier_state_subalgebra_equality}
        E_{kl}^\sigma = \sum_{\rho: \sigma \in \text{Res}(\rho)^A_B} E^\rho_{k \circ\rho, l \circ\rho}
    \end{equation}
\end{lemma}
\noindent Although $E_{kl}^\sigma$ is a Fourier basis element in $B$, we consider it here to be an element of the larger algebra $A$ in~\cref{eq:fourier_state_subalgebra_equality}.
\begin{proof}
     Apply any irrep $\tau \in \wh{A}$ to the right hand side. By definition of the Fourier basis, the resulting matrix is exactly
     \begin{equation}
        \sum_{\rho: \sigma \in \text{Res}(\rho)^A_B} \tau\left(E^\rho_{k \circ\rho, l \circ\rho}\right)
        =
        \sum_{\rho: \sigma \in \text{Res}(\rho)^A_B} \delta_{\tau\rho} \ket{k 
        \circ\tau}\bra{l\circ\tau} 
        = 
        \begin{cases}
            \ket{k \circ\tau}\bra{l\circ\tau}
            &
            \sigma \in \text{Res}(\tau)^A_B
            \\
            0
            &
            \text{otherwise}
        \end{cases}
     \end{equation}
     on the left hand side, we obtain 
     \begin{equation}
         \tau(E^{\sigma}_{kl})_{k^\prime \circ\tau,  l^\prime \circ\tau} = d_\sigma \sum_{b \in \mathcal{B}(B)} \sigma(b^*)_{lk}\tau(b)_{k^\prime \circ\tau,  l^\prime \circ\tau} = d_\sigma \sum_{b \in \mathcal{B}(B)} \sigma(b^*)_{lk}\sigma(b)_{k^\prime l^\prime}
     \end{equation}
     where the first equality is by definition of a Fourier basis element. The second equality follows from substituting~\cref{eq:subalgebra_adapted_element}, and requires $k^\prime$ and $l^\prime$ to be basis vectors of the same irrep space $V^\sigma$, where $\sigma$ also appears in the restriction of $\tau$ to $B$. Applying Schur Orthogonality (\cref{thm:schur_orthogonality}), we conclude that 
     \begin{equation}
         \tau(E^{\sigma}_{kl}) 
         =
         \begin{cases}
            \ket{k \circ\tau}\bra{l\circ\tau}
            &
            \sigma \in \text{Res}(\tau)^A_B
            \\
            0
            &
            \text{otherwise}
        \end{cases}
     \end{equation}
    Since $\tau$ was arbitrary, it holds for all irreps, and therefore the two sides are equal.
\end{proof}
\noindent When $A$ is $\delta$-nice,~\cref{thm:subalgebra_restrictions} has the following corollary: 
  \begin{lemma}
     \label{thm:approx_norm_restriction}
    Assuming that $A$ is $\delta$-nice, and that $B$ is a multiplicity-free subalgebra of $A$,
    \begin{equation}
        ||E_{kl}^\sigma||_2^2 
        = 
        \left(1 + (K - 1) \varepsilon \right)
        \left(\sum_{\rho: \sigma \in \text{Res}(\rho)^A_B} ||E^\rho_{k\circ\rho, l\circ\rho}||_2^2\right)
    \end{equation} 
    with $|\varepsilon| \le \delta$, 
    where $K$ denotes the number of irreps $\rho \in A$ such that $\sigma \in \text{Res}(\rho)^A_B$.
\end{lemma}

\begin{proof}
    We have that, 
    \begin{align}
        ||E_{kl}^\sigma||_2^2 
        % &
        = 
        \left\langle 
        E_{kl}^\sigma,
        E_{kl}^\sigma
        \right\rangle_0 
        % \\
        &
        = 
        \left\langle 
        \sum_{\rho: \sigma \in \text{Res}(\rho)^A_B} 
        E^\rho_{k \circ\rho, l \circ\rho},
        \sum_{\tau: \sigma \in \text{Res}(\tau)^A_B} 
        E^\tau_{k \circ\tau, l \circ\tau}
        \right\rangle_0 
        \\
        &
        = 
        \sum_{
            \substack{
                \rho: \sigma \in \text{Res}(\rho)^A_B 
                \\ 
                \tau: \sigma \in \text{Res}(\tau)^A_B
            }
        } 
        \left\langle 
        E^\rho_{k \circ\rho, l \circ\rho},
        E^\tau_{k \circ\tau, l \circ\tau}
        \right\rangle_0 
        \\
        &
        = 
        \sum_{
            \rho: \sigma \in \text{Res}(\rho)^A_B 
        } 
        \left\lVert 
        E^\rho_{k \circ\rho, l \circ\rho}
        \right\rVert_2^2 
        +
        \sum_{
            \substack{
                \rho \ne \tau
                \\
                \sigma \in \text{Res}(\rho)^A_B 
                \\ 
                \sigma \in \text{Res}(\tau)^A_B
            }
        } 
        \left\langle 
        E^\rho_{k \circ\rho, l \circ\rho},
        E^\tau_{k \circ\tau, l \circ\tau}
        \right\rangle_0 
    \end{align}
    Next, we bound the absolute value of the right-most term:
    \begin{align}
         \left|\sum_{
            \substack{
                \rho \ne \tau
                \\
                \sigma \in \text{Res}(\rho)^A_B 
                \\ 
                \sigma \in \text{Res}(\tau)^A_B
            }
        } 
        \left\langle 
        E^\rho_{k \circ\rho, l \circ\rho},
        E^\tau_{k \circ\tau, l \circ\tau}
        \right\rangle_0\right| \le& \sum_{
            \substack{
                \rho \ne \tau
                \\
                \sigma \in \text{Res}(\rho)^A_B 
                \\ 
                \sigma \in \text{Res}(\tau)^A_B
            }
        } 
        \left|\left\langle 
        E^\rho_{k \circ\rho, l \circ\rho},
        E^\tau_{k \circ\tau, l \circ\tau}
        \right\rangle_0\right| \\ 
        \le
        & 
        \sum_{
            \substack{
                \rho \ne \tau
                \\
                \sigma \in \text{Res}(\rho)^A_B 
                \\ 
                \sigma \in \text{Res}(\tau)^A_B
            }
        }  \delta \cdot  ||E^\rho_{k\circ\rho, l\circ\rho}||_2 ||E^\tau_{k\circ\tau, l\circ\tau}||_2
    \end{align}
    First, we used the triangle inequality, followed by the $\delta$-nice assumption (\cref{eq:d_nice_equation}), which bounds the inner product of the (normalized) Fourier basis elements. Finally, 
    for any $K$ real numbers $\{x_i\}_i$, $\sum_{i \ne j} x_ix_j \le (K - 1) \sum_{i} x_i^2$, so we can obtain the upper bound 
    \begin{equation}
        (K-1)\delta \cdot \sum_{\rho: \sigma \in \text{Res}(\rho)^A_B} ||E^\rho_{ k\circ\rho, l\circ\rho}||_2^2 
    \end{equation}
\end{proof}
\subsection{Subalgebra Adapted Chains}
\label{sec:subalgebra_chains}
While a $B$-adapted basis only gives a basis for $\{V^\rho\}_{\rho \in \wh{A}}$ in terms of bases for $\{V^\sigma\}_{\sigma \in \wh{B}}$, the idea can be applied inductively to a chain of semisimple subalgebras:
\begin{equation}
    \label{eq:subalgebra_chain}
    \mathbb{C} =  A_0 \subseteq A_1 \subseteq A_2 \subseteq \dots \subseteq A_{n-1} \subseteq A_n = A
\end{equation}
 In this case, we obtain a basis for $\bigoplus_{\rho \in \wh{A}} V^\rho$ indexed by \textit{Bratteli paths} 
\begin{equation}
   P = (\emptyset = \rho_0 \circ \rho_1 \circ \rho_2 \circ \dots \circ \rho_{n-1} \circ \rho_n =\rho) 
\end{equation}
where $\rho_i$ is an irrep of $A_i$, and $\rho_{i-1} \in \text{Res}(\rho_i)^{A_i}_{A_{i-1}}$.\ This is known as a \textit{subalgebra adapted basis} or \textit{Gelfand-Tsetlin basis} with respect to the chain $A_0, A_1, \dots, A_{n-1}, A_n$ \cite{OkounkovVershik1996}. Note that $\rho_0 = \emptyset$ without loss of generality, since a one-dimensional algebra has a trivial Fourier transform. 

A particularly nice way to visualize a subalgebra adapted chain is with a \textit{Bratteli Diagram}, which is a leveled graph where vertices at level $i$ index irreps of $A_i$. An edge connects $\rho_{i-1}$ and $\rho_{i}$ iff $\rho_{i-1} \in \text{Res}(\rho_i)^{A_i}_{A_{i-1}}$. For example, below is the Bratteli diagram of $B_{3,2}(d)$, with respect to the subalgebra chain in~\cref{eq:walleD_brauer_chain}: 
\begin{figure}[H]
    \centering
    \includegraphics[width=0.8\linewidth]{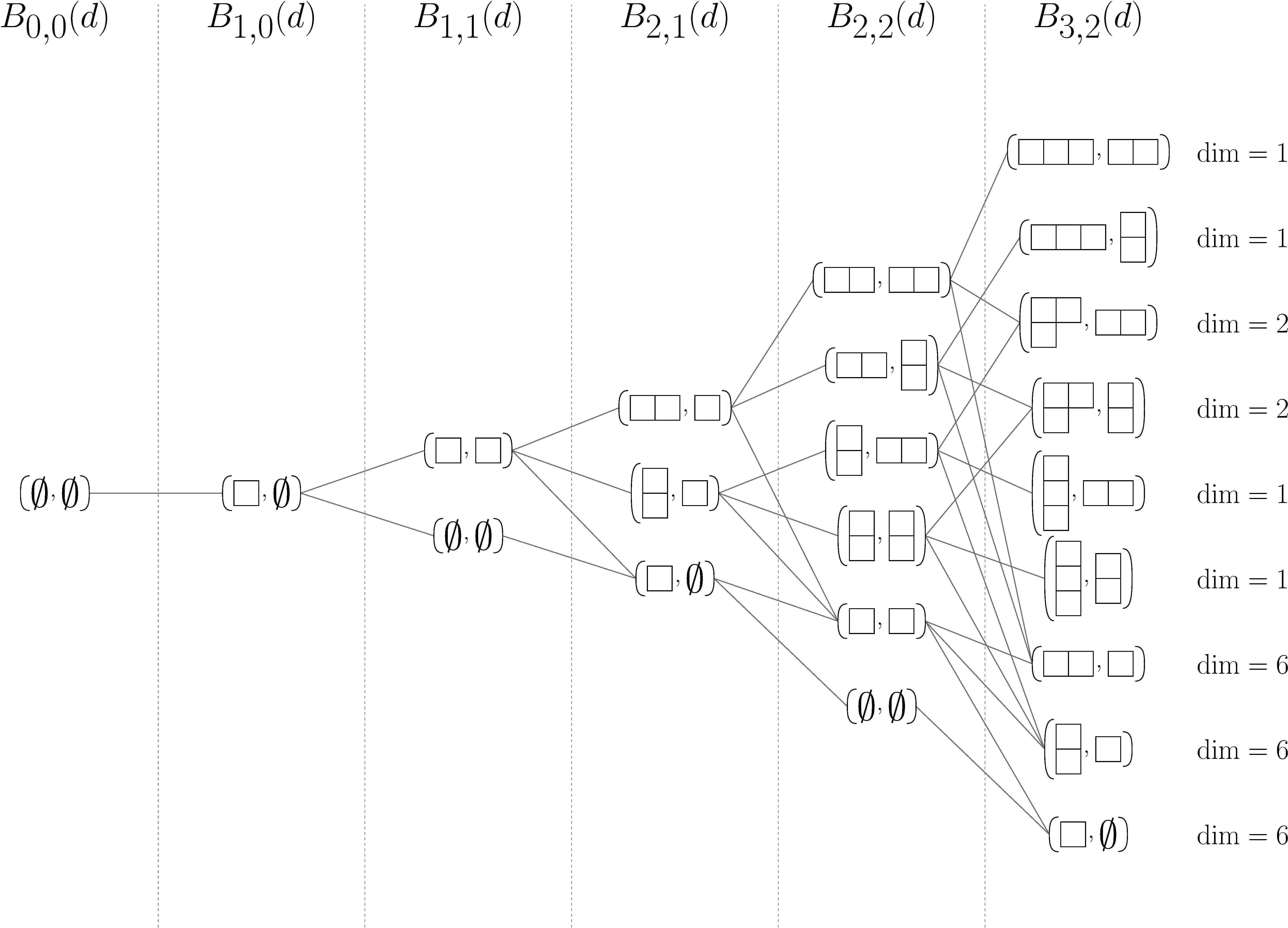}
    \caption{The Bratteli diagram of $B_{3,2}(d)$. The dimension of an irrep $\lambda \in \wh{B_{3,2}(d)}$ is equivalent to the number of paths from the root $(\emptyset, \emptyset) \in \wh{B_{0,0}(d)}$ to the label $\lambda$. Note that the sum of the squares of the dimensions is $120 = |B_{3,2}(d)|$.}
    \label{fig:Bratteli_diagram}
\end{figure}
We will sometimes refer to the subalgebra adapted basis vectors as \textit{Bratteli paths}, and use the letters $P$, $Q$, $R$, $\dots$ to denote such paths. For a Bratteli path $P$, we write $P(i)$ to denote the irrep in the chain corresponding to $A_i$. We also use $P_{\le k}$ to denote the truncated path 
\begin{equation}
    (\rho_0, \rho_1, \dots, \rho_{k-1}, \rho_k)
\end{equation}
\noindent $P_n(d)$ and $B_n(d)$ have natural subalgebra adapted bases with respect to the following chains:
\begin{equation}
    \mathbb{C} \cong P_0(d) \subseteq P_{\frac12}(d) \subseteq P_1(d) \subseteq \dots \subseteq P_{n - \frac12}(d) \subseteq P_n(d), \;\; \mathbb{C} \cong B_0(d) \subseteq B_1(d) \subseteq \dots \subseteq B_n(d)
\end{equation}
Informally, a Bratteli path in $P_n(d)$ has length $2n+1$: on each half-step $P(2i)\to P(2i+1)$ the shape either stays put or loses a box, and on the subsequent step $P(2i+1)\to P(2i+2)$ it either stays put or gains a box. A Bratteli path in $B_n(d)$ has length $n+1$, and at each step the shape changes by adding or removing a single box.

For the walled Brauer algebra, we will consider the subalgebra chain which alternates adding left and right columns, until all left columns have been added. We assume without loss of generality that $r \le s$:
\begin{equation}
    \label{eq:walleD_brauer_chain}
    B_{0,0}(d) \subseteq B_{1, 0}(d) \subseteq B_{1, 1}(d) \subseteq \dots \subseteq B_{r,r}(d) \subseteq B_{r, r + 1}(d) \subseteq \dots\subseteq B_{r,s}(d) 
\end{equation}
When moving from $B_{r,s}(d)$ to $B_{r + 1, s}(d)$, a box is either added to the left irrep or removed from the right irrep, and vice versa for $B_{r,s}(d)$ to $B_{r, s + 1}(d)$.

\subsection{The Orthogonal Form}
\label{sec:main_body_orthogonal_form}
A subalgebra adapted basis is only defined up to rescaling each path $P$ by some constant $\alpha_P$, as~\cref{eq:subalgebra_adapted_matrices} will still be satisfied. With a clever choice of prefactors $\alpha_P$, one can show (\cite{Lehrer1996}) that for any irrep $\rho \in \wh{A}$ and diagram $D$, when expressed in the basis $\{\alpha_P\ket{P}\}_P$,
\begin{equation}
    \label{eq:orthogonal_form}
    \rho'(D) = \rho'(D^{\op})^\dagger
\end{equation}
Here, $\rho'$ is the irrep matrix for $\rho$ expressed in this rescaled basis. To find these constants, one can define a bilinear form on Bratteli paths $\braket{P, Q}_{\text{path}}$, known as the \textit{cellular bilinear form} (\cite{Lehrer1996}, Proposition 2.4), and show that 
\begin{equation}
    \braket{P|\rho(D^{\op})^\dagger|Q}_{\text{path}} = \braket{P|\rho(D)|Q}_{\text{path}} \iff \rho(D^{\op})^\dagger G = G\rho(D),\;\; \text{where } (G_{PQ}) = \braket{P, Q}_{\text{path}}
\end{equation}
For all the diagram algebras, $\braket{\cdot, \cdot}_{\text{path}}$ is \textit{orthogonal} (\cite{Enyang_2013}, Proposition 4.2 (6); \cite{Nazarov1996Brauer}, Theorem 3.12; \cite{grinko2023gelfandtsetlinbasispartiallytransposed}, Theorem 3.2) implying that $G$, the Gram matrix on Bratteli paths with respect to the cellular bilinear form, is a diagonal matrix. Taking $S$ to be the square root of $G$, we see that $S^{\dagger}\rho(D)S = S^{\dagger}\rho(D^{\op})^\dagger S$. This implies that taking $\alpha_P =\sqrt{\braket{P, P}_{\text{path}}}$ achieves \cref{eq:orthogonal_form}. This particular rescaling is often called the \textit{orthogonal form} for the subalgebra adapted basis, or just the orthogonal form for the algebra $A$. 

Orthogonal forms are known for each of the diagram algebras we consider, but the formulas for the irrep matrix elements are somewhat involved. We provide them in~\cref{sec:appendix_matrix_algebras}. In addition to~\cref{eq:orthogonal_form}, an important property of the orthogonal form is that resulting irrep matrix elements of diagram algebra generators are efficiently computable, given an irrep label $\ket{\rho}$ and Bratteli path $\ket{P}$:  
\begin{lemma}
\label{lem:efficient_computation_generator_matrix_elements_all}
For $A \in \{B_n(d), B_{r,s}(d), P_n(d)\}$, there is a quantum algorithm which, given an irrep label $\rho \in \wh{A}$ and a Bratteli path $P \in V^{\rho}$, implements the following isometry to operator norm error $O(\varepsilon \cdot |A|)$: 
\begin{equation}
\ket{s_i,\rho,P}\ \mapsto\ 
\ket{s_i,\rho,P}\otimes\!\!\bigotimes_{Q:\,\rho(s_i)_{QP}\neq 0}\!\ket{\mathrm{loc}(Q),\,\rho(s_i)_{QP}}.
\end{equation}
where $s_i$ is the swap generator (see~\cref{sec:diagram_algebras}), and $\mathrm{loc}(Q)$ succinctly encodes where the Bratteli path $Q$ differs from $P$.
\begin{enumerate}
\item \textbf{Brauer.}
If $A=B_n(d)$, then $\mathrm{loc}(Q)=Q(i)$ and the gate complexity is
\begin{equation}
\wt{O}(n \cdot (n + \log d + \log(1/\varepsilon)))
\end{equation}

\item \textbf{Walled Brauer.}
If $A=B_{r,s}(d)$, then $\mathrm{loc}(Q)=Q(i)$ and the gate complexity is also 
\begin{equation}
\wt{O}(n \cdot (n + \log d + \log(1/\varepsilon))).
\end{equation}
where $n\coloneqq r+s$.

\item \textbf{Partition.}
If $A=P_n(d)$, then $\mathrm{loc}(Q)=(Q(2i-1),Q(2i),Q(2i+1))$ and the gate complexity is
\begin{equation}
\wt{O}(n^{5/2} \cdot (n + \log d + \log(1/\varepsilon)))
\end{equation}
\end{enumerate}
In each case, $\wt{O}(\cdot)$ hides $\polylog (n)$, $\poly\log\log d$, and $\poly\log\log(1/\varepsilon)$ factors.
\end{lemma}
\noindent We defer the proof of~\cref{lem:efficient_computation_generator_matrix_elements_all} to~\cref{sec:appendix_matrix_algebras}. While~\cref{lem:efficient_computation_generator_matrix_elements_all} also holds for diagram algebra generators other than $s_i$ (the proofs will in fact require computing these as subroutines), we will only need~\cref{lem:efficient_computation_generator_matrix_elements_all} to hold for swap generators to efficiently implement the Fourier transform.
\section{An Efficient Quantum Fourier Transform for Diagram Algebras}
\label{sec:efficient_qft_for_algebras}
Equipped with the tools developed in previous sections, we now present an efficient quantum Fourier transform for the partition, Brauer, and walled Brauer algebras.\ Our algorithm uses the \textit{separation of variables} framework of Diaconis and Rockmore~\cite{DiaconisRockmore1990}, and its quantum generalization due to Beals~\cite{beals1997qft} and Moore, Rockmore, and Russell~\cite{moore2003genericquantumfouriertransforms}. 

The separation of variables technique consists of five steps. First, on an input $\ket{a}$, we factor $a$ into an element $b$ of a subalgebra and some ``transversal'' element $t$.\footnote{In our case, we will actually have both a left transversal element and right transversal element.} Second, we recursively perform a Fourier transform for the subalgebra. Third, we promote $\wt{\ft_B}\ket{b}$ to $\wt{\ft_A}\ket{b}$ with an embedding step. Fourth, we apply an irrep matrix for the transversal element to obtain $\wt{\ft_A}\ket{a}$. Finally, we use a ``sum over transversals'' procedure to perform the entire computation in superposition over different transversals. 

To generalize from group to non-group algebras, we will have to make a few modifications. Since the Fourier transform on diagram algebras is not unitary, we target an approximation: specifically, our algorithm approximately implements \(\wt{\ft_A}\) from \cref{thm:approx_qft}, which itself approximates \(\ft_A\). In the regime \(d \gg \poly(|A|)\), where \(\ft_A\) is negligibly close to unitary, our algorithm also obtains negligible operator norm error using only \(\poly(n)\) gates. 

Besides performing each step approximately as opposed to exactly, the other way our algorithm differs from the typical approach taken for a group algebra is that several steps are modified when the propagating number of the input diagram is zero (and certain nonzero cases for the partition algebra). Crucially, these additional cases allow us to bypass implementing highly non-unitary irrep matrices such as $\lambda(p_i)$ or $\lambda(e_i)$, which cannot be realized on a quantum computer.

The five subsections below detail the five steps of the algorithm. The input state is
\begin{equation}
    \ket{\psi_0} = \sum_{a \in \mathcal{B}(A)} f(a) \ket{a}
\end{equation}
where $A \in \{P_n(d), P_{n - \frac12}(d), B_n(d), B_{r,s}(d)\}$.\ Before presenting the algorithm, we briefly explain how basis elements \(a\in\mathcal{B}(A)\) are encoded into qubits. When \(A\) is the Brauer or walled Brauer algebra, \(a\) is simply a diagram \(D\), so \(\ket{a}\) can be encoded by storing \(\ket{D}\) using \(\widetilde{O}(n)\) qubits. When $A$ is the partition or half-partition algebra, \(a\) is a scalar multiple of some underlying diagram \(D\) (see \cref{eq:basis_for_a_diagram_algebra}). However, the scaling is determined by \(D\), so \(a\in\mathcal{B}(A)\) is uniquely identified by its unscaled diagram, and we represent \(\ket{a}\) by encoding that diagram \(D\) (also \(\widetilde{O}(n)\) qubits). For this reason, we will implicitly associate $a$ with its underlying diagram $D_a$. 
\begin{figure}[H]
    \centering
    \includegraphics[width=0.9\linewidth]{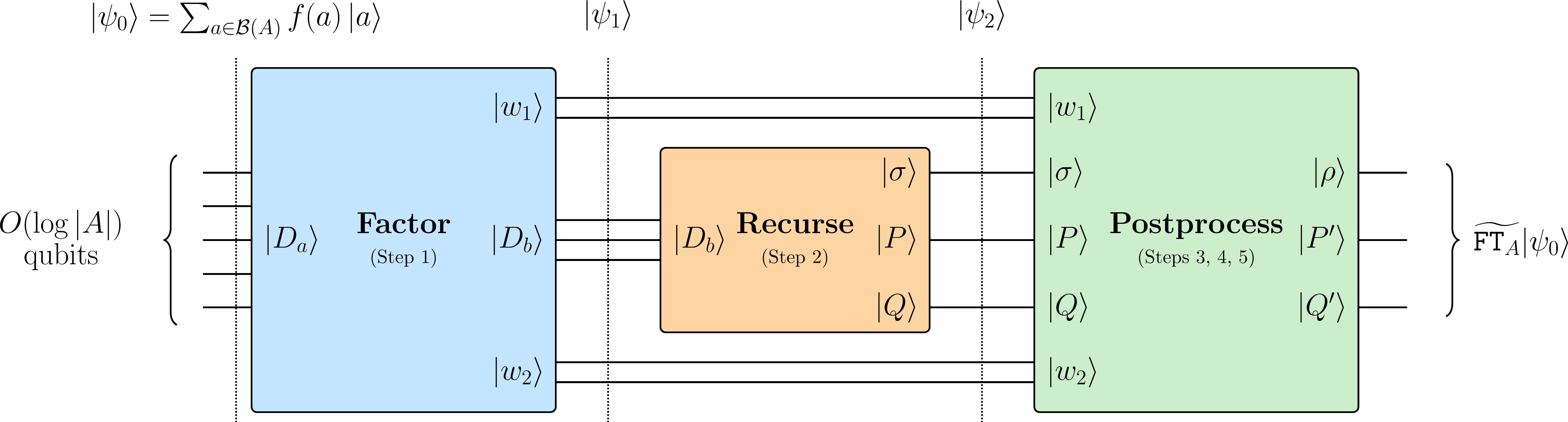}
    \caption{A circuit for the quantum Fourier transform, following the separation of variables approach. On input $\ket{D_a}$, we first factor $D_a$ as $w_1 D_b w_2$, where $D_b$ lies in the subalgebra $B$ and $w_1,w_2$ are left and right ``transversal'' elements. The precise factorization depends on the algebras $A$ and $B$ (see~\cref{sec:factor}). We then recursively apply the Fourier transform over $B$ to $\ket{D_b}$ (see~\cref{sec:recurse}). The final three steps are grouped into a single postprocessing step, which is expanded in~\cref{fig:postprocess}.}
    \label{fig:algorithm_sketch}
\end{figure}
After each step of the algorithm, we give upper bounds on the gate complexity, leaving optimizations for future work.\ We also upper bound the operator norm error incurred from the ``ideal'' version of the step, which correspond to the operators whose products would exactly equal $\wt{\ft_A}$. 
\subsection{Step 1: Factor}
\label{sec:factor}
Given $\ket{\psi_0}$, the first step is to coherently factor $\ket{a}$ into a ``left transversal'' $\ket{w_1}$, an element of a subalgebra $\ket{b}$, and a ``right transversal'' $\ket{w_2}$. How we perform this factorization depends both on $n$ and the propagating number of $D_a$. There are several cases to consider, so it will be helpful to illustrate each with an example.
\paragraph{Case 1: $\pn(D_a) > 0$.}
\begin{itemize}
     \item \textbf{(Brauer 1)} When $A = B_n(d)$ and $n$ is even, there exist permutations $\pi_1 = (i\; n- 1)(j \;n)$, $\pi_2 = (k\; n - 1)(l \; n)$ such that
    \begin{equation}
        D_a = \pi_1D_b\pi_2,\;\; D_b \in B_{n-2}(d)
    \end{equation}
    \item \textbf{(Brauer 2)} When $A = B_n(d)$ and $n$ is odd, there exist permutations $\pi_1 = (i\;n)$, $\pi_2 = (k \; n)$ such that
    \begin{equation}
        D_a = \pi_1D_b\pi_2,\;\; D_b \in B_{n-1}(d)
    \end{equation}
    \begin{figure}[H]
        \centering
        \includegraphics[width=0.7\linewidth]{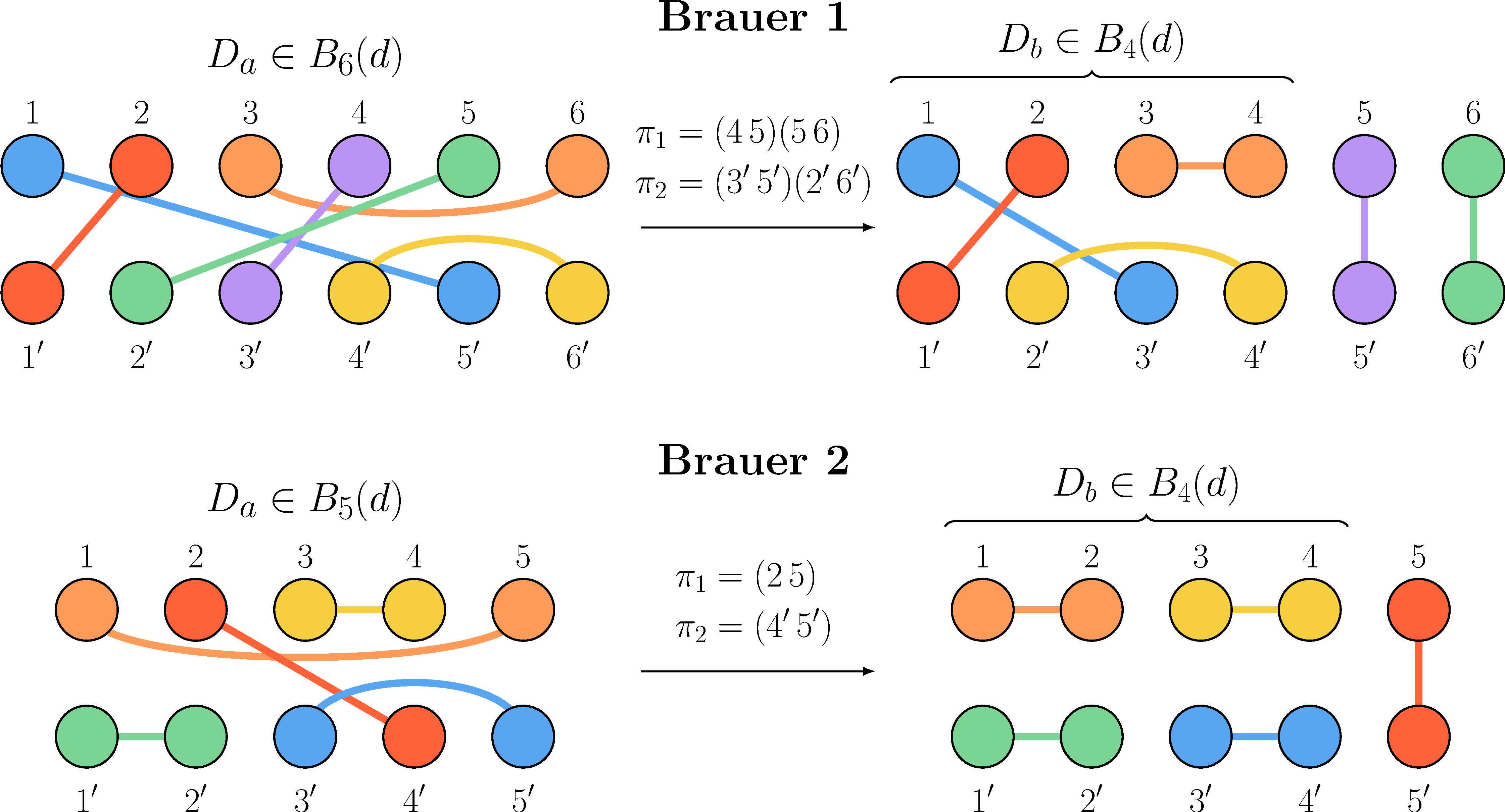}
        \label{fig:brauer_prop}
    \end{figure}
    \item \textbf{(Walled Brauer 1)}  When $A = B_{r,s}(d)$ and $s > r$, there exist permutations $\pi_1 = (j \;n)$, $\pi_2 = (l \; n)$ $\pi_1, \pi_2 \in S_s$ such that
    \begin{equation}
        D_a = \pi_1D_b\pi_2,\;\; D_b \in B_{r, s - 1}(d)
    \end{equation}
     \item \textbf{(Walled Brauer 2)} When $A = B_{r,r}(d)$, there exist permutations $\pi_1 = (i\; r)(j \;n)$, $\pi_2 = (k\;r)(l\; n)$, with $\pi_1, \pi_2 \in S_r \times S_r$ such that
    \begin{equation}
        D_a = \pi_1D_b\pi_2,\;\; D_b \in B_{r - 1, r - 1}(d)
    \end{equation}
   \noindent As in~\cref{eq:walleD_brauer_chain}, we assume without loss of generality that $s \ge r$.
    \begin{figure}[H]
    \centering
    \includegraphics[width=0.8\linewidth]{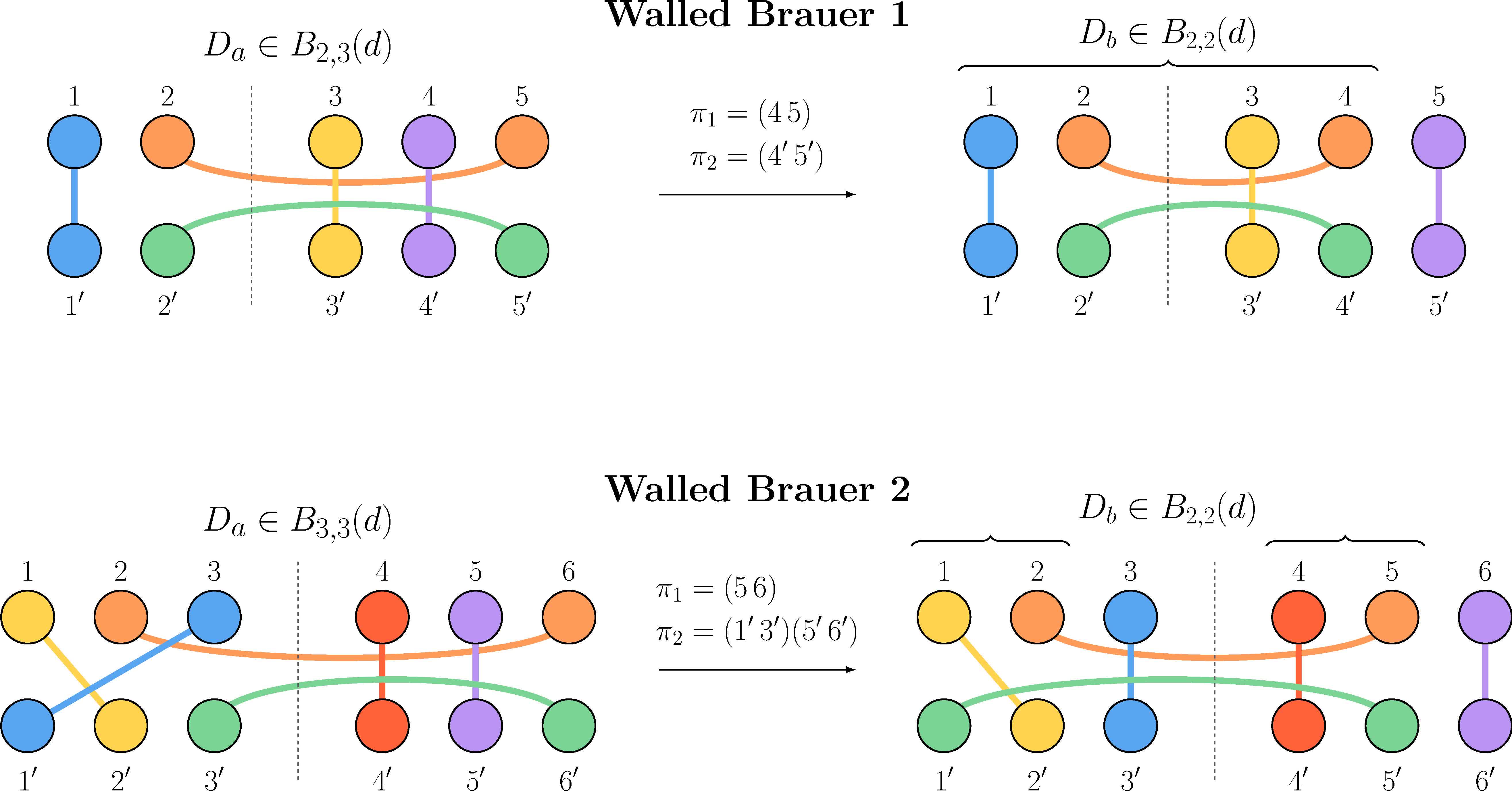}
    \label{fig:walled_brauer_prop_12}
    \end{figure}
\end{itemize}
For $A = P_n(d)$, we split into cases depending on whether there exists a propagating block with exactly two vertices (one in the top row, one on the bottom row):
\begin{itemize}
    \item \textbf{(Partition 1)} When $A$ has a propagating block with exactly two vertices, there exist permutations $\pi_1 = (i\;n)$, $\pi_2 = (k\; n )$ such that
    \begin{equation}
        D_a = \pi_1D_b\pi_2,\;\; D_b \in P_{n-1}(d)
    \end{equation}
     \item \textbf{(Partition 2)} When $A$ has a propagating block with at least two vertices in the top row, there exist permutations $\pi_1 = (i\;n - 1)(j\; n)$, $\pi_2 = (k\; n )$ such that
    \begin{equation}
        D_a = \pi_1b_{n-1}D_b\pi_2,\;\; D_b \in P_{n-1}(d)
    \end{equation}
     \item \textbf{(Partition 3)} When $A$ has a propagating block with at least two vertices in the bottom row, there exist permutations $\pi_1 = (i\;n)$, $\pi_2 = (k\; n - 1 )(l \; n)$ such that
    \begin{equation}
        D_a = \pi_1D_bb_{n-1}\pi_2,\;\; D_b \in P_{n-1}(d)
    \end{equation}
   \begin{figure}[H]
    \centering
    \includegraphics[width=0.7\linewidth]{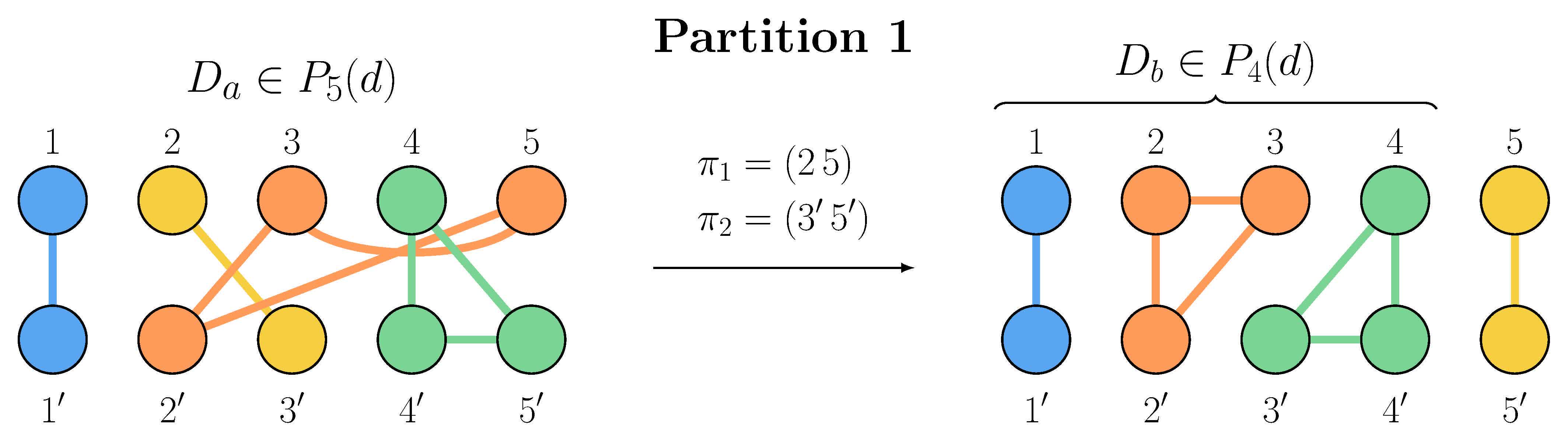}
    \label{fig:partition_prop1}
    \end{figure}
    \begin{figure}[H]
    \centering
    \includegraphics[width=0.7\linewidth]{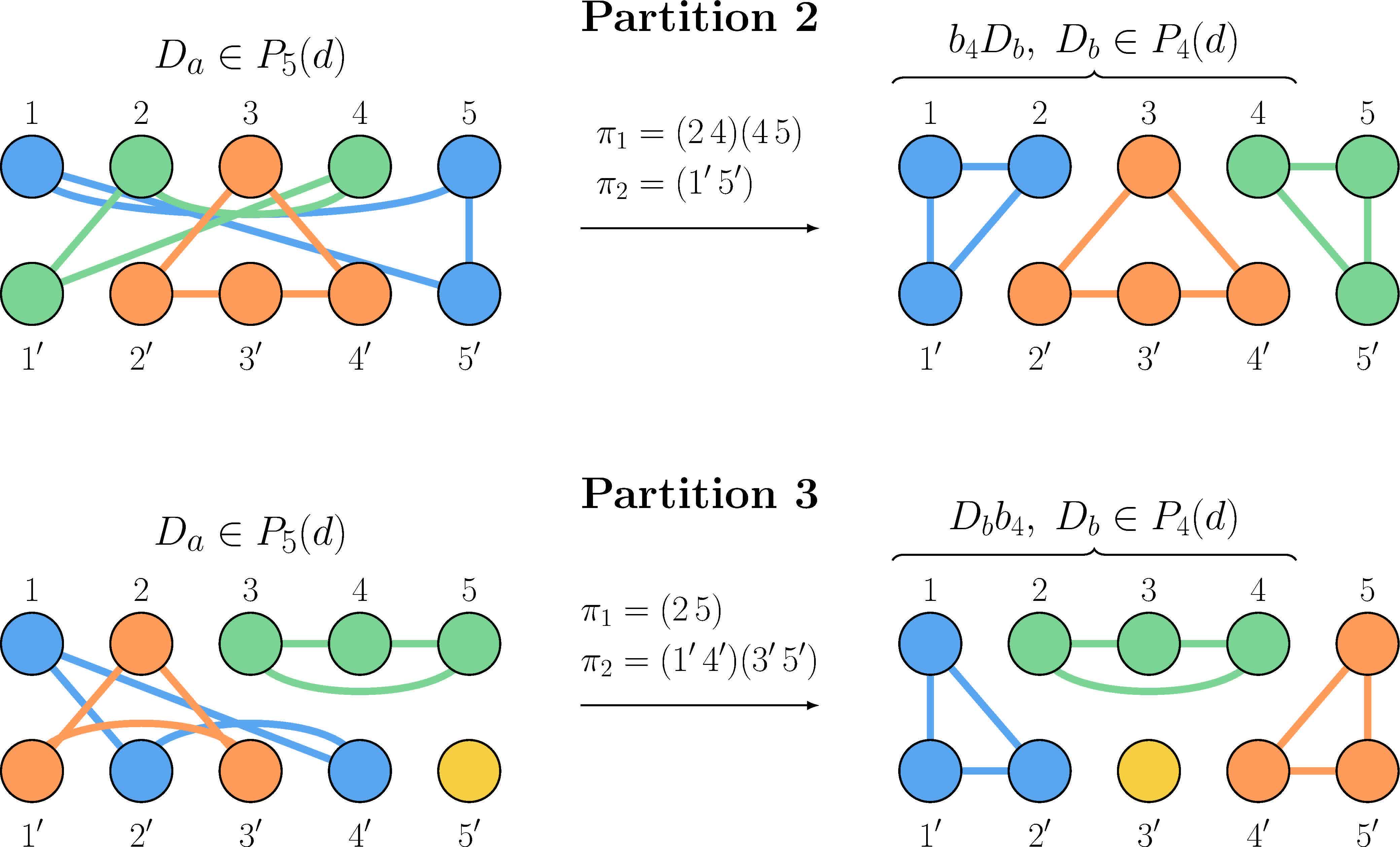}
    \label{fig:partition_prop23}
    \end{figure}
\end{itemize}
\paragraph{Case 2: $\pn(D_a) = 0$.} This case is very similar to the previous one, with the difference being that the factorization includes an extra generator not in the subalgebra. Note that the case with propagating number 0 only appears in the Brauer algebra when $n$ is even and in the walled Brauer algebra when $r = s$. 
\begin{itemize}
    \item \textbf{(Brauer, No Propagation)} When $A = B_n(d)$, there exist permutations $\pi_1 = (i\; n- 1)$, $\pi_2 = (k\; n - 1)$ such that
    \begin{equation}
        D_a = \pi_1D_b e_{n-1}\pi_2,\;\; D_b \in B_{n-2}(d)
    \end{equation}
     \item \textbf{(Walled Brauer, No Propagation)} When $A = B_{r,r}(d)$, there exist permutations $\pi_1 = (i\; r)$, $\pi_2 = (k\; r)$, with $\pi_1, \pi_2 \in S_r \times S_r$ such that
    \begin{equation}
        D_a = \pi_1D_b f_r\pi_2,\;\; D_b \in B_{r-1, r-1}(d)
    \end{equation}
\end{itemize}
\begin{figure}[H]
    \centering
    \includegraphics[width=0.7\linewidth]{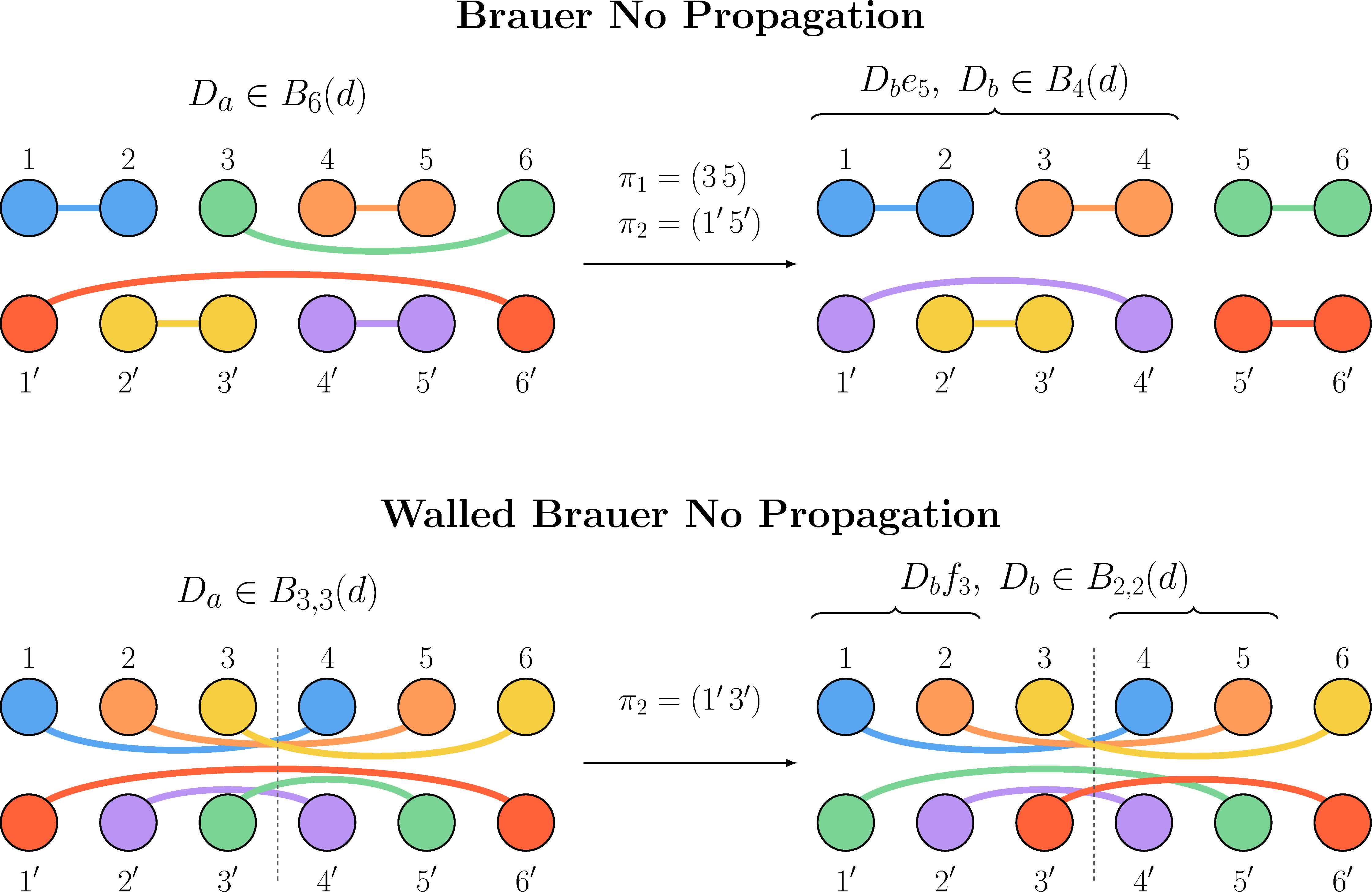}
    \label{fig:brauer_no_prop}
\end{figure}
For the partition algebra, there are four cases depending on whether or not the last two vertices on the top and bottom row belong to the same block: 
\begin{itemize}
    \item \textbf{(Partition, No Propagation 1)} When $A = P_n(d)$ and both $\{n\}$ and $\{n^\prime\}$ are blocks,
    \begin{equation}
        D_a =D_b p_n,\;\; D_b \in P_{n-1}(d)
    \end{equation}
    \item \textbf{(Partition, No Propagation 2)} When $A = P_n(d)$ and $\{n\}$ is a block by itself but not $\{n^\prime\}$, there exists a permutation $\pi = (i \; n - 1)$ such that  
    \begin{equation}
        D_a =D_b p_n b_{n-1} \pi,\;\; D_b \in P_{n-1}(d)
    \end{equation}
    \item \textbf{(Partition, No Propagation 3)} When $A = P_n(d)$ and $\{n^\prime\}$ is a block by itself but not $\{n\}$, there exists a permutation $\pi = (i \; n - 1)$ such that  
    \begin{equation}
        D_a = \pi b_{n-1} D_b p_n,\;\; D_b \in P_{n-1}(d)
    \end{equation}
    \item \textbf{(Partition, No Propagation 4)} When $A = P_n(d)$ and neither $\{n\}$ nor $\{n^\prime\}$ appear as blocks by themselves, there exist permutations $\pi_1 = (i \; n-1)$, $\pi_2 = (j \; n - 1)$
    \begin{equation}
        D_a =\pi_1 b_{n-1} D_b p_n b_{n-1} \pi_2,\;\; D_b \in P_{n-1}(d)
    \end{equation}
    \begin{figure}[H]
        \centering
        \includegraphics[width=\linewidth]{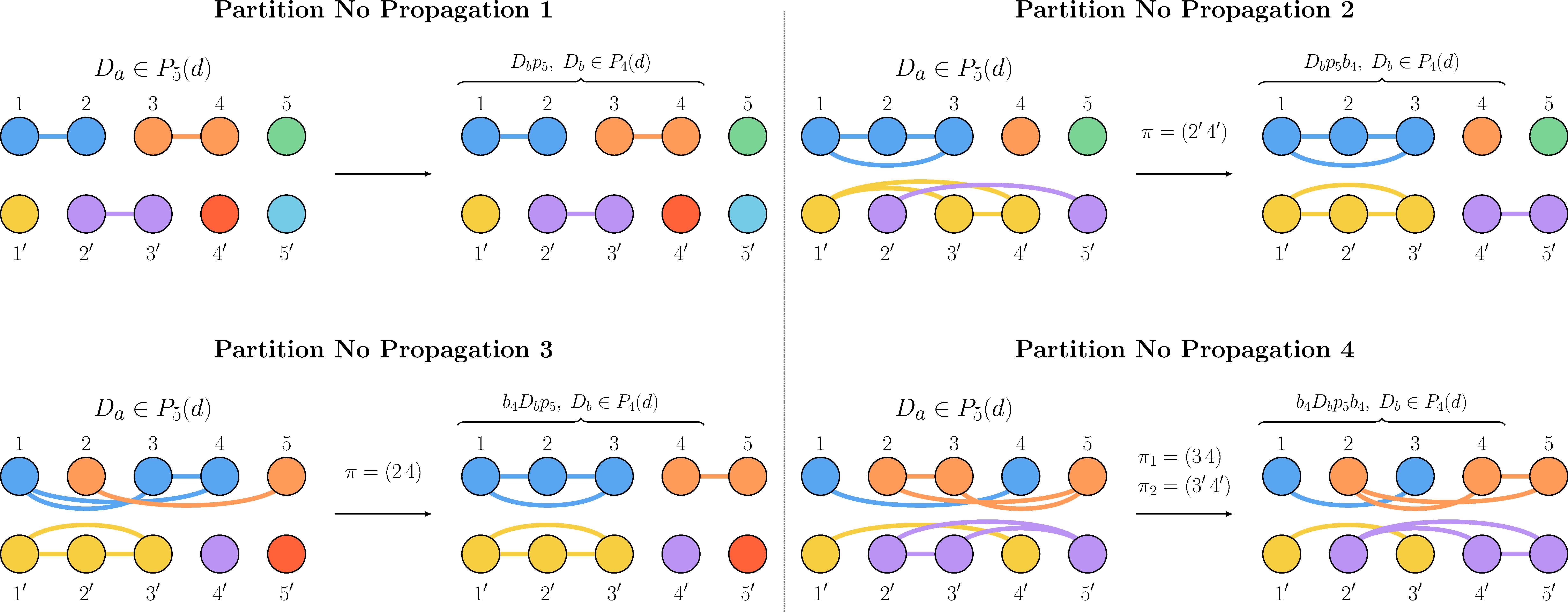}
    \end{figure}
\end{itemize}
 
\paragraph{The Last Possible Factorization.} An important subtlety is that unlike the case of factoring a group element into a subgroup element and a transversal, a factorization of a diagram need not be unique. For example, $s_ib_i = b_i$, so the diagram $b_1b_2$ (the single block element of $P_3(d)$) could also be factored as $s_1b_1b_2$. For our purposes, it will be important to use a specific factorization we call the \textit{last possible factorization}, which plays a crucial role in the final step of the algorithm (\cref{sec:accumulate}). All of the factorizations shown in the examples above are performed with respect to last possible factorizations specific to each algebra (\cref{def:factorings_for_our_algebras}).
\begin{definition}
    \label{def:ordered_factor}
    Let $A$ be an algebra with subalgebra $B$, and let $W = \{(w_1, w_2)\}_{w_1, w_2}$ be a complete set of transversals for $B$ in $A$. Let $<$ be a strict total ordering on $W$. The \textit{last possible factorization} (with respect to $<$) is a factorization of every $a \in \mathcal{D}(A)$ that chooses the last possible valid transversal with respect to $<$.\ That is, we factor $a$ as $a = w_1 b w_2$, with $(w_1, w_2) \in W$,\ $b \in \mathcal{D}(B)$, such that for all $(w_3, w_4) \in W$,
    \begin{equation}
         (w_1, w_2) < (w_3, w_4)  \implies \not\exists b' \in \mathcal{D}(B) \;\text{ such that } a = w_3 b' w_4
    \end{equation}
    In other words, $(w_1, w_2)$ is the greatest transversal (under $<$) for which $a$ can be factored into transversals and an element in $\mathcal{D}(B)$.  
\end{definition}
\noindent We now define orderings for all of the subalgebras we consider above: 
\begin{definition}[Last Possible Factorizations for $B_{n}(d), B_{r,s}(d), P_n(d)$]
    \label{def:factorings_for_our_algebras}
    \phantom{} 
    \newline
     Below are the orderings on the transversals we will use in our algorithm: 
    \begin{itemize}
    \item $A = B_n(d)$ for $n$ odd, so $B = B_{n-1}(d)$. In this case, 
    \begin{equation}
        W = \{(i \; n), (k \; n)\}_{i, k}
    \end{equation}
    corresponding to the case \textbf{Brauer 2} from above. We order $W$ using the natural lexicographic ordering.\footnote{Technically, the order among permutations will not matter, but we choose the lexicographic ordering for concreteness.} 
    \item  $A = B_n(d)$ for $n$ even, so $B = B_{n-2}(d)$. In this case, 
    \begin{equation}
        W =\{(i\; n- 1),\;e_{n-1}(k\; n - 1))\}_{i,k} \cup \{(i\; n- 1)(j \;n),\;(k\; n - 1)(l \; n))\}_{i < j, k <l} 
    \end{equation}
    where the first half corresponds to \textbf{Brauer, No Propagation}, and the second corresponds to \textbf{Brauer 1}. We order both halves lexicographically, but order the \textit{entire first half} (i.e. transversals with a contraction) before the \textit{entire second half}.  
    \item $A = B_{r,s}(d)$, with $s > r$. In this case,
    \begin{equation}
        W = \{(i \; n), (k \; n)\}_{i, k}
    \end{equation}
    corresponding to the case \textbf{Walled Brauer 1} from above.\ We again order $W$ using the natural lexicographic ordering. 
    \item $A = B_{r,r}(d)$. In this case,
     \begin{equation}
        W = \{(i\; r),\;f_r(k\; r))\}_{i, k} \cup \{(i\; r)(j \;n),\;(k\; r)(l \; n))\}_{i, j,k,l} 
    \end{equation}
     where the first half corresponds to \textbf{Walled Brauer, No Propagation}, and the second corresponds to \textbf{Walled Brauer 2}. Like the second case for the Brauer algebra, we order both halves lexicographically, but order the entire first half before the entire second half.  
     \item $A = P_n(d)$, so $B = P_{n-1}(d)$. In this case, 
        \begin{align}
        W_{0,1}
        &=
        \{(I,\;p_n)\},
        \\
        W_{0,2}
        &=
        \{(I,\;p_n b_{n-1}(i\;n-1))\}_{i},
        \\
        W_{0,3}
        &=
        \{((i\;n-1)b_{n-1},\;p_n)\}_{i},
        \\
        W_{0,4}
        &=
        \{((i\;n-1)b_{n-1},\;p_n b_{n-1}(j\;n-1))\}_{i,j},
        \\
        W_{+,1}
        &=
        \{((i\;n),\;(k\;n))\}_{i,k},
        \\
        W_{+,2}
        &=
        \{((i\;n-1)(j\;n)b_{n-1},\;(k\;n))\}_{i <j,k},
        \\
        W_{+,3}
        &=
        \{((i\;n),\;b_{n-1}(k\;n-1)(l\;n))\}_{i,k <l},
        \\[2mm]
        W
        &=
        W_{0,1}
        \cup W_{0,2}
        \cup W_{0,3}
        \cup W_{0,4}
        \cup W_{+,1}
        \cup W_{+,2}
        \cup W_{+,3},
        \\[2mm]
        W_{0,4}
        &<
        W_{0,3}
        <
        W_{0,2}
        <
        W_{0,1}
        <
        W_{+,3}
        <
        W_{+,2}
        <
        W_{+,1}.
        \end{align}
     Here, $I$ is the identity diagram, and the seven component sets correspond to the seven cases (\textbf{Partition, No Propagation 1}), (\textbf{Partition, No Propagation 2}), (\textbf{Partition, No Propagation 3}), (\textbf{Partition, No Propagation 4}), \textbf{Partition 1}, \textbf{Partition 2}, \textbf{Partition 3}. 
     
      The only important property of this ordering is that all component sets corresponding to cases with propagating number zero appear before those corresponding to cases with propagating number greater than zero. This is also the reason for the ordering of the two halves in the second Brauer case and the second walled Brauer case.
     \end{itemize}

\end{definition}
By choosing the orderings on transversals as in~\cref{def:factorings_for_our_algebras}, we ensure that when factoring a diagram $D_a$ with $\pn(D_a) > 0$, the chosen transversal in the factorization contains only permutations (or permutation and bridge generators for $P_n(d)$), and not any additional non-propagating generators such as $e_{i-1}$, $f_r$, or $p_i$. In other words, with the ordering in~\cref{def:factorings_for_our_algebras}, a diagram $D_a$ with $\pn(D_a) > 0$ will be factored according to some subcase of \textbf{Case 1} above.\footnote{Also, a diagram $D_a$ with $\pn(D_a) = 0$ will be factored according to some subcase of \textbf{Case 2}, although this is true for any possible ordering of the transversals, since factorings in \textbf{Case 1} will not be valid for such diagrams.} This is important, as we will see in~\cref{sec:apply_irreps}, because we give efficient implementations for applying these non-propagating generators only when the propagating number is zero.

Regardless of which case we are in, we can perform the corresponding last possible factorization using $\wt{O}(n^5)$ gates for the Brauer and walled Brauer algebras, and $\wt{O}(n^4)$ gates for the partition algebra.\footnote{This step could likely be further optimized. However, since the gate complexity of this step will be dominated by later steps of the algorithm anyways, we do not attempt to optimize it here.}: we simply iterate over all transversals from greatest to least with respect to the ordering, and use $\wt{O}(n)$ gates to check if the transversal would give a valid factorization. There are at most $O(n^4)$ transversals for $B_{n}(d)$ and for $B_{r,s}(d)$, and $O(n^3)$ transversals for $P_n(d)$ (see~\cref{def:factorings_for_our_algebras}).\ Since we are iterating over the transversals from greatest to least with respect to the ordering, the first valid factorization we find will be the last possible factorization. After performing the last possible factorization, we obtain the state
\begin{equation}
    \ket{\psi_1} =   \sum_{w_1, w_2} \sum_{b \in \mathcal{F}_{(w_1, w_2)}} f(w_1bw_2) \ket{w_1}_{\mathsf{L}} \otimes \ket{D_b}_\mathsf{B} \otimes  \ket{w_2}_{\mathsf{R}} \
\end{equation}
Where $D_b \in \mathcal{D}(B)$, $w_1$ and $w_2$ are words in the generator diagrams, and $\mathcal{F}_{(w_1, w_2)} \subseteq \mathcal{D}(B)$ denotes the set of all diagrams in $B$ for which some $D_a \in \mathcal{D}(A)$ is factored as $w_1D_bw_2$, according to the last possible factorization. Note that whether or not $\pn(D_a) = 0$ is implicitly stored in $w_2$, since $\pn(D_a) = 0$ if and only if either $e_{n-1}$, $f_r$, or $p_n$ is present at the start of $w_2$. 
\begin{center}
\fbox{
  \parbox{0.92\linewidth}{
    \centering
         \textbf{Gate Complexity for $B_{r,s}(d), B_n(d)$: $\wt{O}(n^5)$. For $P_n(d)$: $\wt{O}(n^4)$.} \textbf{Operator Norm Error}: $0$. 
  }
}
\end{center}
\subsection{Step 2: Recurse}
\label{sec:recurse}
Next, we recursively apply $\wt{\ft_B}$ on the $\mathsf{B}$ register: 
\begin{equation}
    \label{eq:after_recurse}
    \ket{\psi_2} =    
    \sum_{w_1, w_2} 
    \sum_{b \in \mathcal{F}_{(w_1, w_2)}} 
    f(w_1bw_2) 
    \ket{w_1}_{\mathsf{L}} 
    \otimes   
    \ket{w_2}_{\mathsf{R}} 
    \otimes  
    \sum_{\sigma \in \wh{B}}
    \sum_{P, Q} 
    \left\lVert E_{PQ}^{\sigma}\right\rVert_2 
    \;
    \sigma(b)_{PQ} 
    \;
    \ket{\sigma, P, Q}_\mathsf{S}
\end{equation}
Note that $D_b$, when renormalized as a basis element of $A$, may have a different normalization than as a basis element of $B$. These is handled implicitly when promoting a Fourier state of $B$ to a Fourier state of $A$, as described in the sections below. 
\begin{center}
\fbox{
  \parbox{0.92\linewidth}{
    \centering
         \textbf{Gate Complexity: Complexity of $\wt{\ft_B}$.} \textbf{Operator Norm Error}: Error of $\wt{\ft_B}$. 
  }
}
\end{center}
\subsection{Step 3: Embed}
\label{sec:embedding}
Next, our goal is to promote the $\mathsf{S}$ register, which contains a Fourier state of $B$, to a Fourier state of $A$. This can be accomplished by the following \textit{embedding map}:
\begin{definition}
    Assume that $A$ has a subalgebra $C$. Define the embedding map $\mathbf{E}(A, C)$ as follows: 
\begin{equation*}
         \mathbf{E}(A, C)\ket{\sigma, P, Q} = \sum_{\rho: \sigma \in \text{Res}(\rho)^{A}_{C}} m_{\rho, \sigma} \frac{||E^\rho_{P \circ \rho, Q \circ \rho}||_2}{||E^{\sigma}_{PQ}||_2}\ket{\rho, P \circ \rho, Q \circ \rho}
\end{equation*}
where $m_{\rho, \sigma}$ denotes the multiplicity of $V^\sigma$ in $V^\rho \!\downarrow_C$ (\cref{sec:subalgebra_adapteD_basis}). 
\end{definition}
When $A$ and $C$ are clear from context, we will simply write $\mathbf{E}$. We only define $\mathbf{E}$ on states of the form $\ket{\sigma, P, Q}$, where $\sigma \in \wh{C}$ and $P, Q \in V^{\sigma}$, as we will only apply $\mathbf{E}$ with inputs supported on this subspace. On this subspace, we can check that $\mathbf{E}$ does indeed promote Fourier states of $C$ into Fourier states of $A$:
\begin{lemma}
    For any $c \in C$, $\mathbf{E} \cdot \wt{\ft_C}\ket{c} = \wt{\ft_A}\ket{c}$.
\end{lemma}
\begin{proof}
    \begin{align}
        \mathbf{E}\cdot \wt{\ft_C}\ket{c} 
        &=
        \sum_{\sigma \in \wh{C}}
        \sum_{P, Q} 
        \left\lVert E_{PQ}^{\sigma}\right\rVert_2 
        \;
        \sigma(c)_{PQ} 
        \;
        \ket{\sigma, P, Q}_\mathsf{S}
        \\ 
        &= \sum_{\substack{\rho, \sigma \\  \sigma \in \text{Res}(\rho)^{A}_{C}}} 
        \sum_{PQ} m_{\rho, \sigma}  
        \frac{\lVert E^\rho_{P \circ \rho, Q \circ \rho}\rVert_2}{\lVert E^{\sigma}_{PQ}\rVert_2} 
        \cdot 
        \left\lVert E_{PQ}^{\sigma}\right\rVert_2 
        \;
        \sigma(c)_{PQ} 
        \;
        \ket{\rho, P \circ \rho, Q \circ \rho} \\ 
        &= 
        \sum_{\rho} \sum_{P^\prime, Q^\prime} 
        \left\lVert E^{\rho}_{P^\prime Q^\prime}\right\rVert_2
        \;
        \rho(c)_{P^\prime, Q^\prime}
        \;
        \ket{\rho, P^\prime, Q^\prime} 
        = 
        \wt{\ft_A}\ket{c}
    \end{align}
    where $P^\prime = P \circ \rho$ and $Q^\prime = Q \circ\rho$. The third equality follows from~\cref{eq:subalgebra_adapted_matrices}.
\end{proof}
As a consequence, if we have the chain $B \subseteq C \subseteq A$ of algebra inclusions, then $\mathbf{E}(A, C)\mathbf{E}(C, B) = \mathbf{E}(A, B)$.\ We emphasize this because in our context, it is conceptually simpler to move two steps at a time along the subalgebra chains defined in~\cref{sec:subalgebra_chains} (i.e. $B_{n-2}(d) \hookrightarrow B_{n}(d)$ or $P_{n-1}(d) \hookrightarrow P_{n}(d)$), in which case we may need to apply the embedding operator twice. 

\paragraph{Implementing the Embedding Operator.} When trying to implement $\mathbf{E}$ on a quantum computer, we immediately run into an obstacle: $\mathbf{E}$ may not be an isometry. Fortunately,~\cref{thm:subalgebra_restrictions} and~\cref{thm:approx_norm_restriction} imply the existence of an \textit{approximate} embedding map $\wt{\mathbf{E}}$ which is an isometry. Moreover, $\wt{\mathbf{E}}$ is efficiently implementable:

\begin{lemma}
    \label{lem:embedding_summary}
Assume that $A \in \{P_n(d), P_{n - \frac12}(d), B_n(d), B_{r,s}(d)\}$, with subalgebra $C$ given by the chains in~\cref{sec:subalgebra_chains}. Then, there exists an isometric operator $\wt{\mathbf{E}}$, 
implementable to operator norm error $\varepsilon$ with gate
complexity 
\begin{equation}
    \wt{O}(\sqrt{n}\cdot(n^2\log d+\log(1/\varepsilon))),
\end{equation}
such that
\begin{equation}
    ||\wt{\mathbf{E}} - \mathbf{E}||_\infty
    \le O\left(n^{1/4} \cdot (\delta + \varepsilon)\right).
\end{equation}
\end{lemma}
Therefore, whenever we would like to apply $\mathbf{E}$, we can instead apply $\wt{\mathbf{E}}$. As long as $d$ is sufficiently large, $\delta$ will be small, and the incurred operator norm error will also be small.

\begin{proof}
In giving the implementation, we assume that $C$ is a subalgebra of $A$ with multiplicity-free branching corresponding to one of the chains in~\cref{sec:subalgebra_chains}.
\begin{enumerate}
    \item First, conditioned on $\ket{\sigma}$, determine all irreps $\rho \in \wh{A}$ such that $\sigma \in \text{Res}(\rho)^{A}_{C}$. Using the branching rules from~\cref{thm:branching_rules}, this can be done efficiently. Note that there are at most $O(\sqrt{n})$ such irreps $\rho$.\footnote{The total number of addable and removable boxes for an $n$ box Young diagram is maximized by the Young diagram shaped as a ``staircase'', corresponding to the partition $(k, k-1, \dots, 2, 1)$, where $k = O(\sqrt{n})$.} Therefore, the following map can be realized with $\wt{O}(n^{3/2})$ gates:
    \begin{equation}
        \ket{\sigma, P, Q}
        \mapsto
        \ket{\sigma, P, Q}
        \otimes
        \bigotimes_{\rho: \sigma \in \text{Res}(\rho)^{A}_{C}}
        \ket{\rho}.
    \end{equation}

        \item Next, conditioned on $\rho$, compute the ratio $r_{\rho,\sigma}$, defined as follows:
    \begin{equation}
        \label{eq:step_two_embed_eq}
        r_{\rho, \sigma}
        =
        \begin{cases}
            \frac{m_\rho}{d \cdot m_\sigma}
            & \text{if $A = B_n(d)$, $A = B_{r,s}(d)$, or $A = P_{n-\frac12}(d)$}, \\
            \frac{m_\rho}{m_\sigma}
            & \text{if $A = P_n(d)$}.
        \end{cases}
    \end{equation}
    By~\cref{lem:appendix_ratio_closeness}, this approximates ${||E^\rho_{P \circ \rho, Q \circ \rho}||_2^2}/{||E^{\sigma}_{PQ}||_2^2}$ up to $O(\delta)$ error. In the second case, the Schur representation of $P_n(d)$ restricts to that of $P_{n-\frac12}(d)$ on the same tensor space $(\mathbb{C}^d)^{\otimes n}$, so no additional factor of $d$ appears.

    There are known formulas for $m_\rho, m_\sigma$ that are efficient to compute (see \cref{sec:multiplicity_formulas}). By~\cref{lem:efficient_to_compute_multiplicities}, computing the ratio $r_{\rho,\sigma}$ up to error $\eta$ requires $\wt{O}(n^2\log d + \log(1/\eta))$ gates.\footnote{Although this upper bound is loose for the algebras other than the walled Brauer algebra, later steps of the algorithm will dominate this gate count anyways. To simplify the presentation, we use a uniform bound for all algebras.}
    The branching rules and restriction of the Schur representations imply that
    \begin{equation}
    \label{eq:embedding_weights_sum_one}
        \sum_{\rho: \sigma \in \text{Res}(\rho)^{A}_{C}}
        r_{\rho,\sigma}
        = 1.
    \end{equation}
    Taking $\eta = O(\varepsilon^2/\sqrt{n})$ and choosing the rotation precision appropriately suffices to implement the state preparation described in Step 3 to operator norm accuracy $\varepsilon$, with at most
    \begin{equation}
        \wt{O}\left(\sqrt{n} \cdot
        (n^2\log d + \log(1/\varepsilon))\right)
    \end{equation}
    gates. Suppressing the finite precision of the coefficient registers, we have prepared the following state:
    \begin{equation}
        \ket{\sigma, P, Q}
        \otimes
        \bigotimes_{\rho: \sigma \in \text{Res}(\rho)^{A}_{C}}
        \ket{\rho, r_{\rho,\sigma}}.
    \end{equation}
    We will also uncompute $\sigma$, using the final irrep of the Bratteli path for either $P$ or $Q$ (which must itself also be $\ket{\sigma}$).

    \item Next, for each $\ket{\rho}$, we apply a Givens rotation
    \begin{equation}
        \label{eq:givens_rotation}
        \ket{P, Q}
        \mapsto
        \cos\theta_i\ket{P, Q}
        +
        \sin\theta_i\ket{P \circ \rho, Q \circ \rho},
    \end{equation}
    which rotates amplitude from the Bratteli paths for irreps of $C$ to Bratteli paths for irreps of $A$. Here, $\theta_i$ is chosen using partial sums of the weights $r_{\rho,\sigma}$. Choosing the angles $\theta_i$ in this way is a standard state-preparation technique; see, e.g., \cite[Thm.~9]{ShendeBullockMarkov2006Synthesis} or \cite[Eq.~(8)]{GuiDalzellAchilleSucharaChong2024SpacetimeEfficient}. The gate complexity will be dominated by the $\wt{O}(\sqrt{n} \cdot (n^2\log d + \log(1/\varepsilon)))$ gate cost of computing the coefficients. After this step, suppressing implementation error, we obtain
    \begin{equation}
        \left(
            \sum_{\rho: \sigma \in \text{Res}(\rho)^{A}_{C}}
            \sqrt{r_{\rho,\sigma}}
            \ket{P \circ \rho, Q \circ \rho}
        \right)
        \otimes
        \bigotimes_{\rho: \sigma \in \text{Res}(\rho)^{A}_{C}}
        \ket{\rho, r_{\rho,\sigma}}.
    \end{equation}

    \item Finally, $\rho$ can be computed by copying the last element of the Bratteli path, and all intermediate steps can also be uncomputed. This results in the final state
    \begin{equation}
        \wt{\mathbf{E}}\ket{\sigma,P,Q}
        =
        \sum_{\rho: \sigma \in \text{Res}(\rho)^{A}_{C}}
        \sqrt{{r_{\rho,\sigma}}}
        \ket{\rho, P \circ \rho, Q \circ \rho}.
    \end{equation}
\end{enumerate}

The operator $\wt{\mathbf{E}}$ is an isometry because it maps distinct $\ket{\sigma,P,Q}$ basis vectors to normalized states with disjoint supports.

For a fixed basis vector $\ket{\sigma,P,Q}$, define $a_\rho
    =
    {||E^\rho_{P\circ\rho,Q\circ\rho}||_2}/
         {||E^\sigma_{PQ}||_2}$.
By~\cref{lem:appendix_ratio_closeness}, each amplitude satisfies
$
    \left|a_\rho-\sqrt{r_{\rho,\sigma}}\right|
    \le O(\delta)$.
Since there are at most $K=O(\sqrt{n})$ possible extensions $\rho$, we obtain
\begin{equation}
    \label{eq:basis_vector_bound_embedding}
    ||(\mathbf{E}-\wt{\mathbf{E}})\ket{\sigma,P,Q}||_2
    =
    \left(
        \sum_{\rho:\sigma\in\text{Res}(\rho)^A_C}
        \left|a_\rho-\sqrt{r_{\rho,\sigma}}\right|^2
    \right)^{1/2}
\end{equation}
\begin{equation}
    \le
    \sqrt{K\cdot O(\delta^2)}
    =
    O(n^{1/4} \cdot \delta).
\end{equation}
By disjointness of the output supports, this also bounds
$||\mathbf{E}-\wt{\mathbf{E}}||_\infty$. Including the
$\varepsilon$ implementation error from Steps 2--3 gives total
operator norm error
\begin{equation}
    O(n^{1/4} \cdot \delta+\varepsilon)
    \le
    O\left(n^{1/4} \cdot (\delta+\varepsilon)\right),
\end{equation}
with the desired gate complexity.
\end{proof}

 We will not yet apply $\wt{\mathbf{E}}$ to \cref{eq:after_recurse}, since it will used as part of a larger subroutine involving the next step.
\subsection{Step 4: Applying Irreps}
\label{sec:apply_irreps}
Suppose that we have a Fourier state 
\begin{equation}
    \wt{\ft_A}\ket{a} = \sum_{\rho \in \wh{A}}  \sum_{P, Q} ||E_{PQ}^\rho||_2\cdot  \rho(a)_{PQ} \ket{\rho, P, Q} 
\end{equation}
which we would like to map to $\wt{\ft_A}\ket{wa}$. We can accomplish this up to operator norm error $O(\poly(|A|) \cdot \delta)$ by applying $\rho(w)$ to the register containing the left Bratteli path, conditioned on $\ket{\rho}$:
\begin{align}
    \label{eq:applying_irreps}
    \wt{\ft_A}\ket{a} 
    \mapsto
    \left(
        \sum_{\rho \in \wh{A}}
        \ket{\rho}\!\!\bra{\rho} 
        \otimes
        \rho(w)
        \otimes
        I
    \right)
    \wt{\ft_A}\ket{a} 
    =
    &
    \sum_{\rho \in \wh{A}}  \sum_{P, Q, R} ||E_{PQ}^\rho||_2\cdot  \rho(w)_{RP} \; \rho(a)_{PQ} \ket{\rho, R, Q} 
    \\
    = 
    &
    \sum_{\rho \in \wh{A}}  \sum_{Q, R} ||E_{PQ}^\rho||_2\cdot  \rho(wa)_{RQ} \ket{\rho, R, Q} 
    \\
    \approx
    &
    \sum_{\rho \in \wh{A}}  
    \sum_{Q, R} ||E_{RQ}^\rho||_2\cdot  \rho(wa)_{RQ} \ket{\rho, R, Q}
    = 
    \wt{\ft_A}\ket{wa}
\end{align}
The additive error incurred by replacing $||E_{PQ}^\rho||_2$ with $||E_{RQ}^\rho||_2$ on the final line is bounded by $O(\delta)$ (\cref{cor:norm_of_fourier_states}). By~\cref{lem:opnorm_bound}, the total operator norm error incurred is bounded by $O(\poly(|A|) \cdot \delta)$.

Similarly, we can obtain $\wt{\ft_A}\ket{aw}$ by applying $\rho(w^{\op}) = \rho(w)^{\dagger}$ on the column register. Now, assume that $w = s_i$ is a swap diagram. By \cref{eq:orthogonal_form}, $\rho(s_i) = \rho(s_i)^\dagger$ and $\rho(s_i)\rho(s_i)^\dagger = I$. Therefore, $\rho(s_i)$ is unitary, and so for any Bratteli path $P$, $\sum_{Q} |\rho(s_i)_{QP}|^2 = 1$. All non-zero entries in the $P$-column can be computed using~\cref{lem:efficient_computation_generator_matrix_elements_all}, after which the same sequential Givens rotation technique used in~\cref{eq:givens_rotation} can be used to apply $\rho(s_i)$ to either the row or column register. The gate complexity of performing these Givens rotations is dominated by the cost of computing the coefficients.  
\begin{lemma}
    \label{lem:swap_application}
    Up to operator norm error $O(\poly(|A|) \cdot (\delta + \varepsilon))$, the maps
    \begin{equation}
        \wt{\ft_A}\ket{a} \mapsto \wt{\ft_A}\ket{s_ia},\;\;\;\wt{\ft_A}\ket{a} \mapsto \wt{\ft_A}\ket{as_i}
    \end{equation}
    can be performed by applying a controlled-$\rho(s_i)$ gate to either the row or column register. The gate complexities are asymptotically equivalent to the gate complexities in~\cref{lem:efficient_computation_generator_matrix_elements_all}.
\end{lemma}
For the contraction, bridge, and point diagrams, the associated irrep matrices are non-unitary. At first glance, this is problematic, since the factorizations in Step 1 involve these generators. However, we do know that if both the input and output state are normalized Fourier basis states, there is some unitary transformation which (approximately) maps one to the other. 

To implement such a unitary, it will be useful to combine the application of a non-invertible generator with the embedding step. The cases for which this can be done precisely correspond to the different factorizations in Step 1 which involve a non-invertible generator, and are given in the statements of~\cref{lem:brauer_extension},~\cref{lem:walled_brauer_extension_rule},~\cref{lem:efficient_add_a_point}, and~\cref{lem:efficient_add_a_bridge}.\ We defer the corresponding proofs to~\cref{lem:brauer_extension_appendix},~\cref{lem:walled_brauer_extension_appendix},~\cref{lem:efficient_add_a_point_appendix}, and~\cref{lem:efficient_add_a_bridge_appendix} respectively. 

\begin{lemma}
    \label{lem:brauer_extension}
    Let $B = B_{n-2}(d)$, $A = B_n(d)$. For any diagram $D \in B$ with $\pn(D) = 0$, the isometry
    \begin{equation}
        \mathbf{F}(A, B)\wt{\ft_B}\ket{D} = \wt{\ft_A}\ket{D e_{n-1}}
    \end{equation}
    can be implemented up to operator norm error $O(\poly(|A|) \cdot \delta)$ by the map
    \begin{equation}
        \wt{\mathbf{F}}(A, B)\ket{\emptyset, P, Q} = \ket{\emptyset, P \circ \Box \circ \emptyset, Q \circ \Box \circ \emptyset}
    \end{equation} 
\end{lemma}
 \noindent Like with the embedding map $\mathbf{E}$, we will write $\mathbf{F} \coloneqq \mathbf{F}(A, B)$, $\wt{\mathbf{F}} \coloneqq \wt{\mathbf{F}}(A,B)$ when $A$ and $B$ are clear from context. An analogy of~\cref{lem:brauer_extension} holds for the walled Brauer algebra:
\begin{lemma}
    \label{lem:walled_brauer_extension_rule}
    Let $B = B_{r-1, r-1}(d)$, $A = B_{r,r}(d)$. For any diagram $D \in B$ with $\pn(D) = 0$, the isometry
    \begin{equation}
        \mathbf{F}(\wt{\ft_B}\ket{D}) = \wt{\ft_A}\ket{D f_{r}}
    \end{equation}
    can be implemented up to operator norm error $O(\poly(|A|) \cdot \delta)$ by the map 
    \begin{equation}
        \wt{\mathbf{F}}\ket{(\emptyset, \emptyset), P, Q} = \ket{(\emptyset, \emptyset), \;P \circ (\Box, \emptyset) \circ (\emptyset, \emptyset), \;Q \circ (\Box, \emptyset) \circ (\emptyset, \emptyset)}
    \end{equation}
\end{lemma}
\noindent We also derive similar results for the partition algebra, where we handle various extensions by bridge and point diagrams. Recall that basis elements for the partition algebra are rescaled according to the number of connected components (\cref{eq:basis_for_a_diagram_algebra}). 
\begin{lemma}
    \label{lem:efficient_add_a_point}
    Let $B = P_{n - 1}(d)$, $A = P_{n}(d)$. For any basis element $b = d^{({n - 1 - \cc(D)})/2}D \in \mathcal{B}(B)$, the isometries
    \begin{equation}
        \mathbf{F}_1(\wt{\ft_B}\ket{b}) = \wt{\ft_A}\ket{d^{(n - 2 - \cc(D))/2} Dp_{n}}
    \end{equation}
    \begin{equation}
        \mathbf{F}_2(\wt{\ft_B}\ket{b}) = \wt{\ft_A}\ket{d^{(n - 1 - \cc(D))/2}b_{n-1}D p_{n}}
    \end{equation}
    \begin{equation}
        \mathbf{F}_3(\wt{\ft_B}\ket{b}) = \wt{\ft_A}\ket{d^{(n - 1 - \cc(D))/2}D p_{n}b_{n-1}}
    \end{equation}
    \begin{equation}
        \mathbf{F}_4(\wt{\ft_B}\ket{b}) = \wt{\ft_A}\ket{d^{(n - \cc(D))/2} b_{n-1}D p_{n}b_{n-1}}
    \end{equation}
    can be implemented up to operator norm error $O(\poly(|A|) \cdot \delta)$ by the maps 
    \begin{equation}
        \label{eq:partition_on_zero_first}
        \wt{\mathbf{F}}_1\ket{\emptyset, P, Q} = \ket{\emptyset,\;\;P_{\le 2n - 4} \circ \emptyset \circ \emptyset \circ \emptyset \circ \emptyset,\;\; Q_{\le 2n - 4}\circ \emptyset \circ \emptyset \circ \emptyset \circ \emptyset}
    \end{equation}
    \begin{equation}
        \wt{\mathbf{F}}_2\ket{\emptyset, P, Q} = \ket{\emptyset,\;\;P_{\le 2n - 4} \circ \emptyset \circ \Box \circ \emptyset \circ \emptyset,\;\; Q_{\le 2n - 4}\circ \emptyset \circ \emptyset \circ \emptyset \circ \emptyset}
    \end{equation}
     \begin{equation}
        \wt{\mathbf{F}}_3\ket{\emptyset, P, Q} = \ket{\emptyset,\;\;P_{\le 2n - 4} \circ \emptyset \circ \emptyset \circ \emptyset \circ \emptyset,\;\; Q_{\le 2n - 4}\circ \emptyset \circ \Box \circ \emptyset \circ \emptyset}
    \end{equation}
     \begin{equation}
     \label{eq:partition_on_zero_fourth}
        \wt{\mathbf{F}}_4\ket{\emptyset, P, Q} = \ket{\emptyset,\;\;P_{\le 2n - 4} \circ \emptyset \circ \Box \circ \emptyset \circ \emptyset,\;\; Q_{\le 2n - 4}\circ \emptyset \circ \Box \circ \emptyset \circ \emptyset}
    \end{equation}
     Note that $\pn(D) = 0$ implies that $P(2n - 3) = P(2n - 2) = Q(2n - 3) = Q(2n - 2) = \emptyset$.
\end{lemma}
\noindent $\pn(D) = 0$ implies that the outputs of~\cref{eq:partition_on_zero_first} to~\cref{eq:partition_on_zero_fourth} are Fourier transforms of the basis states in $\mathcal{B}(A)$. 

The maps in~\cref{lem:brauer_extension},~\cref{lem:walled_brauer_extension_rule}, and~\cref{lem:efficient_add_a_point} correspond to the different possible factorizations in Step 1 with $\pn(D_a) = 0$.  
In addition to these, we also give an efficient implementation of extending a diagram by multiplication with a bridge generator, which corresponds to the factorizations \textbf{Partition 2} and \textbf{Partition 3}. Since $D_b$ may not have propagating number zero, this lemma is somewhat more involved. Nevertheless, we are still able to obtain an efficient implementation.
\begin{lemma}
    \label{lem:efficient_add_a_bridge}
    Let $B = P_{n-1}(d)$, $A = P_n(d)$. For any basis element $b = d^{({n - 1 - \cc(D)})/2}D \in \mathcal{B}(B)$, the maps 
    \begin{equation}
        \label{eq:bridge_add_one}
        \mathbf{G}_1(\wt{\ft_B}\ket{b}) = \wt{\ft_A}\ket{d^{(n - \cc(D))/2} b_{n-1}D}
    \end{equation}
    \begin{equation}
    \label{eq:bridge_add_two}
        \mathbf{G}_2(\wt{\ft_B}\ket{b}) = \wt{\ft_A}\ket{d^{(n - \cc(D))/2} Db_{n-1}}
    \end{equation}
    can be implemented to operator norm error $O(\poly(|A|) \cdot (\delta + \varepsilon))$ by operators $\wt{\mathbf{G}}_1$ and $\wt{\mathbf{G}}_2$. Both $\wt{\mathbf{G}}_1$ and $\wt{\mathbf{G}}_2$ can be implemented with $\wt{O}(n^2 \cdot (n + \log d + \log(1/\varepsilon)))$ gates. 
\end{lemma}
Like~\cref{lem:efficient_add_a_point}, the outputs of~\cref{eq:bridge_add_one},~\cref{eq:bridge_add_two} are Fourier transforms of states in $\mathcal{B}(A)$—the number of connected components does not change, since the bridge generator will identify $\{n, n^\prime\}$ with some connected component of $D$.

\subsection{Step 5: Accumulate}
\label{sec:accumulate}
Finally, we are ready to apply the results of the previous two sections to update the state from~\cref{eq:after_recurse}. This step generalizes the ``sum over cosets'' procedure used in \cite{beals1997qft, moore2003genericquantumfouriertransforms}, with modifications to account for the presence of non-invertible transversals. For convenience, we copy the state from~\cref{eq:after_recurse} below:
\begin{equation}
    \ket{\psi_2} =   \sum_{w_1, w_2} \sum_{b \in \mathcal{F}_{(w_1, w_2)}} f(w_1bw_2) \ket{w_1}_{\mathsf{L}} \otimes   \ket{w_2}_{\mathsf{R}} \otimes \wt{\ft_B}\ket{D_b}_\mathsf{S}
\end{equation}
The procedure for Step 5 is described in~\cref{alg:accumulate} and~\cref{fig:postprocess}.
In the algorithm, we use $\ket{\perp}$ to denote a fixed symbol which is distinct from any possible valid word $\ket{w}$ in the generators of $A$. 

\begin{figure}[H]
    \centering
    \includegraphics[width=\linewidth]{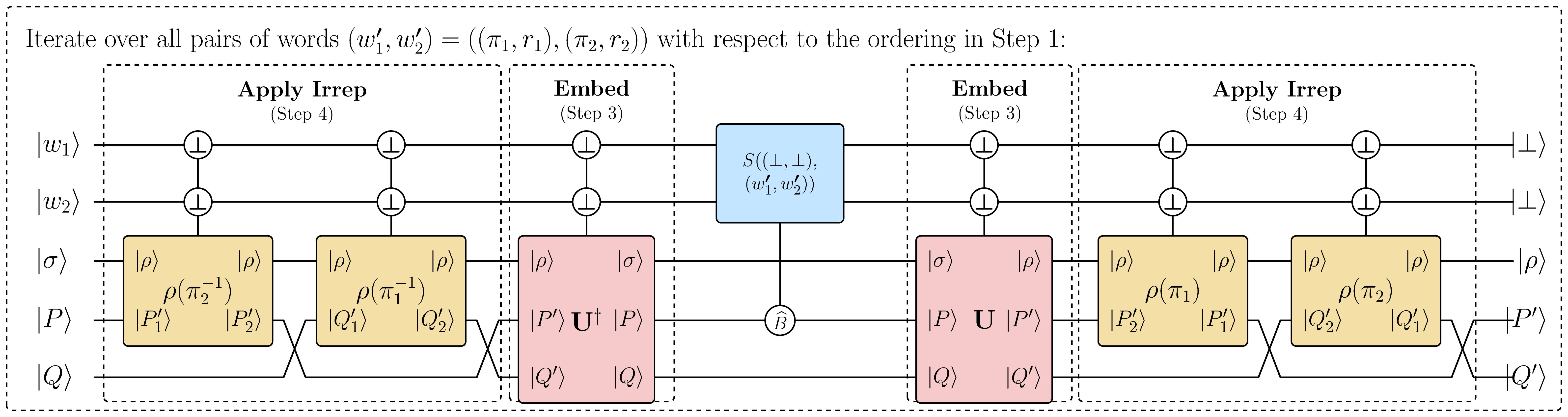}
    \caption{The circuit implementing~\cref{alg:accumulate}, equivalent to the postprocessing gate in~\cref{fig:algorithm_sketch}. 
    In this circuit diagram, the loop variables iterate over all possible transversals in the order defined in~\cref{def:factorings_for_our_algebras}, and are denoted by \(w_1^\prime\) and \(w_2^\prime\) to distinguish them from the input registers \(\ket{w_1}\) and \(\ket{w_2}\). The controlled $S((\perp,\perp), (w_1^\prime, w_2^\prime))$ is a gate which exchanges the basis vectors $\ket{(\perp,\perp)}$ and $\ket{(w_1^\prime, w_2^\prime)}$ (Line 13 of~\cref{alg:accumulate}), conditioned on the middle three registers containing a Fourier state of $B$ (this can be done by checking the length of either $\ket{P}$ or $\ket{Q}$; we have chosen to use the former Bratteli path above). 
    At different points in the circuit, the last three wires may be in a superposition of states containing an irrep $\sigma \in \wh{B}$ and states containing an irrep $\rho \in \wh{A}$, along with their corresponding Bratteli paths. However, note that a controlled-$\rho(\pi)$ gate is only ever applied on the part of the superposition where these registers contain irreps and Bratteli paths of $A$ (since the rest of the superposition will not have $\bot$ in the control registers).}
    \label{fig:postprocess}
\end{figure}

\begin{algorithm}[t]
\caption{Accumulate: Obtaining $\wt{\ft_A}$ from $\wt{\ft_B}$}
\label{alg:accumulate}
\begin{algorithmic}[1]
\For{\textbf{each} pair of words $(w_1^\prime, w_2^\prime)$, with respect to the last possible ordering in~\cref{def:factorings_for_our_algebras}}
    \State Split $w_1^\prime$ and $w_2^\prime$ into a permutation and non-permutation: $w_1 = (\pi_1, r_1)$ and $w_2 = (\pi_2, r_2)$.
    \If{$e_{n-1}$, $f_r$, or $p_n$ is in $r_2$}
        \State \texttt{PropZero = True}
    \Else
        \State \texttt{PropZero = False}
    \EndIf
    \If{Both $\mathsf{L}$ and $\mathsf{R}$ are equal to $\ket{\perp}$}
        \State Apply a controlled-$\rho(\pi_1^{-1})$ gate to the row register and a controlled-$\rho(\pi_2^{-1})$ gate to the column register.
        \State Let $\mathbf{U}$ be the embedding operator specified by \Cref{tab:embedding_operators}. Apply $\wt{\mathbf{U}}^\dagger$ to $\mathsf{S}$.
    \EndIf
    \If{$\mathsf S$ contains a Fourier state of $B$}
         \State Apply the operation that exchanges $\ket{\perp}_{\mathsf L}\otimes \ket{\perp}_{\mathsf R}$ with $\ket{w_1^\prime}_{\mathsf L}\otimes \ket{w_2^\prime}_{\mathsf R}$, while fixing all other basis vectors.
    \EndIf
     \If{Both $\mathsf{L}$ and $\mathsf{R}$ are equal to $\ket{\perp}$}
         \State Let $\mathbf{U}$ be the embedding operator specified by \Cref{tab:embedding_operators}. Apply $\wt{\mathbf{U}}$ to $\mathsf{S}$.
        \State Apply a controlled-$\rho(\pi_1)$ gate to the row register and a controlled-$\rho(\pi_2)$ gate to the column register.
    \EndIf
\EndFor
\end{algorithmic}
\end{algorithm}

\begin{table}[t]
    \centering
    \renewcommand{\arraystretch}{1.15}
    \begin{tabular}{|p{0.58\textwidth}|p{0.34\textwidth}|}
        \hline
        \textbf{Case} & \textbf{Embedding Operator ($\mathbf{U}$)} \\
        \hline
        $\neg\texttt{PropZero}$, $A = B_n(d)$, $n$ even
        & $\mathbf{E}(B_n(d), B_{n-2}(d))$ \\
        \hline

        $\neg\texttt{PropZero}$, $A = B_n(d)$, $n$ odd
        & $\mathbf{E}(B_n(d), B_{n-1}(d))$ \\
        \hline

        $\neg\texttt{PropZero}$, $A = B_{r,s}(d)$, $s-r \ge 1$ odd
        & $\mathbf{E}(B_{r,s}(d), B_{r,s-1}(d))$ \\
        \hline

        $\neg\texttt{PropZero}$, $A = B_{r,r}(d)$
        & $\mathbf{E}(B_{r,r}(d), B_{r-1,r-1}(d))$ \\
        \hline

        $\neg\texttt{PropZero}$, $A = P_n(d)$, $b_{n-1} \notin r_1$, $b_{n-1} \notin r_2$
        & $\mathbf{E}(P_n(d), P_{n-1}(d))$ \\
        \hline

        $\neg\texttt{PropZero}$, $A = P_n(d)$, $b_{n-1} \in r_1$, $b_{n-1} \notin r_2$
        & $\mathbf{G}_1(P_n(d), P_{n-1}(d))$ \\
        \hline

        $\neg\texttt{PropZero}$, $A = P_n(d)$, $b_{n-1} \notin r_1$, $b_{n-1} \in r_2$
        & $\mathbf{G}_2(P_n(d), P_{n-1}(d))$ \\
        \hline

        $\texttt{PropZero}$, $A = B_n(d)$
        & $\mathbf{F}(B_n(d), B_{n-2}(d))$ \\
        \hline

        $\texttt{PropZero}$, $A = B_{r,r}(d)$
        & $\mathbf{F}(B_{r,r}(d), B_{r-1,r-1}(d))$ \\
        \hline

        $\texttt{PropZero}$, $A = P_n(d)$, $b_{n-1} \notin r_1$, $b_{n-1} \notin r_2$
        & $\mathbf{F}_1(P_n(d), P_{n-1}(d))$ \\
        \hline

        $\texttt{PropZero}$, $A = P_n(d)$, $b_{n-1} \in r_1$, $b_{n-1} \notin r_2$
        & $\mathbf{F}_2(P_n(d), P_{n-1}(d))$ \\
        \hline

        $\texttt{PropZero}$, $A = P_n(d)$, $b_{n-1} \notin r_1$, $b_{n-1} \in r_2$
        & $\mathbf{F}_3(P_n(d), P_{n-1}(d))$ \\
        \hline

        $\texttt{PropZero}$, $A = P_n(d)$, $b_{n-1} \in r_1$, $b_{n-1} \in r_2$
        & $\mathbf{F}_4(P_n(d), P_{n-1}(d))$ \\
        \hline
    \end{tabular}
    \caption{The choice of embedding operation to use during each iteration of~\cref{alg:accumulate}, which depends on the factors $r_1$ and $r_2$. In the algorithm, we do not apply the embedding operator $\mathbf{U}$ or $\mathbf{U}^\dagger$ directly, but rather the approximate embedding $\wt{\mathbf{U}}$ or $\wt{\mathbf{U}}^\dagger$.}
    \label{tab:embedding_operators}
\end{table}
\begin{lemma}
    Up to $O(\poly(|A|) \cdot (\delta + \varepsilon))$ error, \cref{alg:accumulate} transforms $\ket{\psi_2}$ into 
\begin{equation}
    \label{eq:after_accumulate}
     \sum_{a \in \mathcal{B}(A)} f(a) \ket{\perp}_\mathsf{L} \otimes \ket{\perp}_{\mathsf{R}} \otimes \wt{\ft_A}\ket{D_a}_{\mathsf{S}} 
\end{equation}
which is $\wt{\ft_A}\ket{\psi_0}$ with two unentangled ancilla registers in the $\ket{\perp}$ state. 
\end{lemma}
\begin{proof}
    To simplify the analysis, we consider the case where the input $\ket{\psi_0}$ is a single diagram $\ket{D_a}$ (the rest will follow by linearity). In this case, 
\begin{equation}
    \ket{\psi_2} =   \ket{w_1}_{\mathsf{L}} \otimes   \ket{w_2}_{\mathsf{R}} \otimes \wt{\ft_B}\ket{D_b}_\mathsf{S}
\end{equation}
where $D_a = w_1D_bw_2$. 

Now, we begin iterating over the ordered set of transversals as in~\cref{alg:accumulate}. 
This is akin to guessing each possible transversal, one by one, in this order, until we happen to correctly guess the correct transversal $(w_1, w_2)$.
There are three stages to analyze: the iterations \textit{before} guessing the correct transversal, the iteration in which we guess the correct transversal, and the subsequent iterations. The idea is that iterations in which we have guessed incorrectly should have negligible effect, and the iteration in which we guess correctly should promote $\wt{\ft_B}\ket{D_b}$ to $\wt{\ft_A} \ket{w_1D_bw_2} = \wt{\ft_A} \ket{D_a}$.

Until we reach the correct transversal $(w_1^\prime, w_2^\prime) = (w_1, w_2)$, nothing happens, as none of the controls are set to $\ket{\bot}$, and the controlled applications of $S((\perp,\perp), (w_1^\prime, w_2^\prime))$ are exchanging basis states with zero amplitude. 

Once we reach the transversal $(w_1, w_2)$ corresponding to the factoring of $D_a$, the $S((\perp,\perp), (w_1^\prime, w_2^\prime))$ operator replaces the state on the $\mathsf L$ and $\mathsf R$ registers with the state $\ket{\perp}_{\mathsf L}\otimes \ket{\perp}_{\mathsf R}$. After this, we apply the embedding operator $\mathbf{U}$, chosen from~\cref{tab:embedding_operators} for the relevant algebra according to $(w_1^\prime, w_2^\prime)$, as well as the controlled $\rho(\pi_1)$ and $\rho(\pi_2)$ irrep matrices, transforming the state into 
\begin{equation}
    \wt{\ft_A}\ket{w_1D_bw_2} = \wt{\ft_A}\ket{D_a}
\end{equation}
up to error $O(\poly(|A|) \cdot (\delta + \varepsilon))$ (see the analysis in~\cref{sec:embedding} and~\cref{sec:apply_irreps}). 

Finally, we argue that in subsequent iterations of the for loop, corresponding to transversals $(w_3^\prime, w_4^\prime) > (w_1, w_2)$, the remaining gates act nearly as identity, and do not significantly alter the state. 
To do so, we first define the operator $T_{(w_3^\prime, w_4^\prime)}$ to be the ideal forward operation implemented by Lines 15--17 for this
choice of transversal, i.e. the combination of the embedding step and the controlled-$\rho(\pi_1)$ and controlled-$\rho(\pi_2)$ steps. 
That is, 
\begin{align}
    T_{(w_3^\prime, w_4^\prime)}
    \;
    \wt{\ft_B}\ket{D_c} = \wt{\ft_A}\ket{w_3^\prime D_cw_4^\prime}, \quad \forall  D_c\in\mathcal D(B)
\end{align}
With this definition, when we apply $T_{(w_3^\prime, w_4^\prime)}$ to the image of the Fourier transform on $B$, we get
\begin{equation}
     \label{eq:accumulate_one}
    T_{(w_3^\prime, w_4^\prime)}
    \left(
        \operatorname{im} \wt{\ft_B}
    \right)
    =
    \operatorname{span}
    \left\{
        \wt{\ft_A}\ket{w_3^\prime D_cw_4^\prime}
        :
        D_c\in\mathcal D(B)
    \right\}
\end{equation}

Since $(w_3^\prime, w_4^\prime) > (w_1, w_2)$ (\cref{def:ordered_factor}), we have that $w_3^\prime D_cw_4^\prime \ne D_a$ for all $D_c \in \mathcal{D}(B)$ by the last possible factorization property (\cref{def:factorings_for_our_algebras}). Moreover, from~\cref{eq:new_gram_matrix}, the Gram matrix $G_{\wt{\ft_A}}$ of the basis $\{\ft_A\ket{D_a}\}_{a \in \mathcal{D}(A)}$ with respect to the computational inner product satisfies
\begin{equation}
    \label{eq:accumulate_two}
    ||G_{\wt{\ft_A}} - I||_\infty \le O(\poly(|A|) \cdot \delta)
\end{equation}
From~\cref{eq:accumulate_one} and~\cref{eq:accumulate_two}, it follows that $\wt{\ft_A}\ket{D_a}$ has at most $O(\poly(|A|) \cdot \delta)$ overlap with the image of $T_{(w_3^\prime, w_4^\prime)}$, i.e.
\begin{align}
    \left\|
        \Pi_{T_{(w_3^\prime, w_4^\prime)}(\operatorname{im} \wt{\ft}_B)}
       \wt{\ft_A}\ket{D_a}
    \right\| = & \sum_{D_c \in \mathcal{D}(B)} \left| \braket{w_3^\prime D_cw_4^\prime | \wt{\ft}_A^\dagger \wt{\ft}_A |D_a} \right|^2 \\
    \le & 
    \;O(|B| \cdot \poly(|A|) \cdot \delta) \\
    = &
   \; O(\poly(|A|) \cdot \delta)
\end{align}
and so up to $O(\poly(|A|) \cdot \delta)$ error, $\wt{\ft_A}\ket{D_a}$ is orthogonal to the subspace $T_{(w_3^\prime, w_4^\prime)}(\operatorname{im}\wt{\ft}_B)$. Up to this error, $T^\dagger_{(w_3^\prime, w_4^\prime)}\wt{\ft_A}\ket{D_a}$ is orthogonal to the subspace corresponding to the image of $\wt{\ft}_B$, which is spanned by all basis states $\ket{\sigma, P, Q}$ for the subalgebra $B$. Therefore, the state after Line 11 has no support on any basis states $\ket{\sigma, P, Q}$ (again up to the aforementioned error), and so the control on Line 12 will not trigger. Hence, up to $O(\poly(|A|) \cdot \delta)$ error, the effect of the entire iteration will be to apply $T_{(w_3^\prime, w_4^\prime)}T_{(w_3^\prime, w_4^\prime)}^\dagger  = I$, which preserves the state. Repeating this argument for all remaining iterations of the for loop (at most $\poly(n) = \polylog(|A|)$ iterations), the final state of the algorithm is
\begin{equation}
      \ket{\perp}_\mathsf{L} \otimes \ket{\perp}_{\mathsf{R}} \otimes \wt{\ft_A}\ket{D_a}_{\mathsf{S}} 
\end{equation}
up to $O(\poly(|A|) \cdot (\delta + \varepsilon))$ error (combining the error from all stages of the algorithm). By linearity, we recover the state in~\cref{eq:after_accumulate} when the input is an arbitrary superposition of diagrams in $\mathcal{D}(A)$. 
\end{proof}

\paragraph{Error Analysis.} Each iteration of the for loop incurs an operator norm error of $O(\poly(|A|) \cdot (\delta + \varepsilon))$. There are at most $O(n^4)$ transversals for the Brauer and walled Brauer algebras, and $O(n^3)$ transversals for the partition algebra  (see~\cref{def:factorings_for_our_algebras}). Since $\poly(n) = \polylog(|A|)$, the total operator norm error for the entire algorithm is also $O(\poly(|A|) \cdot (\delta + \varepsilon))$.

\paragraph{Gate Complexity.} First, we need to bound the gate complexity of a single iteration of the loop. Within each iteration, we apply a controlled-$\mathbf{U}$, controlled-$\mathbf{U}^\dagger$, and four controlled applications of $\rho(\pi)$. There are also some extra processing steps to compute the values of certain conditional expressions, which have lower-order gate complexity. 

In all cases, $\pi$ is the product of at most two transpositions, each of which can be (efficiently) factored as a product of $O(n)$ swap generators $s_i$. Using~\cref{lem:swap_application} to implement each controlled-$\rho(s_i)$, we obtain the gate complexities for each controlled-$\rho(\pi)$:
\begin{itemize}
    \item Brauer Algebra: $\wt{O}(n^{2} \cdot (n + \log d + \log(1/\varepsilon)))$
    \item Walled Brauer Algebra: $\wt{O}(n^2 \cdot (n + \log d + \log(1/\varepsilon)))$ 
    \item Partition Algebra: $\wt{O}(n^{7/2} \cdot (n + \log d + \log(1/\varepsilon)))$ 
\end{itemize}
The embedding operator $\wt{\mathbf{U}}$ has the following gate complexity:
\begin{itemize}
    \item Brauer Algebra: $\wt{O}(\sqrt{n} \cdot (n^2\log d + \log(1/\varepsilon)))$ (\cref{lem:embedding_summary})
    \item Walled Brauer Algebra: $\wt{O}(\sqrt{n} \cdot (n^2\log d + \log(1/\varepsilon)))$ (\cref{lem:embedding_summary})
    \item Partition Algebra: $\max\left\{\wt{O}(\sqrt{n} \cdot (n^2\log d + \log(1/\varepsilon))), \;\wt{O}(n^2 \cdot (n + \log d + \log(1/\varepsilon)))\right\}$ (\cref{lem:embedding_summary}, \cref{lem:efficient_add_a_bridge})
\end{itemize}
In total, we have gate complexity at most:\footnote{These upper bounds are quite loose. Since we do not attempt to optimize the algorithm, we expect that these could be significantly improved with further analysis, which we leave to future work.}
\begin{itemize}
    \item Brauer Algebra: $\wt{O}(n^{5/2} \cdot (\sqrt{n} + \log d + \log(1/\varepsilon)))$
    \item Walled Brauer Algebra: $\wt{O}(n^{5/2} \cdot (\sqrt{n} + \log d + \log(1/\varepsilon)))$ 
    \item Partition Algebra: $\wt{O}(n^{7/2} \cdot (n + \log d + \log(1/\varepsilon)))$ 
\end{itemize}
Finally, multiplying by the total number of transversals ($O(n^4)$ for the Brauer and walled Brauer algebras, and $O(n^3)$ for the partition algebra, see~\cref{def:factorings_for_our_algebras}) gives the final gate counts of~\cref{alg:accumulate}:
\begin{center}
\fbox{
  \parbox{0.92\linewidth}{
    \centering
        \textbf{Gate Complexity For $B_n(d), B_{r,s}(d)$: $\wt{O}(n^{13/2}\cdot (\sqrt{n} + \log d + \log(1/\varepsilon)))$} \\
        \textbf{Gate Complexity For $P_n(d)$: $\wt{O}(n^{13/2}\cdot (n + \log d + \log(1/\varepsilon)))$} \\
      \textbf{Operator Norm Error: $O(\poly(|A|) \cdot (\delta + \varepsilon))$}
  }
}
\end{center}
\subsection{An Efficient Quantum Fourier Transform for Diagram Algebras}
Finally, we combine the boxed results of Steps 1, 2, and 5 from~\cref{sec:factor},~\cref{sec:recurse}, and~\cref{sec:accumulate}. We begin with the Brauer algebra, where the following recurrence holds when $n$ is even: 
\begin{equation}
    \label{eq:brauer_recurrence}
     \text{size}(\wt{\ft_{B_{n}(d)}}) = \text{size}(\wt{\ft_{B_{n - 2}(d)}}) + \wt{O}(n^{13/2}\cdot (\sqrt{n} + \log d + \log(1/\varepsilon)))
\end{equation}
\begin{equation}
     \label{eq:brauer_recurrence_error}
    \text{error}(\wt{\ft_{B_{n}(d)}}) = \text{error}(\wt{\ft_{B_{n - 2}(d)}}) + O(\poly(|B_{n}(d)|) \cdot (\delta + \varepsilon))
\end{equation}
When $n$ is odd, $\ft_{B_{n}(d)}$ instead reduces to $\ft_{B_{n - 1}(d)}$ on the first recursion level, but since $n$ then becomes even, the asymptotics of the recurrence relations will remain the same. When $n = 2$, $\wt{\ft_{B_{n}(d)}}$ is a two-dimensional operator, and therefore has constant gate complexity. Solving the recurrences in~\cref{eq:brauer_recurrence} and~\cref{eq:brauer_recurrence_error},
\begin{equation}
     \text{size}(\wt{\ft_{B_{n}(d)}}) =  \wt{O}(n^{15/2}\cdot (\sqrt{n} + \log d + \log(1/\varepsilon))),\;\; \text{error}(\wt{\ft_{B_{n}(d)}}) =   O(\poly(|B_{n}(d)|) \cdot (\delta + \varepsilon))
\end{equation}
Similar recurrence relations yield analogous statements for the walled Brauer algebra and partition algebra:
\begin{equation}
     \text{size}(\wt{\ft_{B_{r,s}(d)}}) =  \wt{O}(n^{15/2}\cdot (\sqrt{n} + \log d + \log(1/\varepsilon))),\;\; \text{error}(\wt{\ft_{B_{r, s}(d)}}) =   O(\poly(|B_{r,s}(d)|) \cdot (\delta + \varepsilon))
\end{equation}
\begin{equation}
     \text{size}(\wt{\ft_{P_{n}(d)}}) =  \wt{O}(n^{15/2}\cdot (n + \log d + \log(1/\varepsilon))),\;\; \text{error}(\wt{\ft_{P_{n}(d)}}) =   O(\poly(|P_{n}(d)|) \cdot (\delta + \varepsilon))
\end{equation}
Finally,~\cref{thm:approx_qft} allows us to replace $\wt{\ft_A}$ with $\ft_A$, at the cost of additional $\poly(|A|)$ factors in the operator norm error. We also substitute $\delta = O(|A| \cdot d^{-1/2})$ from \cref{cor:d_nice_algebras}. We now state the main result of the paper:  
\begin{theorem}
    \label{thm:qft_is_efficient_to_implement}
    Let $A \in \{B_n(d), B_{r,s}(d), P_n(d)\}$.%
    \footnote{
        \cref{thm:qft_is_efficient_to_implement} also extends to the half partition algebra $P_{n - \frac12}(d)$. While we omit it from~\cref{thm:qft_is_efficient_to_implement}, the proof is essentially identical to the $P_n(d)$ case.
    }
    The Fourier transform $\ft_A$ can be implemented on a quantum computer up to operator norm error
    \begin{equation}
        O\!\left(\poly(|A|)\cdot \bigl(d^{-1/2} + \varepsilon\bigr)\right),
    \end{equation}
    by quantum circuits with gate complexity
    \begin{equation}
        \wt{O}\!\left(n^{15/2}\cdot (\sqrt{n} + \log d + \log(1/\varepsilon))\right)
    \end{equation}
    for $B_n(d)$ and $B_{r,s}(d)$, and
    \begin{equation}
        \wt{O}\!\left(n^{15/2}\cdot (n + \log d + \log(1/\varepsilon))\right)
    \end{equation}
    for $P_n(d)$. The Fourier transform is performed with respect to the subalgebra adapted bases from~\cref{sec:subalgebra_chains}, and the orthogonal form in~\cref{sec:main_body_orthogonal_form}.
\end{theorem}

\ifanon
\else 
\section*{Acknowledgements}
We thank Dmitry Grinko, Maris Ozols, Alex Lombardi, and Henry Yuen for helpful discussions. 
This work was done in part while BN was visiting the Simons Institute for the Theory of Computing, supported by NSF QLCI Grant No. 2016245.
BN acknowledges the following funding sources: AFOSR award FA9550-23-1-0363, NSF awards CCF-2530159, CCF-2144219, and CCF-2329939, and the Sloan Foundation.
YD acknowledges partial support by the National Science Foundation (under awards CCF-2312754 and CCF-2338063), by the U.S. Department of Energy, Office of Science, National Quantum Information Science Research Center, Co-design Center for Quantum Advantage (C2QA) under Contract No. DE-SC0012704, by QuantumCT (under NSF Engines award ITE-2302908), by AFOSR MURI (FA9550-26-1-B036), by Boehringer Ingelheim, and NSF NQVL-ERASE (under award OSI-2435244). External interest disclosure: YD is a consultant and equity holder of D-Wave Quantum, Inc.
\fi

\paragraph{AI Statement.}
During the course of this work, we used OpenAI's Codex (with GPT 5.3 and GPT 5.4\ Thinking) to build a graphical interface to visualize different properties of Fourier states, including software for calculating Fourier basis matrix elements. This software was used to build intuition that informed several of the proofs, particularly those in Section 4
\ifanon
\!.
\else
\!, and is freely available.\footnote{ \url{https://github.com/Ben-Foxman/semisimple_algebra_visualization}}
\fi
In the second version of the paper, we used GPT-6 Astra to help simplify our original proofs of Lemma 6.5 and Lemma E.6. The correctness and originality of the results in this paper were verified independently of any AI assistance.

\printbibliography
\newpage
\appendix
\begin{appendices}
\section{Matrix Inequalities}
We use $||M||_{\infty}$ to denote the operator norm induced from the $2$-norm, equivalent to the maximum singular value. We use the following matrix inequalities.
\begin{fact}
    \label{lem:operator_norm_of_inverse}
    Assume that $||G-I||_{\infty} \le \varepsilon < 1$. Then $G$ is invertible, and $||G^{-1} - I||_{\infty} \le \frac{\varepsilon}{1 - \varepsilon}$. 
\end{fact}
\begin{proof}
    If $G$ is non-invertible, then $||G - I||_\infty \ge 1$. Therefore $G$ is invertible, and 
    \begin{equation}
        ||G^{-1} - I||_{\infty} = ||G^{-1}(I - G)||_{\infty} \le ||G^{-1}||_{\infty}||G - I||_{\infty} \le \varepsilon ||G^{-1}||_\infty 
    \end{equation}
    Now, since $||G-I||_\infty \le \varepsilon$, the singular values of $G$ lie in $[1 - \varepsilon, 1 + \varepsilon]$, so the singular values of $G^{-1}$ lie in $\left[\frac{1}{1 + \varepsilon}, \frac{1}{1- \varepsilon} \right]$. Hence,
    \begin{equation}
        ||G^{-1}||_\infty \le \frac{1}{1 - \varepsilon}
    \end{equation}
    finishing the proof.
\end{proof}
\begin{fact}
    \label{lem:off_diagonal_sums}
    If $||G - I||_\infty \le \varepsilon$, then the sum of the magnitudes of all matrix elements in a column of $G - I$ is at most $\varepsilon\cdot \sqrt{d}$, where $d$ is the dimension of $G - I$. 
\end{fact}
\begin{proof}
    The $k$th column of $G - I$ is equivalent to $(G - I)e_k$. Using $v_k$ to denote this column, we have
    \begin{equation}
        ||v_k||_2 \le \varepsilon \implies \sum_{i} |(v_k)_i|^2 \le \varepsilon^2
    \end{equation}
    Using Cauchy-Schwarz, 
    \begin{equation}
        \sum_{i} |(v_k)_i| \le \sqrt{d} \cdot  \sqrt{\sum_{i} |(v_k)_i|^2} \le \sqrt{d} \cdot \varepsilon
    \end{equation}
\end{proof}

\begin{fact}
\label{lem:opnorm_bound}
Let $M$ be a linear operator on $\mathbb{C}^d$ with $|M_{ij}| \le \varepsilon$ for all $i, j$. Then, $||M||_{\infty} \le \varepsilon d$.
\end{fact}
\begin{proof}
Let \(\|\cdot\|_{\mathrm{row}\text{-}\infty}\) denote the maximum absolute row sum. Since \(|M_{ij}| \le \varepsilon\) for all \(i,j\), we have
\begin{equation}
    \|M\|_{1} \le \varepsilon d
    \qquad\text{and}\qquad
    \|M\|_{\mathrm{row}\text{-}\infty} \le \varepsilon d.
\end{equation}
By the standard inequality
\begin{equation}
    \|M\|_{2} \le \sqrt{\|M\|_{1}\|M\|_{\mathrm{row}\text{-}\infty}},
\end{equation}
it follows that
\begin{equation}
    \|M\|_{\infty} = \|M\|_{2} \le \sqrt{(\varepsilon d)(\varepsilon d)} = \varepsilon d.
\end{equation}
\end{proof}
\begin{lemma}
\label{lem:close_gram_implies_close_norms}
Let $\langle\cdot,\cdot\rangle_g$ and $\langle\cdot,\cdot\rangle_h$ be two inner products on $\mathbb{C}^d$, and let $B=\{b_1,\dots,b_d\}$ and $C=\{c_1,\dots,c_d\}$ be two bases of $\mathbb{C}^d$. Assume that the Gram matrices $G_g(B)=G_h(C)=I$, and that for some $c>0$,
\begin{equation}
\left\|\frac{G_h(B)}{c}-I\right\|_\infty \le \varepsilon.
\end{equation}
Then, for $\varepsilon<1$,
\begin{equation}
\|c\,G_g(C)-I\|_\infty \le \frac{\varepsilon}{1-\varepsilon}.
\end{equation}
\end{lemma}
\begin{proof}
Let \(B,C\) be the basis matrices (with columns containing the respective basis vectors) and set \(X=B^{-1}C\), so \(C=BX\). Let \(H=G_h(B)\) and \(H'=H/c\).
Then
\begin{equation}
G_g(C)=G_g(BX)=X^\dagger G_g(B)X=X^\dagger X
\end{equation}
and
\begin{equation}
I=G_h(C)=G_h(BX)=X^\dagger HX=c\,X^\dagger H'X,
\end{equation}
so \(X^\dagger H'X=\frac{1}{c}I\). By the Courant-Fischer Theorem (\cite{Watrous2018TQI}, Theorem 1.2), \(\|H'-I\|_\infty\le \varepsilon\) implies that
\begin{equation}
(1-\varepsilon)I \preceq H' \preceq (1+\varepsilon)I.
\end{equation}
Conjugating by \(X\) gives
\begin{equation}
(1-\varepsilon)X^\dagger X \preceq X^\dagger H'X = \frac{1}{c}I \preceq (1+\varepsilon)X^\dagger X,
\end{equation}
hence
\begin{equation}
\frac{1}{1+\varepsilon}I \preceq c\,X^\dagger X \preceq \frac{1}{1-\varepsilon}I.
\end{equation}
Therefore all eigenvalues of \(c\,G_g(C)=c\,X^\dagger X\) lie in \([\frac{1}{1+\varepsilon},\frac{1}{1-\varepsilon}]\), and so
\begin{equation}
\|c\,G_g(C)-I\|_\infty
\le \max\!\left\{1-\frac{1}{1+\varepsilon},\,\frac{1}{1-\varepsilon}-1\right\}
= \frac{\varepsilon}{1-\varepsilon}.
\end{equation}
\end{proof}

\section{Proofs from Section~\ref{sec:properties_of_irreps}}
\label{app:properties_of_irreps}
\begin{theorem}
    For all $A \in \{P_n(d), P_{n - \frac12}(d),  B_n(d), B_{r,s}(d)\}$, and $0 \le k \le n$, 
    \begin{equation}
        |\{D \in \mathcal{D}(A): \pn(D) = k\}| = \sum_{\rho \in \wh{A}: |\rho| = k} d_\rho^2
    \end{equation}
\end{theorem}
\begin{proof}
    We begin by giving a general way for constructing a partition diagram $D \in \mathcal{D}(P_n(d))$ with $\pn(D) = k$. First, choose a partition of the top $n$ vertices into $i \ge k$ blocks. There are $\stirlingii{n}{i}$ choices. Among these $i$ blocks, choose the $k$ blocks that will become propagating; this can be done in $\binom{i}{k}$ ways. Similarly, choose a partition of the bottom $n$ vertices into $j \ge k$ blocks, which can be done in $\stirlingii{n}{j}$ ways, and then choose which $k$ of those $j$ blocks will be propagating; this can be done in $\binom{j}{k}$ ways. 
    
    Now, form the propagating blocks by pairing the chosen $k$ top-blocks with the chosen $k$ bottom-blocks via a bijection. There are $k!$ possible bijections. 

    This construction uniquely specifies a diagram $D$ with propagating number $k$ such that $i$ blocks intersect the top row and $j$ blocks intersect the bottom row. Therefore, for fixed $i$ and $j$ the number of diagrams with $\pn(D) = k$ is 
    \begin{equation}
        \label{eq:initial_partition_count_level_k}
    k!\,\binom{i}{k}\binom{j}{k}\,\stirlingii{n}{i}\,\stirlingii{n}{j}.
    \end{equation}
    Summing over all $i,j$ with $k\le i,j\le n$ gives
    \begin{equation}
     |\{D \in \mathcal{D}(P_n(d)): \pn(D) = k\}|
    =
    k!\sum_{i=k}^{n}\sum_{j=k}^{n}
    \binom{i}{k}\binom{j}{k}\, \stirlingii{n}{i}\,\stirlingii{n}{j},
    \end{equation}
    Also, 
    \begin{equation}
        \label{eq:sum_of_squares_of_dims_for_partition}
         \sum_{\lambda \in \wh{A}: |\lambda| = k} d_\lambda^2 =  \sum_{|\lambda| = k} \left(f^\lambda\sum_{l = k}^n \stirlingii{n}{l}\binom{l}{k} \right)^2 =  \left(\sum_{l = k}^n \stirlingii{n}{l}\binom{l}{k}\right)^2\sum_{|\lambda| = k} \left(f^\lambda\right)^2 = \left(\sum_{l = k}^n \stirlingii{n}{l}\binom{l}{k}\right)^2 k!
    \end{equation}
    with the first equality by \cref{eq:irreps_of_parition_algebra}, and the last by \cref{eq:dim_identity_for_symmetric_group}. Comparing the previous two equations, we see that the claimed equality holds for $P_n(d)$. For $P_{n-\frac12}(d)$, we can repeat almost the same analysis, except that the block containing $n$ and $n^\prime$ is not counted as propagating. Thus, if the top row is partitioned into $l+1$ blocks and the bottom row into $m+1$ blocks, then among the $l$ remaining top-blocks and the $m$ remaining bottom-blocks we choose $k$ propagating blocks, and pair them by a bijection. Therefore the number of diagrams with $\pn(D)=k$ is
    \begin{equation}
    k!\sum_{l=k}^{n-1}\sum_{m=k}^{n-1}
    \binom{l}{k}\binom{m}{k}\stirlingii{n}{l+1}\stirlingii{n}{m+1}.
    \end{equation}
    The remainder of the argument is identical to the $P_n(d)$ case.
    
    Now, we move on to the Brauer algebra $B_n(d)$. We can construct any $D \in \mathcal{D}(B_n(d))$ with $\pn(D) = k$ as follows. First, choose one of $\binom{n}{k}$ possible subsets of vertices on the top row to be in propagating blocks, and do the same on the bottom. Form the propagating blocks by choosing one of $k!$ possible bijections between these sets of $k$ vertices. Then, there are $n - k$ vertices in both the top row and the bottom row belonging to non-propagating blocks (of size 2). Within each row, these vertices can be paired up in $(n - k - 1)!!$ ways. Hence, 
    \begin{equation}
     |\{D \in \mathcal{D}(B_n(d)): \pn(D) = k\}|
    =
    k!\sum_{i=k}^{n}\sum_{j=k}^{n}
    \binom{i}{k}\binom{j}{k}((n - k - 1)!!)^2
    \end{equation}
    A calculation analogous to \cref{eq:sum_of_squares_of_dims_for_partition} proves the result for $B_n(d)$. 
    
    Finally, for $B_{r,s}(d)$, let $\ell = \frac{n-k}{2}$, where $n=r+s$. If $n-k$ is odd, then no such diagrams exist, and no irreps with $n-k$ total boxes exist. Otherwise, in each row we choose $\ell$ vertices on the left and $\ell$ vertices on the right to belong to non-propagating blocks, and pair them by a bijection. This can be done in
    \begin{equation}
        \binom{r}{\ell}\binom{s}{\ell}\ell!
    \end{equation}
ways per row. The remaining $r-\ell$ left vertices and $s-\ell$ right vertices in the top row must then be matched vertically to the corresponding vertices in the bottom row, which can be done in $(r-\ell)!(s-\ell)!$ ways. Therefore,
\begin{equation}
    |\{D \in \mathcal{D}(B_{r,s}(d)) : \pn(D)=k\}|
    =
    \left(\binom{r}{\ell}\binom{s}{\ell}\ell!\right)^2 (r-\ell)!(s-\ell)!.
\end{equation}
We confirm this matches the sum of squares of dimensions of $n-2\ell$ box irreps: 
\begin{align}
    \sum_{\rho \in \wh{A} : |\rho| = n - 2\ell} d_\rho^2
    = \sum_{\substack{|\lambda| = r-\ell \\ |\mu| = s-\ell}}
    \left(\ell!\binom{r}{\ell}\binom{s}{\ell}f^\lambda f^\mu\right)^2 \\
    = \left(\ell!\binom{r}{\ell}\binom{s}{\ell}\right)^2
    \left(\sum_{|\lambda| = r-\ell}(f^\lambda)^2\right)
    \left(\sum_{|\mu| = s-\ell}(f^\mu)^2\right)
    = \left(\ell!\binom{r}{\ell}\binom{s}{\ell}\right)^2 (r-\ell)!(s-\ell)!
\end{align}
\end{proof}

\begin{theorem}
    Assume $A \in \{P_n(d), P_{n - \frac 12}(d),  B_n(d), B_{r,s}(d)\}$, and $\lambda \in \wh{A}$. For any diagram $D \in \mathcal{D}(A)$ with $\pn(D) < |\lambda|$, $\lambda(D) = 0$.
\end{theorem}
\begin{proof}
   For $P_n(d)$ and $P_{n - \frac12}(d)$, the proof is implicit in Proposition~2.43 of~\cite{halverson2004partitionalgebras}. To generalize it, we unroll the proof for $P_n(d)$, and give the changes required to derive the result for the other algebras. 

For any algebra $A$, $x \in A$ is an \textit{idempotent} if $x^2 = x$.\ A \textit{minimal idempotent} is an idempotent of $A$ satisfying $xAx = \mathbb{C}x$, i.e.\ $xax$ is a scalar multiple of $x$ for all $a \in A$. We say that two minimal idempotents are isomorphic if $Ax_1 \cong Ax_2$. For any finite-dimensional semisimple algebra $A$, equivalence classes of minimal idempotents correspond to the equivalence classes in $\wh{A}$: for any irrep $\lambda$, $V^\lambda$ is isomorphic to $Ax_{\lambda}$ for some minimal idempotent $x_\lambda$, and vice versa (see the discussion following \cite{halverson2004partitionalgebras}, Equation~2.31).

Next, let $J_r \subseteq P_n(d)$ be the subalgebra spanned by all partition diagrams with propagating number less than $r$. By~\cref{lem:propgation_number_sub_mult}, $J_r$ is a two-sided ideal of $P_n(d)$, so the quotient algebra $P_n(d)/J_r$ is well-defined for any $n \ge r$. We write $\overline{a} \coloneqq a + J_r$ for the image of $a$ in the quotient. When $n=r$, Proposition~2.11 of~\cite{halverson2004partitionalgebras} gives $P_r(d)/J_r \cong \mathbb{C}[S_r]$. Next, define the inclusion
\begin{equation}
i: P_r(d) \hookrightarrow P_n(d),
\qquad
i(a)=a \cdot p_{r+1}p_{r+2}\cdots p_n.
\end{equation}
and let $y_\lambda = i(x_\lambda)$.\ Note that for every diagram appearing in the support of \(y_\lambda a y_\lambda\), all vertices to the right of the \(r\)th column must lie in singleton blocks. Therefore, $y_\lambda a y_\lambda = i(x_\lambda b x_\lambda)$ for some $b \in P_r(d)$, and so $\overline{y_\lambda}$ is a minimal idempotent of $P_n(d)/J_r$:
\begin{equation}
\overline{y_\lambda a y_\lambda} = \overline{i(x_\lambda b x_\lambda)} = \overline{i(c x_\lambda)} =  \overline{cy_\lambda}
\end{equation}
The second equality holds because $\overline{x_\lambda}$ is a minimal idempotent of $P_r(d)/J_r$, and the third by the definition of $y_\lambda$. Since $\overline{y_\lambda}$ is a minimal idempotent, 
\begin{equation}
(P_n(d)/J_r)\overline{y_\lambda} \cong V^\lambda
\end{equation}
where $\lambda$ is an $r$-box irrep of $\mathbb{C}[S_r]$. Using $\Phi$ to denote this isomorphism, it follows that
\begin{equation}
    (\Phi \circ L_D)(\overline{ay_\lambda}) = (\lambda(D) \circ \Phi)(\overline{ay_\lambda})
\end{equation}
where $L_D(\overline{ay_\lambda}) = \overline{Day_\lambda}$ is left multiplication in the quotient algebra. For any $D \in J_r, \overline{Day_\lambda} = \overline{0}$, so $L_D = 0$. Since $\Phi$ is an isomorphism, $\lambda(D) = 0$ as well. Hence, if $|\lambda|=r$, then $\lambda(D)=0$ whenever $\pn(D)<r$.
    
To generalize to $P_{n - \frac12}(d)$, we make the following changes. First, we restrict $J_{r-1}$ to be all half partition diagrams with propagating number less than $r-1$, so that $P_{r - \frac12}(d)/J_{r-1} \cong \mathbb{C}[S_{r-1}]$ (Proposition 2.11,~\cite{halverson2004partitionalgebras}). Next, we modify the domain and codomain of $i$, and relabel the $r$th and $n$th columns when performing the embedding:
\begin{equation}
    i: P_{r - \frac12}(d) \hookrightarrow P_{n - \frac12}(d),\;\; i(a) =  (r \; n) \cdot  a  \cdot p_{r+1}p_{r+2} \dots p_{n} \cdot (r \; n)
\end{equation}
which is again well defined for every $n \ge r$. Note that irreps of $P_{n - \frac12}(d)$ have at most $n - 1$ boxes. After these modifications, the rest of the proof goes through analogously to the partition algebra. 

For the Brauer algebra $B_n(d)$, we restrict $J_r$ to Brauer diagrams with propagating number less than $r$. Once again, we have $B_r(d)/J_r \cong \mathbb{C}[S_r]$ (Equation 3.9, \cite{MartinDeVisscher2009}). The correct inclusion map becomes 
\begin{equation}
    i: B_r(d) \hookrightarrow B_n(d), \;\; i(a) = ae_{r+1}e_{r+3}\dots e_{n-1}
\end{equation}
where we assume without loss of generality that $n -r$ is even.\footnote{All irreps of the Brauer algebra have $n - 2k$ boxes, so it suffices to prove the statement when $n - r$ is even.} After these modifications, we again follow the proof for the partition algebra. 

Finally, for the walled Brauer algebra $B_{s,t}(d)$, we define $J_{l, r}$ to be the ideal of walled Brauer diagrams with less than $l$ propagating blocks on the left side of the wall, and less than $r$ propagating blocks on the right side of the wall.\ Then, $B_{s,t}(d)/J_{s, t} \cong \mathbb{C}[S_s \times S_t]$ (Proposition 2.3, \cite{cox2007blockswalledbraueralgebra}), which has irreps indexed by pairs of Young diagrams $(\lambda, \mu)$ with $|\lambda| = s$, $|\mu| = t$ since the Fourier transform commutes with a direct product. The inclusion map is modified to 
\begin{equation}
    i: B_{s, t}(d) \hookrightarrow B_{s + k, t + k}(d), \;\; i(a) = af_{s+1}f_{s+2}\dots f_{s + k}
\end{equation}
after which we proceed as above.
\end{proof} 
\section{Proof of Theorem~\ref{thm:fourier_concentration_on_propagating_number}}
\label{sec:proof_of_invariant}
For convenience, we restate~\cref{thm:fourier_concentration_on_propagating_number} below: 
\begin{theorem}
    Let $A \in \{P_n(d), P_{n - \frac12}(d), B_n(d), B_{r,s}(d)\}$ be $\delta$-nice. Let $a = d^{\frac{n - \cc(D)}{2}}D \in \mathcal{B}(A)$, and define $\ket{\wt{a}} = \wt{\ft_A}\ket{a}$.
    
    Define $\Pi_r$ to be the projection onto the subspace $S_r$, spanned by states of the form 
\begin{equation}
    \{\ket{\lambda, i, j}: |\lambda| = r\}
\end{equation}
If $D$ has propagating number $r$, then 
\begin{equation}
    \abs{\braket{\wt{a}|\Pi_r |\wt{a}}_0 - \braket{\wt{a}, \wt{a}}_0} = O(\poly(|A|) \cdot \delta)
\end{equation}
\end{theorem}
\begin{proof}
    First, we show that $\{\wt{\ft_A}\ket{a}\}_{a \in \mathcal{B}(A)}$ is approximately orthonormal. Since $A$ is $\delta$-nice (\cref{cor:d_nice_algebras}), the Gram matrix $G_{\ft_A} = \ft_A^\dagger \ft_A$ satisfies $||G_{\ft_A} - I||_{\infty} \le \delta$. Defining $K \coloneqq \wt{\ft_A} - \ft_A$, 
    \begin{align}
        G_{\wt{\ft_A}} = \wt{\ft_A}^\dagger \wt{\ft_A} = K^\dagger\ft_A  + \ft_A^\dagger K + G_{\ft_A} + K^\dagger K \\ 
        \implies ||G_{\wt{\ft_A}} - G_{\ft_A}||_{\infty} \le 2\delta ||\ft_A||\cdot |A|^{3/2} + \delta^2 |A|^3 \\
        \implies ||G_{\wt{\ft_A}} - G_{\ft_A}||_{\infty} \le 2\delta \sqrt{1 + \delta} \cdot |A|^{3/2} + \delta^2 |A|^3\\
        \implies ||G_{\wt{\ft_A}} - I||_\infty \le \delta + 2\delta \sqrt{1 + \delta} \cdot |A|^{3/2} + \delta^2 |A|^3 \label{eq:new_gram_matrix}
    \end{align}
    with the first implication by subadditivity and sub multiplicativity of the operator norm, the second and third using the $\delta$-niceness of $A$ (\cref{thm:approx_qft},~\cref{lem:operator_norm_of_inverse}). For conciseness, let $\eta \coloneqq \delta + 2\delta \sqrt{1 + \delta} \cdot |A|^{3/2} + \delta^2 |A|^3$.
    
  Now, we prove the theorem using strong induction on $k$. We will simultaneously prove that
  \begin{equation}
      \mathcal{B}_k = \{\ket{\wt{a}}: {a = d^{({n - \cc(D)})/{2}}D, \;\; \pn(D) \le k}\}
  \end{equation}
  is a basis of $\bigoplus_{i=0}^k S_i$. For the base case $k=0$, \cref{thm:zero_irrep_matrix} implies that $\lambda(D)=0$ whenever $|\lambda|>0$, and therefore $\ket{\wt{a}}=\Pi_0 \ket{\wt{a}}$. Moreover, $\mathcal{B}_0$ forms a basis for $S_0$. This is because the Gram matrix of $S_0$ is a principal submatrix of $G_{\wt{\ft_A}}$, so it has operator-norm distance at most $\eta$ of the identity. $\eta < 1$ by assumption, and so this submatrix is invertible. Hence, $\mathcal{B}_0$ is linearly independent, and therefore a basis by \cref{thm:prop_number_k_equals_irrep_dimension_at_level_k}.

    For the inductive step, assume that for every $k < r$, $\mathcal{B}_k$ is a basis of $\bigoplus_{i=0}^k S_i$.\ Fix some $b = d^{({n - \cc(E)})/{2}}E$ with $\pn(E) =r $.\ By \cref{thm:zero_irrep_matrix}, $\ket{\wt{b}}$ is supported on the subspace $\bigoplus_{i=0}^r S_i$, and $\ket{\wt{b_{\le r - 1}}} \coloneqq \sum_{i=0}^{r-1} \Pi_i\ket{\wt{b}}$ is supported on the subspace $\bigoplus_{i=0}^{r-1} S_i$. Also by the induction hypothesis, 
  \begin{align}
      \ket{\wt{b_{\le r - 1}}} = \sum_{\wt{a} \in \mathcal{B}_{r-1}} c_a \ket{\wt{a}} \\ 
      \implies \abs{\braket{\wt{b}, \wt{b_{\le r - 1}}}_0} \le \sum_{\wt{a} \in \mathcal{B}_{r-1}} \abs{c_a} \cdot  \abs{\braket{\wt{a}, \wt{b}}_0}
  \end{align}
  Expanding $c_a$ using the definition of a Gram matrix, together with \cref{eq:new_gram_matrix} and \cref{lem:operator_norm_of_inverse}, we conclude that
  $\abs{c_a} \le (1 + \eta)/(1 - \eta)$ and $\abs{\braket{\wt{a}, \wt{b}}}_0 \le \eta$. Substituting this into the bound,
\begin{equation}
    \abs{\braket{\wt{b}, \wt{b_{\le r - 1}}}_0} \le \sum_{\wt{a} \in \mathcal{B}_{r-1}} \frac{(1 + \eta)}{1 - \eta} \eta \le |A| \cdot \frac{(1 + \eta)}{1 - \eta} \eta = O(\poly(|A|) \cdot d^{-1/2})
\end{equation}
\cref{thm:zero_irrep_matrix} implies that $\ket{\wt{b}} = \ket{\wt{b_{\le r - 1}}} + \ket{\wt{b_{r}}}$, from which the result follows. Finally, $\mathcal{B}_r$ is a basis for $\bigoplus_{i=0}^r S_i$ via the same argument used to show that $\mathcal{B}_0$ formed a basis for $S_0$. 
\end{proof}

\section{Orthogonal Forms for Diagram Algebras}
\label{sec:appendix_matrix_algebras}
 Here, we prove the statements from~\cref{sec:main_body_orthogonal_form}, as well as an additional structural result about irrep matrices of the partition algebra. Throughout this section, we write $P \sim_i Q$ for two Bratteli paths $P$ and $Q$ if $P(j) = Q(j)$ whenever $j \ne i$. 

\ifanon
\else
 To compute orthogonal forms for each of the diagram algebras below, we have written software to calculate and display the irrep matrices and Fourier basis states in the orthogonal form. 
This software is freely available at \url{https://github.com/Ben-Foxman/semisimple_algebra_visualization}.
\fi
\subsection{Matrix Elements For The Brauer Algebra}
\begin{lemma}[\cite{Nazarov1996Brauer}]
    \label{lem:brauer_subalgebra_matrix_elements}
    \normalfont
    We summarize the results of Section 3 of~\cite{Nazarov1996Brauer}, which computes the matrix elements in the orthogonal form for $B_n(d)$. To begin, for any Bratteli path $P$, define $x_i(P)$ as follows:
    \begin{equation}
        x_i(P) = 
        \begin{cases}
            (d - 1)/2 + \cont(a), & P(i) = P(i - 1) \cup \{a\} \\
            -(d - 1)/2 - \cont(a), & P(i) = P(i - 1) \setminus \{a\} 
        \end{cases}
    \end{equation}
     Let $\Delta_i \coloneqq x_{i+1}(P) - x_i(P)$.

    \paragraph{The Swap Generators.} To compute $\lambda(s_i)$, we consider two cases : $P(i - 1) \ne P(i + 1)$ and $P(i - 1) = P(i + 1)$. In the former case, there is at most one other Bratteli path $Q$ with $P \sim_i Q$. If such a $Q$ exists, then 
    \begin{equation}
        \label{eq:brauer_symmetric_elements_one}
        \lambda(s_i)\ket{P} = \frac{1}{\Delta_i}\ket{P} + \sqrt{1 - \frac{1}{\Delta_i^2}}\ket{Q},\;\; \lambda(s_i)\ket{Q} = \sqrt{1 - \frac{1}{\Delta_i^2}}\ket{P} - \frac{1}{\Delta_i}\ket{Q}
    \end{equation}
    otherwise, 
    \begin{equation}
    \label{eq:brauer_symmetric_elements_two}
        \lambda(s_i)\ket{P} = \frac{1}{\Delta_i}\ket{P}
    \end{equation}
    In the case where $P(i - 1) = P(i + 1)$, we can write the matrix elements of $\lambda(s_i)$ in terms of those for $\lambda(e_i)$:
    \begin{equation}
        \lambda(s_i)_{PQ} = \frac{\lambda(e_i)_{PQ} - \delta_{PQ}}{x_i(P) + x_i(Q)}
    \end{equation}

    \paragraph{The Contraction Generators.} Now, we give the matrix elements for $\lambda(e_i)$. If $P(i-1) \ne P(i + 1)$, then 
    \begin{equation}
     \label{eq:brauer_matrix_form_eq_one}
        \lambda(e_i)\ket{P} = 0
    \end{equation}
    If $P(i - 1) = P(i + 1)$, then 
    \begin{equation}
        \label{eq:brauer_matrix_form_eq_two}
        \lambda(e_i)_{PP} = 
        \begin{cases}
            (2\cdot x_i(P) + 1)\prod_{c_i \ne x_i(P)} \frac{x_i(P) + c_i}{x_i(P) -c_i}, & x_i(P) \ne -1/2 \\
            -\prod_{c_i \ne x_i(P)} \frac{x_i(P) + c_i}{x_i(P) - c_i}, & x_i(P)  = -1/2
        \end{cases}
    \end{equation}
    where $c_i = \pm((d - 1)/2 + \cont(a_i))$, and $a_i$ is a box that could be added or removed from $P(i - 1)$ to form another valid Young diagram, with the sign chosen respectively. The product in \cref{eq:brauer_matrix_form_eq_two} runs over all such boxes, of which there are at most $O(\sqrt{n})$. The off-diagonal entries satisfy 
    \begin{equation}
        \label{eq:brauer_matrix_form_eq_three}
        \lambda(e_i)_{PQ} = \delta_{P \sim_i Q} \sqrt{\lambda(e_i)_{PP} \cdot \lambda(e_i)_{QQ}}
    \end{equation}
\end{lemma}
As a corollary to~\cref{lem:brauer_subalgebra_matrix_elements}, the matrix entries in the orthogonal form can be computed efficiently: 
\begin{corollary} \label{cor:efficient_computation_brauer_matrix_elements}
    There is a quantum algorithm which, for any irrep $\lambda \in \wh{B_n(d)}$, implements the isometry
    \begin{equation}
        \ket{s_i, \lambda, P} \mapsto \ket{s_i, \lambda, P} \otimes \bigotimes_{Q: \lambda(s_i)_{QP} \ne 0} \ket{Q(i), \lambda(s_i)_{QP}}
    \end{equation}
    up to operator norm error $O(\varepsilon \cdot |B_n(d)|)$ using $\wt{O}(n \cdot (n + \log d + \log(1/\varepsilon)))$ gates.
\end{corollary} 
\begin{proof}
    For any irrep $\lambda$, there are at most $O(\sqrt{n})$ Bratteli paths for which $\lambda(s_i)_{QP} \ne 0$, since any $n$-box Young diagram has at most $O(\sqrt{n})$ addable/removable boxes.\ To determine the set of nonzero indices, first check if $P(i - 1) \ne P(i + 1)$, which can be done using $\wt{O}(n)$ gates. If so, there is at most one other valid Bratteli path. Otherwise, there are at at most $O(\sqrt{n})$ Bratteli paths. In the latter case, storing all such paths (or more precisely, all the $i$th level irreps of these paths) requires $\wt{O}(n^{3/2})$ gates. At this point, we have the state
    \begin{equation}
        \ket{s_i, \lambda, P} \otimes \bigotimes_{Q: \lambda(s_i)_{QP} \ne 0} \ket{Q(i)}.
    \end{equation}
    To compute the entries $\lambda(s_i)_{QP}$, we again condition on $P(i -1) = P(i + 1)$: 
    \begin{enumerate}
        \item $P(i -1) \ne P(i + 1)$: First, we first compute $\Delta_i$. This requires $O(n + \log d)$ gates, with $O(n)$ to compute the content and $O(\log d)$ to compute the final fraction. The gate cost to invert $\Delta_i$ (a $\wt{O}(\log d)$-bit number) to additive error $\varepsilon$ is $\wt{O}(\log d + \log (1/\varepsilon))$~\cite{BrentZimmermann2010, HarveyVanDerHoeven2021}, as is computing the subsequent squares and square roots. The overall gate complexity in this case is $\wt{O}(n + \log d + \log(1/\varepsilon))$.
        \item $P(i -1) = P(i + 1)$: To compute $\lambda(s_i)_{QP}$, we first compute $\lambda(e_i)_{QP}$, which requires computing $\lambda(e_i)_{PP}$ and $\lambda(e_i)_{QQ}$. Determining the set of valid $c_i$'s in~\cref{eq:brauer_matrix_form_eq_two} takes $\wt{O}(n)$ time by iterating over cells in $P(i - 1)$. Then, computing each numerator and denominator requires $O(n + \log d)$ gates as in the first case, and multiplying all of the (at most $O(\sqrt{n})$) numerators and denominators together requires $\wt{O}(\sqrt{n}\log d)$ gates. Taking their ratio to error $\varepsilon$ requires another $\wt{O}(\log d + \log(1/\varepsilon))$ gates. Once $\lambda(e_i)_{PP}$ and $\lambda(e_i)_{QQ}$ are computed, $\lambda(e_i)_{QP}$ and then $\lambda(s_i)_{QP}$ can also be computed with $\wt{O} (\log d + (\log(1/\varepsilon)))$ gates. An upper bound on the gate complexity in this case is $\wt{O} (\sqrt{n} \cdot (n + \log d + (\log(1/\varepsilon)))$.
    \end{enumerate}
    Since we repeat for all $Q$, the total gate count for computing all $O(\sqrt{n})$ matrix entries is $\wt{O} (n\cdot (n + \log d + \log(1/\varepsilon)))$.\ Combining this gate count with the count of determining the valid Bratteli paths gives the bound. Note that all intermediate steps can be uncomputed at a cost of a constant factor increase in the gate count. Since the action of this isometry on any basis state is implemented to additive error $\varepsilon$,~\cref{lem:opnorm_bound} implies a total operator norm error of $O(\varepsilon \cdot |B_n(d)|)$. 
\end{proof}

\subsection{Matrix Elements For The Walled Brauer Algebra}
\begin{lemma}[\cite{grinko2023gelfandtsetlinbasispartiallytransposed}] \normalfont The matrix elements of the orthogonal form for the walled Brauer algebra are similar to those of the Brauer algebra, with a few definitions modified. These are computed by Grinko, Burchardt, and Ozols in Section 3 of~\cite{grinko2023gelfandtsetlinbasispartiallytransposed}. Although~\cite{grinko2023gelfandtsetlinbasispartiallytransposed} define the orthogonal form with respect to the chain 
\begin{equation}
    B_{0, 0}(d) \subseteq B_{1, 0}(d) \subseteq \dots B_{r,0}(d) \subseteq B_{r,1}(d) \subseteq \dots \subseteq B_{r,s}(d)
\end{equation}
as opposed to our oscillating chain in~\cref{eq:walleD_brauer_chain}, the derivation of the orthogonal form in~\cite{grinko2023gelfandtsetlinbasispartiallytransposed} \textit{depends only on branching rules at each step}, with the difference in the final matrix elements reflected in the set of valid Bratteli paths (which will change the set $\{Q: P \sim_i Q\}$ appearing in some of the formulas below). Since our branching rules are the same, we can still apply the formulas in~\cite{grinko2023gelfandtsetlinbasispartiallytransposed}. 

For a Bratteli path $P$, define $x_i(P)$ as follows: 
\begin{equation}
    x_i(P) = 
    \begin{cases}
        \cont(a) & \text{if $i \le r$, $P(i - 1) = (\lambda, \emptyset)$, $P(i) = (\lambda \cup \{a\}, \emptyset) $} \\ 
        d + \cont(a) & \text{if $i > r$, $P(i - 1) = (\lambda, \mu)$, $P(i) = (\lambda, \mu \cup \{a\}) $} \\
        -\cont(a) & \text{if $i > r$, $P(i - 1) = (\lambda, \mu)$, $P(i) = (\lambda \setminus \{a\}, \mu) $} \\ 
    \end{cases}
\end{equation}
and define $\Delta_i \coloneqq x_{i+1}(P) - x_i(P)$ as before.

\paragraph{The Swap Generators.} To compute $\rho(s_i) = (\lambda, \mu)(s_i)$, we use a direct analogy of \cref{eq:brauer_symmetric_elements_one} and \cref{eq:brauer_symmetric_elements_two}, noting that $P(i-1)$ must be distinct from $P(i + 1)$: 
\begin{equation}
        \label{eq:walleD_brauer_symmetric_elements_one}
        (\lambda, \mu)(s_i)\ket{P} = \frac{1}{\Delta_i}\ket{P} + \sqrt{1 - \frac{1}{\Delta_i^2}}\ket{Q},\;\; \lambda(s_i)\ket{Q} = \sqrt{1 - \frac{1}{\Delta_i^2}}\ket{P} - \frac{1}{\Delta_i}\ket{Q}
    \end{equation}
    when there exists another Bratteli path $Q$ with $P \sim_i Q$. Otherwise, 
    \begin{equation}
    \label{eq:walleD_brauer_symmetric_elements_two}
        (\lambda, \mu)(s_i)\ket{P} = \frac{1}{\Delta_i}\ket{P}
    \end{equation}
    
\paragraph{The Contraction Generator.}
For $\rho(e_r) = (\lambda, \mu)(e_r)$, $P(r - 1) \ne P(r + 1)$ again implies that
\begin{equation}
    \label{eq:walleD_brauer_contraction}
    (\lambda, \mu)(e_r)\ket{P} = 0
\end{equation}
For $P(r - 1) = P(r + 1)$, we use an analogy of \cref{eq:brauer_matrix_form_eq_two}. Define $c_r(Q) \coloneqq \cont(a)$ if $(\lambda, \mu) = Q(r - 1) = Q(r + 1)$ and $Q(r) = Q(r - 1) \cup \{a\}$, and $c_r(Q) \coloneqq 0$ otherwise. Then, 
\begin{equation}
    \label{eq:walleD_brauer_contraction_off_diagonals}
(\lambda, \mu)(e_r)_{PQ} = \delta_{P \sim_r Q}K(P)K(Q), \;\; K(Q) = \sqrt{(d + c_r(Q)) \frac{\prod_{a \in R(P(r - 1))} (c_r(Q) - \cont(a))}{\prod_{a \in A(P(r - 1))} (c_r(Q) - \cont(a))}}
\end{equation}
where $R(\cdot)$ denotes the set of all boxes that can be removed from the \textit{left} Young diagram to form another valid Young diagram (the product is one if the left irrep is $\emptyset$), and $A(\cdot)$ denotes the set of all boxes that can be added to the \textit{left} Young diagram to form another valid Young diagram. 
\end{lemma}
We also give an analog of~\cref{cor:efficient_computation_brauer_matrix_elements} for the walled Brauer algebra: 
\begin{corollary}
    \label{cor:efficient_computation_walleD_brauer_matrix_elements}
     There is a quantum algorithm which, for any irrep $\rho = (\lambda, \mu) \in \wh{B_{r,s}(d)}$, implements the isometry 
    \begin{equation}
        \ket{s_i, \rho, P} \mapsto \ket{s_i, \rho, P} \otimes \bigotimes_{Q: \rho(s_i)_{QP} \ne 0} \ket{Q(i), \rho(s_i)_{QP}}
    \end{equation}
    up to operator norm error $O(\varepsilon \cdot |B_{r,s}(d)|)$ using $\wt{O}(n \cdot (n + \log d + \log(1/\varepsilon)))$ gates.
\end{corollary}
\begin{proof}
    In the walled Brauer algebra, there are again at most $O(\sqrt{n})$ Bratteli paths $Q$ for which $\rho(s_i)_{QP} \ne 0$. Following the argument for the Brauer algebra, we can prepare
    \begin{equation}
        \ket{s_i, \rho, P} \otimes \bigotimes_{Q: \rho(s_i)_{QP} \ne 0} \ket{Q(i)}
    \end{equation}
    using $\wt{O}(n)$ gates, and then compute each of the nonzero matrix entries using $\wt{O}(n + \log d + \log(1/\varepsilon))$ gates. 
\end{proof}
\subsection{Matrix Elements For The Partition Algebra}
\begin{lemma}[\cite{Enyang_2013}] \normalfont
    First, we summarize Theorems 5.1 to 5.4 in~\cite{Enyang_2013}, using our notation for the swap, point, and bridge diagrams. These theorems give a \textit{seminormal} form for the partition algebra, which can be made into an orthogonal form by applying the proper normalization to each matrix element. 
    
    First, we define some variables. For a Young diagram $\lambda$, we define $R(\lambda)$ to be the set of removable boxes, and $A(\lambda)$ to be the set of addable boxes. For a box $b \in R(\lambda)$, define $R(\lambda)^{< b}$ to be the set of all removable boxes in rows above $b$, and define $A(\lambda)^{< b}$ similarly. Now, for $\mu = \lambda \cup \{a\}$, define 
    \begin{equation}
        \Psi_{\lambda \rightarrow \mu} \coloneqq \frac{\prod_{b \in R(\lambda)^{<a}} (\cont(a) - \cont(b))}{\prod_{b \in A(\lambda)^{<a}} (\cont(a) - \cont(b))}
    \end{equation}
    To convert the seminormal form into the orthogonal form, we will need a formula for $\braket{P, P}_{\text{path}}$ as defined in~\cref{sec:main_body_orthogonal_form}. This is given in Equation 4.2 and Proposition 4.12 of~\cite{Enyang_2013}, which we now summarize. For any Bratteli path $P$ of length $k+1$,
\begin{equation}
    \label{eq:orthogonal_norm_partition_algebra}
    \braket{P,P}_{\text{path}}
    \;=\;
    \prod_{r=0}^{k-1}
    \frac{\braket{P_{\le r+1},P_{\le r+1}}_{\text{path}}}{\braket{P_{\le r},P_{\le r}}_{\text{path}}},
\end{equation}
where $P_{\le r}$ denotes the Bratteli path truncated to level $r$. Proposition 4.12 of~\cite{Enyang_2013} gives formulas for the factors:
\begin{equation}
    \frac{\braket{P_{\le r+1},P_{\le r+1}}_{\text{path}}}{\braket{P_{\le r},P_{\le r}}_{\text{path}}}
    \;=\;
    \begin{cases}
        1, & P(r+1)=P(r)\\[6pt]
        \Psi_{P(r)\rightarrow P(r+1)}, & P(r+1)=P(r)\cup\{a\}\\[6pt]
        \lambda(b_i)_{PP}\,\Psi_{P(r+1)\rightarrow P(r)}, & r = 2i, \; P(r + 1)=P(r)\setminus\{a\}\\[6pt]
    \lambda(p_i)_{PP}\,\Psi_{P(r+1)\rightarrow P(r)}, & r = 2i - 1,\;  P(r + 1)=P(r)\setminus\{a\}
    \end{cases}
\end{equation}
where in the removal case one has \(P(r)=P(r+1)\cup\{a\}\), so \(a\) is the box removed in that step; the extra factor is \(\lambda(b_i)_{PP}\) when \(r=2i\), and \(\lambda(p_i)_{PP}\) when \(r=2i-1\).

    \paragraph{The Bridge Generators.} We give the matrix elements for $\lambda(b_i)$ in the seminormal basis. 
    There are three cases for the diagonal elements: 
    \begin{equation}
        \label{eq:bridge_gen_formula_one}
        \lambda(b_i)_{PP} = 
        \begin{cases}
            \frac{\prod_{b \in R(P(2i))} d - \cont(b) - |P(2i)|}{\prod_{b \in A(P(2i))} d - \cont(b) - |P(2i)|}, & P(2i - 1) = P(2i) = P(2i + 1) \\
             \frac{d - \cont(a) - |P(2i)| - 1}{d - \cont(a) - |P(2i)|}\cdot \frac{\prod_{b \in R(P(2i))} \cont(a) - \cont(b)}{\prod_{b \ne a \in A(P(2i))} \cont(a) - \cont(b)}, & P(2i - 1) = P(2i) \cup \{a\} = P(2i + 1) \\
             0, & P(2i - 1) \ne P(2i + 1)
        \end{cases}
    \end{equation}
    There are also three cases for the off diagonal elements: 
    \begin{equation}
        \label{eq:bridge_gen_formula_two}
        \lambda(b_i)_{QP} = \delta_{P \sim_{2i} Q} \cdot 
        \begin{cases}
            \Psi_{P(2i) \rightarrow Q(2i)}^{-1}, & Q(2i) = P(2i) \cup \{a\} \\
            \lambda(b_i)_{PP}\lambda(b_i)_{QQ}\Psi_{Q(2i) \rightarrow P(2i)}, & Q(2i) = P(2i) \setminus \{a\} \\
            \lambda(b_i)_{PP}\frac{\Psi_{Q(2i + 1) \rightarrow Q(2i)}}{\Psi_{P(2i + 1) \rightarrow P(2i)}}, & \text{ otherwise}
        \end{cases}
    \end{equation}

    \paragraph{The Point Generators.} Next, we give the matrix elements for $\lambda(p_i)$ in the seminormal basis. Like the bridge generator, there are three cases for the diagonal elements: 
    \begin{equation}
     \label{eq:partition_point_generator_one}
        \lambda(p_i)_{PP} = 
        \begin{cases}
             \frac{\prod_{b \in A(P(2i - 1))} d - \cont(b) - |P(2i - 1)|}{\prod_{b \in R(P(2i - 1))} d - \cont(b) - |P(2i - 1)|}, & P(2i - 2) = P(2i - 1) = P(2i) \\
             -\frac{d - \cont(a) - |P(2i - 1)| + 1}{d - \cont(a) - |P(2i - 1)|}\cdot \frac{\prod_{b \in A(P(2i - 1))} \cont(b) - \cont(a)}{\prod_{b \ne a \in R(P(2i - 1))} \cont(b) - \cont(a)}, & P(2i - 2) = P(2i - 1) \setminus \{a\} = P(2i) \\
             0, & P(2i -2) \ne P(2i)
        \end{cases}
    \end{equation}
    and three cases for the off-diagonal elements: 
    \begin{equation}
     \label{eq:partition_point_generator_two}
        \lambda(p_i)_{QP} = \delta_{P \sim_{2i - 1} Q} \cdot 
        \begin{cases}
            \Psi_{Q(2i - 1) \rightarrow P(2i - 1)}, & Q(2i - 1) = P(2i - 1) \setminus \{a\} \\
            \frac{\lambda(p_i)_{PP}\lambda(p_i)_{QQ}}{\Psi_{P(2i - 1) \rightarrow Q(2i - 1)}}, & Q(2i - 1) = P(2i - 1) \cup \{a\} \\
            \lambda(p_i)_{QQ}\frac{\Psi_{P(2i) \rightarrow P(2i - 1 )}}{\Psi_{Q(2i) \rightarrow Q(2i - 1)}}, & \text{ otherwise}
        \end{cases}
    \end{equation}

    \paragraph{The Swap Generators.} Instead of matrices for the swap generators directly, Enyang instead gives seminormal forms for the what we will call \textit{sigma}-elements, denoted $\sigma_i$. Defined in~\cite{Enyang_2012}, these elements satisfy 
    \begin{equation}
       s_i = \sigma_{2i}\sigma_{2i + 1} 
    \end{equation}
  We define $c_P(i)$ by parity of the step as follows:
    \begin{equation}
        c_P(i)=
        \begin{cases}
            d-|P(i)|, & \text{if $i$ is even and } P(i)=P(i-1),\\
            \cont(a), & \text{if $i$ is even and } P(i)=P(i-1)\cup\{a\},\\
            |P(i)|, & \text{if $i$ is odd and } P(i)=P(i-1),\\
            d-\cont(a), & \text{if $i$ is odd and } P(i)=P(i-1)\setminus\{a\}.
        \end{cases}
    \end{equation}
    Also, we write $P\sim_{i, j} Q$ if $P$ and $Q$ agree on all levels except possibly the $i$th and/or $j$th. 
    Now, we give the matrix elements of $\lambda(\sigma_{2i})$. There are five cases to consider. In each of them, the off diagonal entry $(P, Q)$ is zero unless $P \sim_{2i - 1, 2i} Q$.
    
\label{sec:cases_partition}
\noindent\textbf{Case 1:} $P(2i-1)=P(2i+1)$ and $P(2i-2)=P(2i)$. Then
\begin{equation}
    \lambda(\sigma_{2i})_{PP}
    =
    \frac{c_P(2i)}{\lambda(p_i)_{PP}},
\end{equation}
and for $Q\neq P$,
\begin{equation}
    \lambda(\sigma_{2i})_{QP}
    =
    \frac{d-c_Q(2i)-c_P(2i-1)-\lambda(p_i)_{PP}}
    {c_Q(2i+1)-c_P(2i-1)}
    \,\lambda(b_i)_{QP}.
\end{equation}

\noindent\textbf{Case 2:} $P(2i-1)\neq P(2i+1)$ and $P(2i)=P(2i-2)$.
Let $V$ be the unique path with $V\sim_{2i-1}P$ and $V(2i-1)=P(2i+1)$. Then
\begin{equation}
    \lambda(\sigma_{2i})_{VP}
    =
    c_P(2i)\,\lambda(b_i)_{VP},
\end{equation}
and
\begin{equation}
    \lambda(\sigma_{2i})_{QP}
    =
    \frac{\delta_{QP}-\lambda(p_i)_{VP}\lambda(b_i)_{QV}}
    {c_Q(2i+1)-c_P(2i-1)}.
\end{equation}

\noindent\textbf{Case 3:} $P(2i+1)=P(2i-1)$ and $P(2i)\neq P(2i-2)$. Then if $Q(2i)\neq P(2i-2)$,
\begin{equation}
    \lambda(\sigma_{2i})_{QP}
    =
    \frac{\delta_{QP}+\bigl(d-c_P(2i-1)-c_Q(2i)\bigr)\lambda(b_i)_{QP}}
    {c_Q(2i+1)-c_P(2i-1)}.
\end{equation}
If instead $Q(2i)=P(2i-2)$, then
\begin{equation}
    \lambda(\sigma_{2i})_{QP}
    =
    \lambda(\sigma_{2i})_{PQ}
    \frac{\braket{P,P}_{\mathrm{path}}}{\braket{Q,Q}_{\mathrm{path}}}.
\end{equation}

\noindent\textbf{Case 4:} $P(2i+1)\neq P(2i-1)$ and $P(2i)\neq P(2i-2)$, and $P\sigma_{2i}$ does not exist. Then
\begin{equation}
    \lambda(\sigma_{2i})_{QP}
    =
    \frac{\delta_{QP}}{c_P(2i+1)-c_P(2i-1)}.
\end{equation}

\noindent\textbf{Case 5:} $P(2i+1)\neq P(2i-1)$ and $P(2i)\neq P(2i-2)$, and there exists a distinct $Q$ satisfying $P\sim_{2i-1,2i}Q$.
\begin{equation}
\lambda(\sigma_{2i})_{QP}=
\begin{cases}
\dfrac{1}{c_P(2i+1)-c_P(2i-1)},&Q=P,\\[0.8em]
1-\dfrac{1}{(c_P(2i+1)-c_P(2i-1))^2},&Q\succ P,\\[0.8em]
1,&P\succ Q,\\
\end{cases}
\end{equation}
and $0$ otherwise. 
Here, $P \succ Q $ means that $P$ comes after $Q$ in lexicographic order, where $P(i) > Q(i)$ if $|P(i)| > |Q(i)|$.\footnote{More accurately, irreps at the same level of the path are compared by \textit{partition dominance}, but for the $P$ and $Q$ appearing in this case these two comparators are equivalent, so we use the simpler one.}

Now, we move on to $\lambda(\sigma_{2i + 1})$. There are again five cases. Analogous to $\lambda(\sigma_{2i})$, the off diagonal entry $(Q,P)$ is zero unless $P \sim_{2i, 2i + 1} Q$.

\noindent\textbf{Case 1:} $P(2i-1)=P(2i)=P(2i+1)=P(2i+2)$.
\begin{equation}
\lambda(\sigma_{2i+1})_{PP}=\frac{c_P(2i)}{\lambda(b_i)_{PP}},\qquad
\lambda(\sigma_{2i+1})_{QP}=-\frac{\lambda(b_i)_{PP}\,\lambda(b_i)_{QP}}{c_Q(2i+2)-c_P(2i)},
\end{equation}
for $Q\neq P$, and $0$ otherwise.

\noindent\textbf{Case 2:} $P(2i+1)=P(2i-1)$ and $P(2i+2)\neq P(2i)$.
Let $V$ be the path with $V\sim_{2i+1}P$ and $V(2i)$ the unique predecessor vertex used in Enyang Thm.~5.4(2).
\begin{equation}
\lambda(\sigma_{2i+1})_{VP}=c_P(2i)\,\lambda(p_{i+1})_{VP},
\end{equation}
and
\begin{equation}
\lambda(\sigma_{2i+1})_{QP}=
\frac{\delta_{QP}-\lambda(b_i)_{VP}\lambda(p_{i+1})_{QV}+(c_V(2i)-c_P(2i))\,\lambda(b_i)_{QP}}{c_Q(2i+2)-c_P(2i)}.
\end{equation}

\noindent\textbf{Case 3:} $P(2i+2)=P(2i)$ and $P(2i+1)\neq P(2i-1)$.
If $Q(2i+1)\neq P(2i-1)$, then
\begin{equation}
\lambda(\sigma_{2i+1})_{QP}=\frac{\delta_{QP}}{c_P(2i+2)-c_P(2i)}.
\end{equation}
If $Q(2i+1)=P(2i-1)$, then
\begin{equation}
\lambda(\sigma_{2i+1})_{QP}=\lambda(\sigma_{2i+1})_{PQ}\cdot\frac{\braket{P,P}_{\text{path}}}{\braket{Q,Q}_{\text{path}}}.
\end{equation}

\noindent\textbf{Case 4:} $P(2i+2)\neq P(2i)$ and $P(2i+1)\neq P(2i-1)$, and no other path satisfies $P\sim_{2i,2i+1}Q$.
\begin{equation}
\lambda(\sigma_{2i+1})_{QP}=
\begin{cases}
\dfrac{1}{c_P(2i+2)-c_P(2i)},&Q=P,
\end{cases}
\end{equation}
and $0$ otherwise.

\noindent\textbf{Case 5:} $P(2i+2)\neq P(2i)$ and $P(2i+1)\neq P(2i-1)$, and there exists a distinct $Q$ satisfying $P\sim_{2i,2i+1}Q$.
\begin{equation}
\lambda(\sigma_{2i+1})_{QP}=
\begin{cases}
\dfrac{1}{c_P(2i+2)-c_P(2i)},&Q=P,\\[0.8em]
1-\dfrac{1}{(c_P(2i+2)-c_P(2i))^2},&Q\succ P,\\[0.8em]
1,&P\succ Q,
\end{cases}
\end{equation}
and $0$ otherwise.

\paragraph{From the seminormal to orthogonal form.}
To convert the seminormal form into the orthogonal form, we scale every $(P, Q)$ matrix entry by 
\begin{equation}
   \sqrt{\frac{\braket{P, P}_{\text{path}}}{\braket{Q, Q}_{\text{path}}}}
\end{equation}
which captures the effect of scaling each Bratteli path $P$ by the normalization constant $1/\sqrt{\braket{P, P}_{\text{path}}}$, computed in~\cref{eq:orthogonal_norm_partition_algebra}. In the orthogonal form, $\lambda(s_i)$ is a Hermitian and unitary operator, $\lambda(b_i)$ is an orthogonal projector (since $b_i^2 = b_i$), and $\lambda(p_i)$ is a Hermitian operator satisfying $\lambda(p_i)^2 = d\lambda(p_i)$. As an example, we give the explicit orthogonal forms for $P_2(d)$: 

\paragraph{$\lambda=\emptyset$:}
\begin{equation}
\lambda(p_1)=
\begin{pmatrix}
d&0\\
0&0
\end{pmatrix},
\qquad
\lambda(b_1)=
\begin{pmatrix}
\frac1d&\frac{\sqrt{d-1}}{d}\\
\frac{\sqrt{d-1}}{d}&\frac{d-1}{d}
\end{pmatrix},
\qquad
\lambda(p_2)=
\begin{pmatrix}
d&0\\
0&0
\end{pmatrix},
\qquad
\lambda(s_1)=
\begin{pmatrix}
1&0\\
0&1
\end{pmatrix}.
\end{equation}

\paragraph{$\lambda=\Box$:}
\begin{equation}
\lambda(p_1)=
\begin{pmatrix}
d&0&0\\
0&0&0\\
0&0&0
\end{pmatrix},
\qquad
\lambda(b_1)=
\begin{pmatrix}
\frac1d&\frac{\sqrt{d-1}}{d}&0\\
\frac{\sqrt{d-1}}{d}&\frac{d-1}{d}&0\\
0&0&0
\end{pmatrix},
\qquad
\lambda(p_2)=
\begin{pmatrix}
0&0&0\\
0&\frac{d}{d-1}&\frac{d\sqrt{d-2}}{d-1}\\
0&\frac{d\sqrt{d-2}}{d-1}&\frac{d(d-2)}{d-1}
\end{pmatrix}.
\end{equation}
\begin{equation}
\lambda(s_1)=
\begin{pmatrix}
0&\frac{1}{\sqrt{d-1}}&\frac{\sqrt{d-2}}{\sqrt{d-1}}\\
\frac{1}{\sqrt{d-1}}&\frac{d-2}{d-1}&-\frac{\sqrt{d-2}}{d-1}\\
\frac{\sqrt{d-2}}{\sqrt{d-1}}&-\frac{\sqrt{d-2}}{d-1}&\frac{1}{d-1}
\end{pmatrix}.
\end{equation}

\paragraph{$\lambda=\raisebox{0.5ex}{\scalebox{0.3}{$\ydiagram{1,1}$}}$\;:}
\begin{equation}
\lambda(p_1)=0,\qquad \lambda(b_1)=0,\qquad \lambda(p_2)=0,\qquad \lambda(s_1)=-1.
\end{equation}

\paragraph{$\lambda=\Box\!\Box$:}
\begin{equation}
\lambda(p_1)=0,\qquad \lambda(b_1)=0,\qquad \lambda(p_2)=0,\qquad \lambda(s_1)=1.
\end{equation}
\end{lemma}
We will also need the following characterization of the $+1$ eigenspace of $\lambda(b_i)$:
\begin{lemma}
    \label{lem:bridge_equivalence_class_structure}
    Let $[P]$ be the equivalence class of a Bratteli path $P$ under $\sim_{2i}$. 
    \begin{enumerate}
        \item If $P(2i - 1) \ne P(2i + 1)$, then $\lambda(b_i)\ket{Q} = 0$ for all $Q \in [P]$.
        \item If $P(2i - 1) = P(2i + 1)$, then $\lambda(b_i)$ has exactly one $+1$ eigenvector supported on $[P]$, up to scaling.
    \end{enumerate}
\end{lemma}
\begin{proof}
Define
\begin{equation}
    V_{[P]} \coloneqq \text{span}\{\ket{Q} : Q \in [P]\}.
\end{equation}
From~\cref{eq:bridge_gen_formula_two}, every off-diagonal entry $\lambda(b_i)_{QP}$ is multiplied by $\delta_{P \sim_{2i} Q}$. Hence, if $\ket{Q}$ is supported on $[P]$, then so is $\lambda(b_i)\ket{Q}$, and therefore $V_{[P]}$ is invariant under $\lambda(b_i)$.

For the first claim, if $P(2i-1)\ne P(2i+1)$, then for every $Q\in [P]$ we also have $Q(2i-1)\ne Q(2i+1)$, since $Q$ agrees with $P$ at every level except possibly $2i$.\ ~\cref{eq:bridge_gen_formula_one} then gives $\lambda(b_i)_{QQ}=0$ for all $Q\in [P]$, and the off-diagonal formulas also vanish. Thus $\lambda(b_i)|_{V_{[P]}}=0$.

For the second claim, ~\cref{eq:bridge_gen_formula_one} and~\cref{eq:bridge_gen_formula_two} imply that in the seminormal form $\lambda(b_i)_{QP}$ factors as $x_Qy_P$:
\begin{equation}
    x_Q =
    \begin{cases}
        \lambda(b_i)_{QQ}, & Q(2i)=Q(2i + 1), \\[0.5em]
        \Psi_{Q(2i + 1)\rightarrow Q(2i)}^{-1}, & Q(2i)=Q(2i - 1)\cup\{a\},
    \end{cases}
\end{equation}
\begin{equation}
    y_P =
    \begin{cases}
        1, & P(2i)=p(2i + 1), \\[0.5em]
        \lambda(b_i)_{PP}\,\Psi_{P(2i + 1)\rightarrow P(2i)}, &P(2i)=P(2i + 1)\cup\{a\}
    \end{cases}
\end{equation}
Since $\lambda(b_i)|_{V_{[P]}}$ is the outer product of two vectors, it has rank 1. 

Finally, converting the seminormal form to the orthogonal form is a change of basis, which does not affect the rank. In the orthogonal form, $\lambda(b_i)|_{V_{[P]}}$ is an orthogonal projector, and so $\lambda(b_i)|_{V_{[P]}}$ has a one-dimensional $+1$ eigenspace.
\end{proof}

Finally, we give an analog of~\cref{cor:efficient_computation_brauer_matrix_elements} for the partition algebra: 
\begin{corollary}
    \label{cor:efficient_computation_partition_matrix_elements}
    There is a quantum algorithm which, for any irrep $\lambda \in \wh{P_n(d)}$, implements the isometry 
    \begin{equation}
        \label{eq:partition_isometry}
        \ket{s_i, \lambda, P} \mapsto \ket{s_i, \lambda, P} \otimes \bigotimes_{Q: \lambda(s_i)_{QP} \ne 0} \ket{Q(2i - 1), Q(2i), Q(2i + 1),  \lambda(s_i)_{QP}}
    \end{equation}
    up to operator norm error $O(\varepsilon \cdot |P_n(d)|)$ using $\wt{O}(n^{5/2} \cdot (n + \log d + \log(1/\varepsilon)))$ gates.  
\end{corollary}
\begin{proof}
Since 
\begin{equation}
    \label{eq:swap_is_two_sigmas}
        \lambda(s_i)_{QP} = \sum_{\substack{R \\ Q \sim_{2i - 1, 2i} R \\ R \sim_{2i, 2i + 1} P}} \lambda(\sigma_{2i})_{QR}\lambda(\sigma_{2i + 1})_{RP}
    \end{equation}
    $P$ and $Q$ may differ on levels $2i-1$, $2i$, and $2i + 1$ while still satisfying $\lambda(s_i)_{QP} \ne 0$. There are up to $O(n)$ possible paths $Q$ satisfying this condition. To see this, first note that there are $O(\sqrt{n})$ possible choices for $Q(2i)$, since $Q(2i)$ is obtained from $Q(2i-2)$ by adding or removing at most one box. If $Q(2i) \ne Q(2i-2)$, then $Q(2i-1)$ and $Q(2i+1)$ are uniquely determined. On the other hand, if $Q(2i)=Q(2i-2)$, then there are $O(\sqrt{n})$ choices for each of $Q(2i-1)$ and $Q(2i+1)$, contributing up to $O(n)$ valid Bratteli paths. As a consequence, implementing 
   \begin{equation}
        \ket{s_i, \lambda, P} \otimes \bigotimes_{Q: \lambda(s_i)_{QP} \ne 0} \ket{Q(2i - 1),Q(2i), Q(2i + 1)}
    \end{equation}
    requires $\wt{O}(n^{2})$ gates.
    
    Our next task is to compute $\lambda(\sigma_{2i})_{QR}$ for every $R\sim_{2i - 1, 2i} Q$, for each path $Q$. For each $Q$, there are $O(\sqrt{n})$ choices for $R$, so with $\wt{O}(n^{3/2})$ gates we can implement the map 
    \begin{equation}
        \ket{Q(2i - 1),Q(2i)} \mapsto \ket{Q(2i - 1),Q(2i)} \otimes \bigotimes_{R\sim_{2i - 1, 2i } Q} \ket{R(2i), R(2i + 1)}
    \end{equation}
    For each $R$, we determine which of the five cases we are in (\cref{sec:cases_partition}) using $\wt{O}(n)$ gates. After determining which case we are in, computing the appropriate \textit{seminormal} matrix entry $\lambda(\sigma_{2i})_{QR}$ requires the following components: 
    \begin{enumerate}
        \item A constant number of terms of the form $c_Q(j)$, $c_R(j)$, for some index $j$. These can be computed using $\wt{O}(n + \log d)$ gates. 
        \item \label{item:bridge_time_complexity} A constant number of bridge and point matrix entries, either $\lambda(b_j)_{QQ}$, $\lambda(b_j)_{QR}$, $\lambda(p_j)_{QQ}$, or $\lambda(p_j)_{RQ}$. Each of these requires $\wt{O}(n \cdot (n + \log d + \log(1/\varepsilon)))$ gates to compute, with similar analysis as~\cref{cor:efficient_computation_brauer_matrix_elements}. 
        \item A constant number of additions, multiplications, and divisions of terms from the previous two steps, each on registers with at most $\wt{O}(\log d)$ bits. This requires $\wt{O}(\log d + \log(1/\varepsilon))$ additional gates, to compute the final matrix entries to additive error $\varepsilon$. 
    \end{enumerate}
    Hence, in the seminormal basis, computing a single matrix entry $\lambda(\sigma_{2i})_{QR}$ requires $\wt{O}(n \cdot (n + \log d + \log(1/\varepsilon)))$ gates. For each $Q$, there are $O(\sqrt{n})$ possible paths $R$, so computing $\lambda(\sigma_{2i})_{QR}$ for all $R$ requires $\wt{O}(n^{3/2} \cdot (n + \log d + \log(1/\varepsilon)))$ gates. Computing all $\lambda(\sigma_{2i})_{RP}$ can be done in an analogous fashion, which implies that computing a single matrix entry $\lambda(s_i)_{QP}$ in the seminormal form (\cref{eq:swap_is_two_sigmas}) also requires $\wt{O}(n^{3/2} \cdot (n + \log d + \log(1/\varepsilon)))$ gates. To convert this matrix entry into the orthogonal form, we multiply the entry by the normalization factor 
    \begin{equation}
        \sqrt{\frac{\braket{Q, Q}_{\text{path}}}{\braket{P, P}_{\text{path}}}}
    \end{equation}
    Each factor of the telescoping product in~\cref{eq:orthogonal_norm_partition_algebra} requires at most $\wt{O}(n^{1/2} \cdot (n + \log d + \log(1/\varepsilon)))$ gates to compute, dominated by the cost of computing $\lambda(b_i)_{PP}$ or $\lambda(b_i)_{QQ}$. Therefore, $\wt{O}(n^{3/2} \cdot (n + \log d + \log(1/\varepsilon)))$ gates are also required to compute and apply the normalization factor. 
    
    In conclusion, the computing a single matrix entry $\lambda(s_i)_{QP}$ in the orthogonal form requires 
    \begin{equation}
        \wt{O}(n^{3/2} \cdot (n + \log d + \log(1/\varepsilon)))
    \end{equation}
    gates. Since there are $O(n)$ nonzero entries, the overall gate complexity of~\cref{eq:partition_isometry} is 
    \begin{equation}
        \wt{O}(n^{5/2} \cdot (n + \log d + \log(1/\varepsilon)))
    \end{equation}
    The operator norm bound is analogous to~\cref{cor:efficient_computation_brauer_matrix_elements}.
\end{proof}

\section{Multiplicities of the Schur Representation}
Here, we give the formulas needed for the embedding step in \cref{sec:embedding}.
\label{sec:multiplicity_formulas}
\begin{lemma}[{\cite{halverson2004partitionalgebras}}, Proposition 3.24]
    For the Schur representation of the partition algebra $P_n(d)$, 
    \begin{equation}
     \label{eq:partition_multiplicity}
        m_\rho = \frac{f^\rho}{|\rho|!} \cdot \prod_{j=1}^{|\rho|} (d- |\rho| - \rho_j + j)
    \end{equation}
   Also, for $P_{n - \frac12}(d)$:
    \begin{equation}
        m_\rho = \frac{f^\rho}{|\rho|!} \cdot d\prod_{j=1}^{|\rho|} (d -1- |\rho| - \rho_j + j)
    \end{equation}
\end{lemma}
A similar formula is also known for the Brauer algebra $B_n(d)$:
\begin{lemma}[{\cite{Molev_2012}}, Equation 3.15]
    For the Schur representation of the Brauer algebra $B_n(d)$, 
    \begin{equation}
        \label{eq:brauer_multiplicity}
        m_\rho =  \frac{f^\rho}{|\rho|!} \cdot \prod_{(i,j) \in \rho} (d- 1 + b(i, j))
    \end{equation}
    where $b(i,j) = \rho_i + \rho_j - i - j + 1$ if $i \le j$, and $-\rho_i^\prime - \rho_j^\prime + i + j - 1$ if $i > j$. $\rho_i^\prime$ denotes the number of boxes in the $i$th column of $\rho$. 
\end{lemma}
And finally, for irreps of the walled Brauer algebra: 
\begin{lemma}[{\cite{Halverson1996MixedTensor}}, \cite{Koike1989TensorProducts} Theorem 3.10]
    For the Schur representation of the walled Brauer algebra $B_{r,s}(d)$, 
    \begin{equation}
    \label{eq:walled_brauer_schur_multiplicity}
        m_{(\lambda, \mu)} = \prod_{1 \le i < j \le d } \frac{a_i - a_j + j - i}{j - i}
    \end{equation}
    Here, $a_i = \lambda_i$ for $1 \le i \le r$, $a_i = 0$ for $r < i \le d - s$, $a_{d - s + k} = -\mu_{s - k + 1}$ for $1 \le k \le s$. 
\end{lemma}
\begin{lemma}
\label{lem:efficient_to_compute_multiplicities}
For any irrep $\rho \in \wh{A}$, the multiplicity $m_\rho$ in the Schur representation can be computed to additive error $\varepsilon$ with
\begin{equation}
    \begin{cases}
        \wt{O}(n\log d + \log(1/\varepsilon)) \text{ gates}, & A \in \{P_n(d), P_{n - \frac12}(d), B_n(d)\}, \\[4pt]
        \wt{O}(n^2\log d + \log(1/\varepsilon)) \text{ gates}, & A = B_{r,s}(d).
    \end{cases}
\end{equation}
\end{lemma}
\begin{proof}
For $A\in\{P_n(d),P_{n - \frac12}(d),B_n(d)\}$, we first compute 
\begin{equation}
    \frac{f^\rho}{|\rho|!} = \frac{1}{\prod_{(i,j) \in \rho} h(i,j)}
\end{equation}
via $O(n)$ multiplications of $O(\log n)$-bit numbers, followed by taking an inverse. This requires $\wt{O}(n + \log (1/\varepsilon))$ gates~\cite{HarveyVanDerHoeven2021}. Next, we compute each of the $|\rho|$ factors in the second product, which takes $\wt{O}(\log d)$ time, and therefore $\wt{O}(n\log d)$ time to obtain the second product. Therefore, the total gate complexity is upper bounded by $\wt{O}(n\log d + \log(1/\varepsilon))$.

For the walled Brauer algebra,~\cref{eq:walled_brauer_schur_multiplicity} seems to require computing $\Theta(d^2)$ factors. However, when $d \gg n$, most of the $a_i$ terms are equal to zero, which simplifies the calculation considerably. We can rewrite~\cref{eq:walled_brauer_schur_multiplicity} as 
\begin{align}
\label{eq:walled_brauer_multiplicity_simplified_correct}
m_{(\lambda,\mu)}
&=\Bigg(\prod_{1\le i<j\le r}\frac{\lambda_i-\lambda_j+j-i}{j-i}\Bigg)
\Bigg(\prod_{1\le i<j\le s}\frac{\mu_i-\mu_j+j-i}{j-i}\Bigg) \nonumber\\
&\quad\cdot
\Bigg(\prod_{i=1}^{r}\prod_{j=1}^{s}\frac{\lambda_i+\mu_j+d-i-j+1}{d-i-j+1}\Bigg)
\Bigg(\prod_{i=1}^{r}\prod_{t=1}^{d-r-s}\frac{\lambda_i+r+t-i}{r+t-i}\Bigg)
\Bigg(\prod_{j=1}^{s}\prod_{t=1}^{d-r-s}\frac{\mu_j+s+t-j}{s+t-j}\Bigg).
\end{align}
The final two products are ratios of \textit{rising factorials}, which satisfy the following equality (\cite{DLMF}, Equation 5.2.5):
\begin{equation}
\label{eq:rising_factorial_cancellation_identity}
\prod_{t=1}^{m}\frac{x+\ell+t-1}{x+t-1}
\;=\;
\frac{(x+\ell)_m}{(x)_m}
\;=\;
\prod_{q=0}^{\ell-1}\frac{x+m+q}{x+q}
\qquad\qquad
\end{equation}
This rewrite allows us to express each of the five terms as products over $O(n^2)$ terms, where every numerator and denominator requires at most $\wt{O}(\log d)$ bits to express. Hence, $m_\rho$ can be computed using $\wt{O}(n^2\log d + \log(1/\varepsilon))$ gates. 
\end{proof}

\begin{fact}
    \label{fact:quotient_rule}
    For all $a, b >0$,
    \begin{align}
        \left\lvert
            \frac{a \pm \Delta_a}{b \pm \Delta_b}
            -
            \frac{a}{b}
        \right\rvert
        \le
        \left\lvert
            \frac{
                a \Delta_b + b \Delta_a
            }{
                b (b - \Delta_b)
            }
        \right\rvert
    \end{align}
\end{fact}
\begin{proof}
    Let $\abs{\delta_a} \le \Delta_a$ and $\abs{\delta_b} \le \Delta_b$.
    \begin{align}
        \left\lvert
            \frac{a + \delta_a}{b + \delta_b}
            -
            \frac{a}{b}
        \right\rvert
        =
        \left\lvert
            \frac{
                b(a + \delta_a) - a(b + \delta_b)
            }{
                b(b + \delta_b)
            }
        \right\rvert
        =
        \left\lvert
            \frac{
                b\delta_a + a \delta_b
            }{
                b(b + \delta_b)
            }
        \right\rvert
        \le
        \left\lvert
            \frac{
                b\Delta_a + a \Delta_b
            }{
                b(b - \Delta_b)
            }
        \right\rvert
    \end{align}
\end{proof}

\begin{lemma}
    \label{lem:appendix_ratio_closeness}
    Assume that $A \in \{P_n(d), P_{n - \frac12}(d), B_n(d), B_{r,s}(d)\}$, with subalgebra $C$ given by the chains in~\cref{sec:subalgebra_chains}. Let $r_{\rho, \sigma}$ be defined as in~\cref{eq:step_two_embed_eq}. Then,
    \begin{equation}
        \left|\frac{||E^\rho_{P \circ \rho, Q \circ \rho}||^2_2}{||E^{\sigma}_{PQ}||^2_2} - r_{\rho, \sigma}\right | \le O(\delta) \quad  
    \end{equation}
    and
    \begin{equation}
        \label{eq:second_eq_ratio}
        \left|\frac{||E^\rho_{P \circ \rho, Q \circ \rho}||_2}{||E^{\sigma}_{PQ}||_2} - \sqrt{r_{\rho, \sigma}}\right | \le O(\delta)
    \end{equation}
    where $\delta \coloneqq d^{-1/2} \cdot |A|$. 
\end{lemma}
\begin{proof}
By~\cref{cor:norm_of_fourier_states} and the definition of $r_{\rho,\sigma}$ in
\cref{eq:step_two_embed_eq},
\begin{equation}
    \label{eq:r_estimate}
    \frac{||E^\rho_{P\circ\rho,Q\circ\rho}||_2^2}
         {||E^\sigma_{PQ}||_2^2}
    =
    r_{\rho,\sigma}\frac{1+O(\delta)}{1+O(\delta)}
    =
    r_{\rho,\sigma}(1+O(\delta)).
\end{equation}
Here, the Schur inner product scaling factors are $d^n$ for both
algebras when $A=P_n(d)$, and $d^n$ for $A$ and $d^{n-1}$ for $C$
in the other cases. Next, by~\cref{eq:embedding_weights_sum_one}, the positive coefficients
$r_{\rho,\sigma}$ sum to one, so $0<r_{\rho,\sigma}\le 1$.
Consequently,
\begin{equation}
    \left|
    \frac{||E^\rho_{P\circ\rho,Q\circ\rho}||_2^2}
         {||E^\sigma_{PQ}||_2^2}
    -r_{\rho,\sigma}
    \right|
    \le O(\delta) \cdot r_{\rho,\sigma}
    \le O(\delta).
\end{equation}
Finally, taking square roots of \cref{eq:r_estimate} gives
\begin{equation}
    \frac{||E^\rho_{P\circ\rho,Q\circ\rho}||_2}
         {||E^\sigma_{PQ}||_2}
    =
    \sqrt{r_{\rho,\sigma}}(1+O(\delta))
\end{equation}
and hence
\begin{equation}
    \left|
    \frac{||E^\rho_{P\circ\rho,Q\circ\rho}||_2}
         {||E^\sigma_{PQ}||_2}
    -\sqrt{r_{\rho,\sigma}}
    \right|
    \le O(\delta) \cdot \sqrt{r_{\rho,\sigma}}
    \le O(\delta).
\end{equation}
\end{proof}

\section{Rules For Extending Bratteli Paths}
\label{appendix:bratteli_path_extension_rules}
In this section, we provide proofs of the extension rules stated in~\cref{sec:apply_irreps}, beginning with~\cref{lem:brauer_extension}:
\begin{lemma}
    \label{lem:brauer_extension_appendix}
     Let $B = B_{n-2}(d)$, $A = B_n(d)$. For any diagram $D \in B$ with $\pn(D) = 0$, the isometry
    \begin{equation}
        \mathbf{F}(\wt{\ft_B}\ket{D}) = \wt{\ft_A}\ket{D e_{n-1}}
    \end{equation}
    can be implemented up to operator norm error $O(\poly(|A|) \cdot \delta)$ by the map
    \begin{equation}
        \wt{\mathbf{F}}\ket{\emptyset, P, Q} = \ket{\emptyset, P \circ \Box \circ \emptyset, Q \circ \Box \circ \emptyset}
    \end{equation} 
\end{lemma}
\begin{proof}
 Since $D$ has propagating number zero, \cref{thm:zero_irrep_matrix} implies that every $\ket{\lambda,P,Q}$ in the support of $\wt{\ft_B}\ket{D}$ satisfies $\lambda=P(n-2)=Q(n-2)=\emptyset$:
\begin{equation}
    \label{eq:brauer_embedding_the_start_state}
    \wt{\ft_B}\ket{D} = \sum_{PQ} \|E_{PQ}^\emptyset\|_2 \cdot \emptyset(D)_{PQ} \ket{\emptyset, P, Q}.
\end{equation}
After applying $\mathbf{E}(A,B)$, we obtain
\begin{equation}
    \wt{\ft_A}\ket{D} = \sum_{\rho_n} \sum_{RS} \|E_{RS}^{\rho_n}\|_2\cdot \rho_n(D)_{RS} \ket{\rho_n, R, S},
\end{equation}
where $R = P \circ \rho_{n-1}\circ \rho_n$ and $S = Q \circ \rho_{n-1}\circ \rho_n$ range over the valid extensions of the Bratteli paths to basis vectors of $V^{\rho_n}$, with $\rho_n \in \wh{B_n(d)}$. 

 Now, we apply a controlled-$\rho(e_{n-1})$ operation to the row register, resulting in the state $\wt{\ft_A}\ket{De_{n-1}}$ up to $O(\poly(|A|) \cdot \delta)$ error (\cref{eq:applying_irreps}), and inspect the resulting state. Since $De_{n-1}$ also has propagating number zero,~\cref{thm:zero_irrep_matrix} implies that the only $\ket{\rho_{n}, R, S}$ with nonzero amplitude in $\wt{\ft_A}\ket{De_{n-1}}$ correspond to $\rho_{n} = \emptyset$, $R = P \circ\Box\circ\emptyset$ and $S = Q \circ\Box\circ\emptyset$:
\begin{equation}
    \label{eq:replaced_state_brauer}
    \wt{\ft_A}\ket{De_{n-1}} = \sum_{PQ} \|E_{P \circ \Box \circ \emptyset,\; Q \circ \Box \circ \emptyset}^\emptyset\|_2\cdot \emptyset(De_{n-1})_{P \circ \Box \circ \emptyset, Q \circ \Box \circ \emptyset}
    \ket{\emptyset, P \circ \Box \circ \emptyset, Q \circ \Box \circ \emptyset}
\end{equation}
Our remaining task is to prove that applying $\wt{\mathbf{F}}$ to $\wt{\ft_B}\ket{D}$ exactly recovers the state in~\cref{eq:replaced_state_brauer}, which we do by comparing coefficients with~\cref{eq:brauer_embedding_the_start_state}. To begin,
\begin{align}
\emptyset(De_{n-1})_{P \circ \Box \circ \emptyset, Q \circ \Box \circ \emptyset}
&= \sum_S
\emptyset(D)_{P \circ \Box \circ \emptyset, S}\,
\emptyset(e_{n-1})_{S,Q \circ \Box \circ \emptyset} \\
&= \emptyset(D)_{P \circ \Box \circ \emptyset, Q \circ \Box \circ \emptyset}\, \label{eq:contraction_matrix_is_identity}
\emptyset(e_{n-1})_{Q \circ \Box \circ \emptyset,Q \circ \Box \circ \emptyset} \\
&= \emptyset(D)_{PQ}\,
\emptyset(e_{n-1})_{Q \circ \Box \circ \emptyset, Q \circ \Box \circ \emptyset} \\
&= d\,\emptyset(D)_{PQ}.
\end{align}
The first line expands out the definition of a matrix element. The second follows from~\cref{eq:brauer_matrix_form_eq_three}, which implies that $\emptyset(e_{n-1})_{S,Q \circ \Box \circ \emptyset} = 0$ unless $S = Q \circ \Box \circ \emptyset$. The third follows from ~\cref{eq:subalgebra_adapted_matrices}, and the fourth from~\cref{eq:brauer_matrix_form_eq_two}. 

Hence, it suffices to show that $d\cdot \|E_{P \circ \Box \circ \emptyset,\; Q \circ \Box \circ \emptyset}^\emptyset\|_2 = \|E_{PQ }^\emptyset\|_2$, where the left-hand-side is the norm of a Fourier state in $B_n(d)$, and the right-hand side is the norm of a Fourier state in $B_{n-2}(d)$.

To do so, we first prove that 
\begin{equation}
    \label{eq:intermediate_equality}
    d \cdot E_{P \circ \Box \circ \emptyset,\; Q \circ \Box \circ \emptyset}^\emptyset = E_{PQ }^\emptyset e_{n-1}
\end{equation}
as elements of $B_n(d)$.\ By the definition of a Fourier state, the left-hand side of~\cref{eq:intermediate_equality} is the preimage of $d\ket{P \circ \Box \circ \emptyset}\bra{Q \circ \Box \circ \emptyset}_{\emptyset}$ under the Fourier isomorphism (\cref{eq:semisimple_algebra}) for $B_n(d)$.

On the right-hand side of~\cref{eq:intermediate_equality}, $E_{PQ }^\emptyset$ is the preimage of $\ket{P}\bra{Q}_{\emptyset}$ under the Fourier isomorphism for $B_{n-2}(d)$.  Applying~\cref{eq:subalgebra_chain} twice for the Brauer subalgebra chain implies that as an element of $B_n(d)$, $E_{PQ}^\emptyset$ maps to $\ket{P \circ \Box \circ \emptyset}\bra{Q \circ \Box \circ \emptyset}_{\emptyset}$ under the Fourier isomorphism for $B_n(d)$. Furthermore, on the irrep space $V^\emptyset$, it is easy to check using~\cref{eq:brauer_matrix_form_eq_three} that
\begin{equation}
    \label{eq:id_action_of_contraction}
    \emptyset(e_{n-1})\ket{Q \circ \Box \circ \emptyset} = d \cdot \ket{Q \circ \Box \circ \emptyset}
\end{equation}
This fact was already used in~\cref{eq:contraction_matrix_is_identity}). Thus, $E_{PQ}^\emptyset e_{n-1}$ maps to $d\ket{P \circ \Box \circ \emptyset}\bra{Q \circ \Box \circ \emptyset}_{\emptyset}$ under the same isomorphism. Since the two images are equal,~\cref{eq:intermediate_equality} holds, and therefore 
\begin{equation}
    d \cdot ||E_{P \circ \Box \circ \emptyset,\; Q \circ \Box \circ \emptyset}^\emptyset||_2 = ||E_{PQ}^\emptyset e_{n-1}||_2
\end{equation}
Moreover, 
\begin{equation}
   ||E_{PQ}^\emptyset e_{n-1}||_2 = \left\|d_{\emptyset}\sum_{a \in \mathcal{B}(B_{n-2}(d))} \emptyset(a^*)_{QP}\cdot ae_{n-1}\right\|_2 = \left\|d_{\emptyset}\sum_{a \in \mathcal{B}(B_{n-2}(d))} \emptyset(a^*)_{QP} \cdot a\right\|_2 = ||E_{PQ}^\emptyset||_2
\end{equation}
where in the second equality we used that $a \mapsto ae_{n-1}$ sends distinct basis elements in $B_{n-2}(d)$ to distinct basis elements in $B_{n}(d)$. Hence, 
\begin{equation}
    d \cdot ||E_{P \circ \Box \circ \emptyset,\; Q \circ \Box \circ \emptyset}^\emptyset||_2 = ||E_{PQ}^\emptyset||_2
\end{equation}
which concludes the proof. 
\end{proof}
We now move on to the walled Brauer algebra (\cref{lem:walled_brauer_extension_rule}):
\begin{lemma}
    \label{lem:walled_brauer_extension_appendix}
    Let $B = B_{r-1, r-1}(d)$, $A = B_{r,r}(d)$. For any diagram $D \in B$ with $\pn(D) = 0$, the isometry
    \begin{equation}
        \mathbf{F}(\wt{\ft_B}\ket{D}) = \wt{\ft_A}\ket{D f_{r}}
    \end{equation}
    can be implemented up to operator norm error $O(\poly(|A|) \cdot \delta)$ by the map 
    \begin{equation}
        \wt{\mathbf{F}}\ket{(\emptyset, \emptyset), P, Q} = \ket{(\emptyset, \emptyset), \;P \circ (\Box, \emptyset) \circ (\emptyset, \emptyset), \;Q \circ (\Box, \emptyset) \circ (\emptyset, \emptyset)}
    \end{equation}
\end{lemma}
\begin{proof}
    This proof is identical to the proof of~\cref{lem:brauer_extension_appendix}, except that the only valid extension paths after applying $\rho(f_r)$ (\cref{eq:replaced_state_brauer}) are $\rho_{n} = \emptyset$, $R = (P \circ (\Box, \emptyset) \circ (\emptyset, \emptyset))$ and $S = Q \circ (\Box, \emptyset) \circ (\emptyset, \emptyset)$. For brevity, we will not repeat the full proof here. 
\end{proof}
Next, we handle the propagating number zero case for the partition algebra (\cref{lem:efficient_add_a_point}): 
\begin{lemma}
\label{lem:efficient_add_a_point_appendix}
    Let $B = P_{n - 1}(d)$, $A = P_{n}(d)$. For any basis element $b = d^{({n - 1 - \cc(D)})/2}D \in \mathcal{B}(B)$ with $\pn(D) = 0$, the isometries
    \begin{equation}
        \mathbf{F}_1(\wt{\ft_B}\ket{b}) = \wt{\ft_A}\ket{d^{(n - 2 - \cc(D))/2} Dp_{n}}
    \end{equation}
    \begin{equation}
        \mathbf{F}_2(\wt{\ft_B}\ket{b}) = \wt{\ft_A}\ket{d^{(n - 1 - \cc(D))/2}b_{n-1}D p_{n}}
    \end{equation}
    \begin{equation}
        \mathbf{F}_3(\wt{\ft_B}\ket{b}) = \wt{\ft_A}\ket{d^{(n - 1 - \cc(D))/2}D p_{n}b_{n-1}}
    \end{equation}
    \begin{equation}
        \mathbf{F}_4(\wt{\ft_B}\ket{b}) = \wt{\ft_A}\ket{d^{(n -  \cc(D))/2} b_{n-1}D p_{n}b_{n-1}}
    \end{equation}
    can be implemented up to operator norm error $O(\poly(|A|) \cdot \delta)$ as follows: 
    \begin{equation}
        \wt{\mathbf{F}}_1\ket{\emptyset, P, Q} = \ket{\emptyset,\;\;P_{\le 2n - 4} \circ \emptyset \circ \emptyset \circ \emptyset \circ \emptyset,\;\; Q_{\le 2n - 4}\circ \emptyset \circ \emptyset \circ \emptyset \circ \emptyset}
    \end{equation}
    \begin{equation}
        \wt{\mathbf{F}}_2\ket{\emptyset, P, Q} = \ket{\emptyset,\;\;P_{\le 2n - 4} \circ \emptyset \circ \Box \circ \emptyset \circ \emptyset,\;\; Q_{\le 2n - 4}\circ \emptyset \circ \emptyset \circ \emptyset \circ \emptyset}
    \end{equation}
     \begin{equation}
        \wt{\mathbf{F}}_3\ket{\emptyset, P, Q} = \ket{\emptyset,\;\;P_{\le 2n - 4} \circ \emptyset \circ \emptyset \circ \emptyset \circ \emptyset,\;\; Q_{\le 2n - 4}\circ \emptyset \circ \Box \circ \emptyset \circ \emptyset}
    \end{equation}
     \begin{equation}
        \wt{\mathbf{F}}_4\ket{\emptyset, P, Q} = \ket{\emptyset,\;\;P_{\le 2n - 4} \circ \emptyset \circ \Box \circ \emptyset \circ \emptyset,\;\; Q_{\le 2n - 4}\circ \emptyset \circ \Box \circ \emptyset \circ \emptyset}
    \end{equation}
\end{lemma}
\begin{proof}
    For $\mathbf{F}_1$, we follow the proof of~\cref{lem:brauer_extension_appendix}. Since $D$ has propagating number zero,~\cref{thm:zero_irrep_matrix} implies that $P(2n - 2) = Q(2n - 2) = \emptyset$. This also implies that $P(2n - 3) = Q(2n - 3) = \emptyset$, since the number of boxes cannot decrease from level $2n - 3$ to level $2n - 2$. Applying $\mathbf{E}(A, B)$, we obtain the state 
    \begin{equation}
        \wt{\ft_A}\ket{b} = \sum_{\rho_{2n}}\sum_{RS} ||E_{RS}^{\rho_{2n}}|| \cdot \rho_{2n}(b)_{RS}\ket{\rho_{2n}, R, S}
    \end{equation}
    Where $R = P \circ \rho_{2n -1} \circ \rho_{2n}$, $S = Q \circ \rho_{2n-1} \circ \rho_{2n}$.
    
    We will now apply a controlled-$\rho(d^{-1/2}p_n)$ operation to the row register, which up to $O(\poly(|A|) \cdot \delta)$ error results in the state $\wt{\ft_A}\ket{d^{(n - 2 - \cc(D))/2} Dp_{n}}$. Since $Dp_n$ has propagating number zero,~\cref{thm:zero_irrep_matrix} implies that the only $\ket{\rho_{2n}, R, S}$ with nonzero amplitude correspond to $\rho_{2n} = \emptyset$, $R = P \circ\emptyset\circ\emptyset$ and $S = Q \circ\emptyset\circ\emptyset$:
    \begin{align}
        \wt{\ft_A}\ket{d^{(n - 2 - \cc(D))/2} Dp_{n}} \\ = \sum_{PQ} ||E_{P \circ \emptyset \circ \emptyset, Q \circ \emptyset \circ \emptyset}^{\emptyset}||_2 \cdot \emptyset(d^{(n - 2 - \cc(D))/2}Dp_{n})_{P \circ\emptyset\circ\emptyset, Q \circ\emptyset\circ\emptyset}\ket{\emptyset,\;P \circ \emptyset \circ \emptyset,\; Q \circ \emptyset \circ \emptyset} \label{eq:end_f_one}
    \end{align}
We then repeat the same argument from the proof of~\cref{lem:brauer_extension_appendix}, beginning after~\cref{eq:replaced_state_brauer}. The role of $e_{n-1}$ is replaced with $p_n$. Using the formulas for the matrix entries of $\emptyset(p_n)$ (\cref{eq:partition_point_generator_one},~\cref{eq:partition_point_generator_two}), 
\begin{equation}
    \emptyset(p_n)\ket{P \circ \emptyset \circ \emptyset} = d \ket{P \circ \emptyset \circ \emptyset}
\end{equation}
the analog of~\cref{eq:id_action_of_contraction} for the point generator. Following the same argument allows us to replace the coefficients in~\cref{eq:end_f_one} by their counterparts from $\wt{\ft_B}\ket{b}$, showing that simply extending each Bratteli path by $\dots \emptyset,\emptyset$ suffices to approximate $\mathbf{F}_1$ to $O(\poly(|A|) \cdot \delta)$ error. 

    Now, we move on to $\mathbf{F}_2$. Again up to $O(\poly(|A|) \cdot \delta)$ error, left-multiplying the state in~\cref{eq:end_f_one} by $\emptyset(d^{1/2}b_{n-1})$ gives the desired action, so it suffices to first apply $\wt{\mathbf{F}}_1$ and then apply $\emptyset(d^{1/2}b_{n-1})$ to the row register. After applying $\wt{\mathbf{F}}_1$, the final four irreps of the path $P$ are all $\emptyset$. Using ~\cref{eq:bridge_gen_formula_one} and ~\cref{eq:bridge_gen_formula_two}, we can calculate the action of $\emptyset(d^{1/2}b_{n-1})$ on $\ket{P}$:
    \begin{align}
        \emptyset(d^{1/2}b_{n-1})\ket{P} = \emptyset(d^{1/2}b_{n-1})\ket{P_{\le 2n - 4} \circ \emptyset \circ \emptyset \circ \emptyset \circ \emptyset} \\ = \frac{1}{\sqrt{d}}\ket{P_{\le 2n - 4} \circ \emptyset \circ \emptyset \circ \emptyset \circ \emptyset} + \frac{\sqrt{(d - 1)d}}{d}\ket{P_{\le 2n - 4} \circ \emptyset \circ \Box \circ \emptyset \circ \emptyset}
    \end{align}
    which can be simulated up to operator norm error $O(\poly(|A|) \cdot d^{-1/2}) = O(\poly(|A|) \cdot \delta)$ by simply replacing $P(2n - 2)$ with $\Box$. Approximating $\mathbf{F}_3$ and $\mathbf{F}_4$ can be done in the same manner, except that for $\wt{\mathbf{F}_3}$ we replace $Q(2n - 2)$ with $\Box$ instead of $P(2n -2)$, and for $\wt{\mathbf{F}_4}$ we replace both.
\end{proof}
Finally, we prove~\cref{lem:efficient_add_a_bridge}. This is the most involved of the proofs in this section. The bulk of the proof is contained in the following Lemma: 
\begin{lemma}
    \label{lem:efficient_add_a_bridge_appendix}
    Let $B = P_{n-1}(d)$, $A = P_{n - \frac12}(d)$. For any basis element $b = d^{({n - 1 - \cc(D)})/2}D \in \mathcal{B}(B)$, the maps 
    \begin{equation}
        \mathbf{H}_1(\wt{\ft_B}\ket{b}) = \wt{\ft_A}\ket{d^{(n - \cc(D))/2} b_{n-1}D}
    \end{equation}
    \begin{equation}
        \mathbf{H}_2(\wt{\ft_B}\ket{b}) = \wt{\ft_A}\ket{d^{(n - \cc(D))/2} Db_{n-1}}
    \end{equation}
    can be implemented to operator norm error $O(\poly(|A|) \cdot (\delta + \varepsilon))$ by operators $\wt{\mathbf{H}}_1$ and $\wt{\mathbf{H}}_2$. Both $\wt{\mathbf{H}}_1$ and $\wt{\mathbf{H}}_2$ can be implemented with $\wt{O}(n^{3/2} \cdot (n + \log d + \log(1/\varepsilon)))$ gates. 
\end{lemma}
\begin{proof}
    We focus on $\mathbf{H}_1$, since the proof for $\mathbf{H}_2$ will be the same up to exchanging the role of the row and column registers. First, we would like to understand the ideal action of $\mathbf{H}_1$. Given an input state
    \begin{equation}
        \label{eq:bridge_add_starting_state}
        \wt{\ft_B}\ket{b} = \sum_{\rho}\sum_{PQ} ||E_{PQ}^\rho||_2 \rho(b)_{PQ}\ket{\rho, P, Q}
    \end{equation}
    we can simulate $\mathbf{H}_1$ by first applying  $\mathbf{E}(P_{n - \frac12}(d), B)$, obtaining the state
    \begin{equation}
        \label{eq:after_half_embedding}
     \wt{\ft_{P_{n - \frac12}(d)}}\ket{b} =\sum_{\rho}\sum_{PQ}\sum_{\substack{\tau \in \wh{P_{n - \frac12}(d)}: \\ \rho \in \text{Res}^{P_{n - \frac12}(d)}_B(\tau)}}||E_{P \circ \tau, Q \circ \tau}^{\tau}||_2 \cdot \tau(b)_{P \circ \tau, Q \circ \tau}\ket{\tau, P \circ \tau, Q \circ \tau}
    \end{equation}
    and then applying a controlled-$\tau(d^{1/2}b_{n-1})$ operation to the row register, $\ket{P \circ \tau}$.\ From~\cref{eq:bridge_gen_formula_one} and \cref{eq:bridge_gen_formula_two}, applying $\tau(d^{1/2}b_{n-1})$ only acts on the register containing the last irrep of $P$, $P(2n - 2)$.\ To determine this action, we can split into two cases based on the adjacent irreps in the Bratteli path: 
    \begin{itemize}
        \item $\tau \ne P(2n - 3)$. In this case, $P^\prime \sim_{2n - 2} P \circ \tau \implies P^\prime = P \circ \tau$, but then the third case of~\cref{eq:bridge_gen_formula_one} implies that $\tau(d^{1/2}b_{n-1})\ket{P \circ\tau} = 0$. Hence, no paths with $\tau \ne P(2n - 3)$ survive the application of $\tau(d^{1/2}b_{n-1})$.
        \item $\tau = P(2n - 3)$. Following the notation and results of~\cref{lem:bridge_equivalence_class_structure}, $P(2n - 2)$ is projected onto the one-dimensional eigenspace of $\tau(b_{n-1})|_{V_{[P]}}$, spanned by the unit vector 
        \begin{equation}
            \ket{\psi_{[P]}} = \sum_{R \sim_{2n - 2} P } \alpha_{R}\ket{R \circ P(2n - 3)}
        \end{equation}
    \end{itemize}
    Hence, after applying controlled-$\tau(d^{1/2}b_{n-1})$ to the row register, we obtain the state 
    \begin{align}
        \sum_{\rho}\sum_{PQ} \sum_{\substack{\tau \in \wh{P_{n - \frac12}(d)}: \\ \rho \in \text{Res}^{P_{n - \frac12}(d)}_B(\tau)}} d^{1/2} \cdot  \alpha_{P} \cdot  ||E_{P \circ \tau, Q \circ \tau}^{\tau}||_2 \cdot \tau(b)_{P \circ \tau, Q \circ \tau} \sum_{R \sim_{2n - 2} P} \alpha_R  \ket{\tau, R \circ P(2n - 3), Q \circ \tau}
    \end{align}
    which, up to $O(\poly(|A|) \cdot \delta)$ error (\cref{eq:applying_irreps}), is equivalent to $\wt{\ft_{P_{n - \frac12}(d)}}\ket{d^{(n - \cc(D))/2}b_{n-1}D}$.
    
    Moreover, $\ket{\tau, R \circ P(2n - 3), Q \circ \tau}$ has zero amplitude in $\wt{\ft_{P_{n - \frac12}(d)}}\ket{d^{(n - \cc(D))/2}b_{n-1}D}$ if $\tau \ne P(2n - 3)$, since all Bratteli paths in a Fourier state must end with the corresponding irrep label. Fixing $\tau = P(2n - 3)$, we conclude that 
    \begin{align}
        \label{eq:exact_bridge_add}
        \wt{\ft_{P_{n - \frac12}(d)}}\ket{d^{(n - \cc(D))/2}b_{n-1}D} \\ = \sum_{\rho}\sum_{PQ}  d^{1/2} \cdot  \alpha_{P} \cdot  ||E_{P \circ \tau, Q \circ \tau}^{\tau}||_2 \cdot \tau(b)_{P \circ \tau, Q \circ \tau} \sum_{R \sim_{2n - 2} P} \alpha_R  \ket{\tau, R \circ \tau, Q \circ \tau}
    \end{align}
    up to $O(\poly(|A|) \cdot \delta)$ additive error on each coefficient. 
    
    To (approximate) $\wt{\mathbf{H}}_1$, we must find an efficiently implementable unitary which (approximately) maps the input state, $\wt{\ft}_B\ket{b}$, to the state in ~\cref{eq:exact_bridge_add}. Even though both $\wt{\ft}_B\ket{b}$ and $\wt{\ft_{P_{n - \frac12}(d)}}\ket{d^{(n - \cc(D))/2}b_{n-1}D}$ are (approximately) unit vectors, finding a unitary transformation which takes one to the other seems difficult \textit{a priori}, since the controlled-$\tau(d^{1/2}b_{n-1})$ operation is highly non-unitary, and we do not have the $\pn(D) = 0$ constraint which was instrumental in the previous lemmas. However, by applying a global operation on both the row and column registers, we can implement the desired transformation. Starting from the state $\wt{\ft}_B\ket{b}$ in~\cref{eq:bridge_add_starting_state}, we apply the following operations: 
   \begin{algorithm}[H]
    \caption{Implementing $\wt{\mathbf{H}}_1$.}
    \label{alg:bridge_add_extension}
    \begin{algorithmic}[1]
        \State Conditioned on the registers containing $\ket{Q(2n - 2)}$ and $\ket{P(2n - 3)}$, compute a circuit description of a unitary $U_{Q(2n - 2)}$ implementing
        \begin{equation}
            \ket{Q(2n - 2)\circ P(2n - 3)}
            \mapsto
            \ket{\psi_{[P]}}
            =
            \sum_{R \sim_{2n - 2} P} \alpha_R \ket{R \circ P(2n - 3)} .
        \end{equation}
        
        \State Apply $U_{Q(2n - 2)}$ to the register containing $\ket{P(2n - 2)}$, i.e.\ the final  irrep in the Bratteli path $\ket{P}$. Note that $P(2n - 2)=Q(2n - 2)=\rho$. 
        
        \State Using $\ket{Q(2n - 2)}$ and $\ket{P(2n - 3)}$ (which are unaffected by Step~2), uncompute the description of this circuit.
        
        \State Extend the Bratteli path $\ket{Q}$ to $\ket{Q \circ P(2n - 3)}$ by fanning out $\ket{P(2n - 3)}$.

        \State Controlled on $\ket{Q\circ P(2n-3)}$, apply to the
        irrep label register the unitary that exchanges
        $\ket{Q(2n-2)}$ and $\ket{P(2n-3)}$ and fixes all other
        basis states. If the two labels coincide, apply the identity.
    \end{algorithmic}
    \end{algorithm}
First, we show that $\wt{\mathbf{H}}_1$ has an efficient implementation:
\begin{lemma}
    $\wt{\mathbf{H}}_1$ can be implemented to operator norm error $O(\varepsilon \cdot |A|)$ with $\wt{O}(n^{3/2} \cdot (n + \log d + \log(1/\varepsilon)))$ gates. 
\end{lemma}
\begin{proof}
    In order to compute a description of $U_{Q(2n - 2)}$, we first need a description of $\ket{\psi_{[P]}}$. Since $\ket{\psi_{[P]}}$ is supported on $V_{[P]}$, it suffices to compute $(R, \alpha_R)$ for every $R \in [P]$, of which there are at most $O(\sqrt{n})$.\ Since $\tau(b_{n-1})$ is a real matrix, 
    \begin{equation}
        \label{eq:coefficients_of_p}
        \alpha_R^2 = \braket{R|\psi_{[P]}}\braket{\psi_{[P]}|R} = \left(\tau(b_{n-1})|_{V_{[P]}}\right)_{RR}
    \end{equation}
    which can be computed to additive error $\varepsilon$ using $\wt{O}(n \cdot (n + \log d + \log(1/\varepsilon)))$ gates (\cref{item:bridge_time_complexity}). Hence, computing a complete description of $\ket{\psi_{[P]}}$ requires $\wt{O}(n^{3/2} \cdot (n + \log d + \log(1/\varepsilon)))$ gates. After computing this description, one way to implement $U_{Q(2n - 2)}$ is to first map $\ket{Q(2n - 2)}$ to some fixed (non-irrep) state $\ket{\perp}$, and then apply the same Givens rotation technique used in~\cref{eq:givens_rotation}. The gate complexity of Lines 1-3 is dominated by the cost of computing the description of $\ket{\psi_{[P]}}$.
    
    Finally, Lines 4 and 5 require $\wt{O}(n^2)$ gates, since $Q \circ \tau$ is a length-$O(n)$ Bratteli path where each irrep is represented with $\wt{O}(n)$ qubits. Since all entries were computed to additive error $\varepsilon$, ~\cref{lem:opnorm_bound} implies an operator norm error of $O(\varepsilon \cdot |A|)$.
\end{proof}
\noindent Now, we analyze the state after applying $\wt{\mathbf{H}}_1$. After Lines 1-3, $\wt{\ft}_B\ket{b}$ is transformed into
\begin{align}
     \sum_{\rho}\sum_{PQ}  ||E_{PQ}^{\rho}||_2 \cdot \rho(b)_{PQ}\sum_{R \sim_{2n - 2} P} \alpha_{R}  \ket{\rho, R \circ P(2n - 3), Q}
    \end{align}
After Line 4, $\ket{Q}$ is transformed to $\ket{Q \circ P(2n - 3)}$, and after line 5 the irrep label $\ket{\rho}$ becomes $\ket{P(2n - 3)}$, since $\rho$ is the final irrep of the Bratteli path $Q$. Hence, the final state after applying $\wt{\mathbf{H}}_1$ is:
\begin{align}
    \label{eq:approx_bridge_add}
     \sum_{\rho}\sum_{PQ}  ||E_{PQ}^{\rho}||_2 \cdot \rho(b)_{PQ}\sum_{R \sim_{2n - 2} P} \alpha_{R}  \ket{P(2n - 3), R \circ P(2n - 3), Q \circ P(2n - 3)}
\end{align}

The remaining task is to compare the coefficients of the
implemented and desired states. We retain $P,Q$ for the input
paths, write $R$ for the output row path, and set
$\rho=P(2n-2)=Q(2n-2)$ and $\tau=P(2n-3)$.

First, we claim that $
    |\alpha_P|^2 = {m_\rho}/{m_\tau}$.
Indeed, writing $e=b_{n-1}$, the Schur representation $\mathbf{S}_A$ (taken with respect to the algebra $A$) satisfies
\begin{equation}
    \operatorname{Tr}_n\left(\mathbf{S}_A(e)\right) = I.
\end{equation}
For any $x\in B$, we also have $\mathbf{S}_A(x) = \mathbf{S}_B(x)\otimes I$ by definition of the Schur representation (\cref{fig:schur_reps}). Therefore,
\begin{equation}
    m_\rho
    =
    \operatorname{Tr}\left(
        \mathbf{S}_A(e)\cdot\mathbf{S}_A(E^\rho_{PP})
    \right)
    =
\sum_{\lambda:\rho\in\operatorname{Res}_B^A(\lambda)}
        m_\lambda\cdot
        \lambda(e)_{P\circ\lambda,P\circ\lambda}
    =
    m_\tau\cdot|\alpha_P|^2.
\end{equation}
The last equality follows from the third case of
\cref{eq:bridge_gen_formula_one}, which gives
$\lambda(e)_{P\circ\lambda,P\circ\lambda}=0$
unless $\lambda=P(2n-3)=\tau$, and from
\cref{eq:coefficients_of_p}, which identifies the surviving
diagonal entry with $|\alpha_P|^2$.
This identity holds whether or not $\rho$ and $\tau$ have
the same partition label.

We now use $|\alpha_P|^2 = m_\rho/m_\tau$ to show that
\cref{alg:bridge_add_extension} prepares the desired Fourier
state up to the claimed error. Since the algorithm replaces the application of a 
bridge generator by preparation of its unique normalized nonzero eigenvector $\ket{\psi_{[P]}}$, it remains to check that the scaling factor introduced
by the bridge operation is canceled, up to $O(\delta)$ error,
by the ratio of the Fourier basis norms.

By construction, the state $\ket{\psi_{[P]}}$ has eigenvalue $1$ under
$\tau(b_{n-1})$. Thus, on the relevant path subspace,
the (scaled) bridge generator acts as
\begin{equation}
    \tau(d^{1/2}b_{n-1})\ket{P\circ\tau}
    =
    \sqrt{d}\cdot \alpha_P \cdot\ket{\psi_{[P]}}.
\end{equation}
Since~\cref{alg:bridge_add_extension}
prepares $\ket{\psi_{[P]}}$ directly, we must account
for the factor $\sqrt{d}\cdot\alpha_P$ above.
By~\cref{cor:norm_of_fourier_states} and
$|\alpha_P|^2=m_\rho/m_\tau$,
\begin{equation}
    \label{eq:norm_comparison_appendix}
    \sqrt{d}\cdot\alpha_P\cdot
    \frac{||E^\tau_{R\circ\tau,Q\circ\tau}||_2}
         {||E^\rho_{PQ}||_2}
    =
    \alpha_P\cdot\sqrt{\frac{m_\tau}{m_\rho}}
    \cdot(1+O(\delta))
    =
    1+O(\delta).
\end{equation}
Hence, the ratio of the Fourier basis norms cancels
this factor up to $O(\delta)$ error.

Finally, to compare the output coefficients explicitly, fix
$\ket{\tau,R\circ\tau,Q\circ\tau}$.
There is only one contributing input path $P$: it agrees
with $R$ before level $2n-2$, and its final label must be
$Q(2n-2)=\rho$.
The coefficient produced by~\cref{alg:bridge_add_extension} is
\begin{equation}
    c_{\mathrm{alg}}
    =
    \alpha_R\cdot||E^\rho_{PQ}||_2\cdot\rho(b)_{PQ},
\end{equation}
whereas the coefficient in the desired Fourier state is
\begin{equation}
    c_{\mathrm{ideal}}
    =
    \sqrt{d}\cdot\alpha_R\cdot\alpha_P\cdot
    ||E^\tau_{R\circ\tau,Q\circ\tau}||_2\cdot
    \rho(b)_{PQ}.
\end{equation}
By~\cref{eq:norm_comparison_appendix},
\begin{equation}
    |c_{\mathrm{ideal}}-c_{\mathrm{alg}}|
    \le O(\delta)\cdot|c_{\mathrm{alg}}|
    \le O(\poly(|A|)\cdot\delta)
\end{equation}
where the last inequality uses $|\alpha_R|\le 1$ and the fact
that $||E^\rho_{PQ}||_2\cdot\rho(b)_{PQ}$ is a coefficient
of the approximately normalized input state
$\wt{\ft}_B\ket{b}$.
Factoring in the implementation error gives the required
coefficient-wise bound.

Applying~\cref{lem:opnorm_bound} to these coefficient-wise
bounds gives
\begin{equation}
    ||\mathbf{H}_1-\wt{\mathbf{H}}_1||_\infty
    \le O(\poly(|A|)\cdot(\delta+\varepsilon)).
\end{equation}
The argument for $\mathbf{H}_2$ is identical, with the
row and column registers exchanged.
\end{proof}
\begin{corollary}
       \label{cor:efficient_add_a_bridge_appendix}
    Let $B = P_{n-1}(d)$, $A = P_n(d)$. For any basis element $b = d^{({n - 1 - \cc(D)})/2}D \in \mathcal{B}(B)$, the maps 
    \begin{equation}
        \mathbf{G}_1(\wt{\ft_B}\ket{b}) = \wt{\ft_A}\ket{d^{(n - \cc(D))/2} b_{n-1}D}
    \end{equation}
    \begin{equation}
        \mathbf{G}_2(\wt{\ft_B}\ket{b}) = \wt{\ft_A}\ket{d^{(n - \cc(D))/2} Db_{n-1}}
    \end{equation}
    can be implemented to operator norm error $O(\poly(|A|) \cdot (\delta + \varepsilon))$ by operators $\wt{\mathbf{G}}_1$ and $\wt{\mathbf{G}}_2$. Both $\wt{\mathbf{G}}_1$ and $\wt{\mathbf{G}}_2$ can be implemented with $\wt{O}(n^2 \cdot (n + \log d + \log(1/\varepsilon)))$ gates. 
\end{corollary}
\begin{proof}
    Define $\wt{\mathbf{G}}_i \coloneqq \mathbf{\wt{E}}(P_n(d), P_{n - \frac12}(d))\wt{\mathbf{H}}_i$. In other words, simultaneously multiply $b$ by the (normalized) bridge generator $d^{1/2}b_{n-1}$ while embedding the product diagram into $P_{n - \frac12}(d)$, and then embed that diagram into $P_{n}(d)$. This second embedding does not change the normalization, as both $n$ and $\cc(D)$ are unchanged. By~\cref{lem:efficient_add_a_bridge_appendix} and~\cref{lem:embedding_summary}, $\wt{\mathbf{G}}_i$ performs the desired transformation up to $O(\poly(|A|) \cdot (\delta + \varepsilon))$ operator norm error. The asymptotic gate count is determined by taking the dominating terms in the gate counts of $\wt{\mathbf{H}}_i$ and $\mathbf{\widetilde{E}}(P_n(d), P_{n - \frac12}(d))$.   
\end{proof}
\end{appendices}
\end{document}